\documentclass[journal]{new-aiaa}
\usepackage[utf8]{inputenc}

\usepackage{graphicx}
\usepackage{float}

\usepackage{amsmath,bm,bbm,latexsym,subcaption,epstopdf}
\usepackage[english]{babel}
\usepackage[title]{appendix}
\usepackage{hyperref}
\usepackage[pagewise]{lineno}
\newcommand*\patchAmsMathEnvironmentForLineno[1]{%
  \expandafter\let\csname old#1\expandafter\endcsname\csname #1\endcsname
  \expandafter\let\csname oldend#1\expandafter\endcsname\csname end#1\endcsname
  \renewenvironment{#1}{\linenomath\csname old#1\endcsname}{\csname oldend#1\endcsname\endlinenomath}}
\newcommand*\patchBothAmsMathEnvironmentsForLineno[1]{%
  \patchAmsMathEnvironmentForLineno{#1}\patchAmsMathEnvironmentForLineno{#1*}}
\AtBeginDocument{%
  \patchBothAmsMathEnvironmentsForLineno{equation}%
  \patchBothAmsMathEnvironmentsForLineno{align}%
  \patchBothAmsMathEnvironmentsForLineno{flalign}%
  \patchBothAmsMathEnvironmentsForLineno{alignat}%
  \patchBothAmsMathEnvironmentsForLineno{gather}%
  \patchBothAmsMathEnvironmentsForLineno{multline}}

\graphicspath{{./figures/}}
\DeclareGraphicsExtensions{.pdf,.eps}

\newcommand{\dfn}{\triangleq}

\let\abs=\envert

\usepackage{amsthm}

\theoremstyle{definition} 
\newtheorem{exmp}{Example}

\newtheorem{theorem}{Theorem}[section]
\newtheorem{corollary}{Corollary}[theorem]
\newtheorem{lemma}[theorem]{Lemma}
\newtheorem{assumption}{Assumption}

\newtheorem{problem}{Problem}

\newtheorem{remark}{Remark}

\newcommand{\rev}[1]{#1}
\newenvironment{revblock}{\begingroup\color{black}}{\endgroup}

\newcommand{\revthree}[1]{\textcolor{blue}{#1}}
\newenvironment{revthreeblock}{\begingroup\color{blue}}{\endgroup}
\renewcommand{\revthree}[1]{#1}
\renewenvironment{revthreeblock}{\begingroup\color{black}}{\endgroup}

\newcommand{\pX}{(\textbf{x}_{\textbf{p}})}
\newcommand{\pU}{(\textbf{u}_{\textbf{p}})}
\newcommand{\pL}{(\bm{\lambda}_{\textbf{p}})}
\newcommand{\pM}{(\bm{\mu}_{\textbf{p}})}
\newcommand{\pXU}{(\textbf{x}_{\textbf{p}},\textbf{u}_{\textbf{p}})}
\newcommand{\pXULM}{(\textbf{x}_{\textbf{p}},\textbf{u}_{\textbf{p}},\bm{\lambda}_{\textbf{p}},\bm{\mu}_{\textbf{p}})}

\title{Mixed Bernstein-Fourier Approximants for Optimal Trajectory Generation with Periodic Behavior}

\author{Liraz Mudrik\footnote{Postdoctoral fellow, Dept.\ of Mechanical and Aerospace Engineering. Email: \revthree{\texttt{liraz109@gmail.com}}. Member AIAA \textcolor{black}{(Corresponding Author).}}}
\author{Sean Kragelund\footnote{Research Associate Professor, Dept.\ of Mechanical and Aerospace Engineering. Email: \texttt{spkragel@nps.edu}.}}
\author{Isaac Kaminer\footnote{Professor, Dept.\ of Mechanical and Aerospace Engineering. Email: \texttt{kaminer@nps.edu}.
    \newline\newline Partially presented at the 2025 American Control Conference (ACC), Denver, CO, July 08--10, 2025.}}
    
\affil{Naval Postgraduate School, Monterey, CA, 93943}

\begin{document}


\maketitle

\renewcommand{\include}{\input}

\begin{abstract}
Efficient trajectory generation is crucial for autonomous systems; however, current numerical methods often struggle to handle periodic behaviors effectively, particularly when \textcolor{black}{the onboard sensors require equidistant temporal sampling}.
This paper introduces a novel mixed Bernstein-Fourier approximation framework tailored explicitly for optimal motion planning. 
Our proposed methodology leverages the uniform convergence properties of Bernstein polynomials for nonperiodic behaviors while effectively capturing periodic dynamics through the Fourier series. 
Theoretical results are established, including uniform convergence proofs for approximations of functions, derivatives, and integrals, as well as detailed error bound analyses. 
We further introduce a regulated least squares approach for determining approximation coefficients, enhancing numerical stability and practical applicability.
Within an optimal control context, we establish the feasibility and consistency of approximated solutions to their continuous counterparts. 
We also extend the covector mapping theorem, providing theoretical guarantees for approximating dual variables crucial in verifying the necessary optimality conditions from Pontryagin's Maximum Principle.
Numerical examples illustrate the method's superior performance, demonstrating substantial improvements in computational efficiency and precision in scenarios with complex periodic constraints and dynamics. 
Our mixed Bernstein-Fourier methodology thus presents a robust, theoretically grounded, and computationally efficient approach for advanced optimal trajectory planning in autonomous systems.
\end{abstract}

\section{Introduction}
Motion planning is critical for autonomous vehicles tasked with executing complex missions in challenging and dynamic environments. 
Such missions frequently involve navigating around obstacles, efficiently utilizing limited resources, and strictly adhering to operational constraints. 
Optimal control theory provides an efficient framework for tackling these motion planning problems. 
Within this theoretical structure, the optimal trajectory is determined by minimizing a carefully chosen objective function, which typically represents cost, energy, or time, subject to a set of constraints. 
These constraints often include the vehicle’s dynamic equations, initial and terminal conditions, and physical limitations such as bounds on speed, acceleration, and control inputs. 
Due to the complexity inherent in these constraints and objectives, analytical solutions to optimal control problems are extremely rare, necessitating the use of numerical methods.

Direct methods have emerged as the predominant approach for solving optimal control problems numerically. 
These methods transform the infinite-dimensional continuous-time optimal control problem into a finite-dimensional nonlinear programming (NLP) problem by discretizing the state and control variables over time. 
This discretization allows the use of readily available numerical 
solvers to compute approximate solutions efficiently. 
The theoretical underpinnings that ensure feasibility and consistency of these discretized NLP solutions to the original continuous-time problem have been extensively studied and validated~\cite{betts_practical_2010_chap4,polak_optimization_1997_chap4}. 
A wide range of discretization methods have been developed, such as Euler~\cite{polak_optimization_1997_chap4}, Runge-Kutta~\cite{schwartz_consistent_1996}, pseudospectral~\cite{ross_review_2012}, and Bernstein~\cite{cichella_optimal_2021}.
Each method offers distinct advantages and trade-offs in accuracy, computational complexity, and convergence rates.

\textcolor{black}{A fundamental distinction arises regarding the intervals between time nodes. 
Standard pseudospectral methods, which rely on non-uniform node distributions, e.g., Legendre-Gauss-Lobatto (LGL), to mitigate the Runge phenomenon, exhibit an inherent mismatch with hardware constraints in scenarios like optimal observer trajectory generation~\cite{hammel_optimal_1989,oshman_optimization_1999,ponda_trajectory_2009}. 
Since onboard sensors typically operate at a constant rate, the non-uniform solution requires post-processing interpolation to align with the sensor grid. 
This decoupling implies that the optimality and strict constraint satisfaction achieved at the non-uniform time nodes may no longer be guaranteed at the actual sampling instants. 
In contrast, Bernstein polynomials naturally accommodate uniform time intervals without exhibiting the Runge phenomenon. 
By supporting native uniform spacing, they align the optimization grid directly with the physical sensor rate, ensuring consistent performance without the need for lossy resampling. 
However, Bernstein polynomials typically exhibit slower convergence compared to pseudospectral methods~\cite{ross_review_2012}.}

Despite their broad adoption and effectiveness, direct methods primarily focus on accurately approximating primal variables—the system states and control inputs—while often neglecting the dual variables (costates and multipliers) that appear in Pontryagin’s Maximum Principle (PMP)~\cite{kirk_optimal_2004_chap5}.
The PMP provides the necessary conditions for optimality, encompassing both primal and dual variables. 
The convergence of primal variables alone does not imply the convergence of dual variables, and thus, verifying optimality demands careful consideration of costate convergence. 
The \textit{covector mapping theorem} bridges this critical gap by explicitly linking solutions from discretized NLP problems back to the PMP conditions. 
This theorem ensures that the dual variables derived from numerical solutions accurately approximate their continuous-time counterparts. 
Prior work has successfully established this theorem for various discretization methods, including Runge-Kutta~\cite{hager_runge-kutta_2000}, pseudospectral~\cite{gong_connections_2008,garg_unified_2010}, and Bernstein~\cite{cichella_consistency_2022}.

\textcolor{black}{In addition to native uniform discretization,} Bernstein polynomials possess several \textcolor{black}{other} beneficial properties for trajectory generation, notably their uniform convergence to the approximated functions and their derivatives. 
Additionally, Bernstein polynomials feature favorable geometric characteristics, simplifying the computation of critical trajectory constraints like obstacle avoidance and trajectory bounds~\cite{kielas-jensen_bernstein_2022,strano_kinodynamic_2023}.
\textcolor{black}{These properties make them an attractive candidate for robust motion planning, provided their convergence rate can be improved.} 

The mixed Bernstein-Fourier approximation approach proposed in this paper is particularly beneficial for motion planning problems exhibiting periodic behaviors, such as optimal observer trajectory generation~\cite{hammel_optimal_1989,oshman_optimization_1999,ponda_trajectory_2009}. \textcolor{black}{
}
The primary contributions of this work are as follows. 
First, we establish the uniform convergence properties for the mixed Bernstein-Fourier basis, including convergence of the approximated functions, their derivatives, integrals, and explicit error bounds. 
Next, we introduce a regulated least squares (LS) formulation designed to robustly handle cases that may exhibit clear periodicity that is not straightforward to identify, thereby broadening the applicability of our method. 
We then demonstrate how this mixed approximation method can effectively address optimal motion planning problems, providing accurate and efficient solutions. 
Specifically, we provide detailed proofs of feasibility and consistency, showing that the solutions obtained from the approximated problems converge to the true solutions of the original continuous-time optimal control problems. 
These proofs offer theoretical guarantees on the reliability and accuracy of our approximation approach. 
Finally, we generalize the covector mapping theorem to accommodate the mixed Bernstein-Fourier context, ensuring precise and reliable approximations of dual variables.
Preliminary results using this mixed approximation method were initially presented in an earlier conference paper~\cite{mudrik_optimal_2025}. 
This work presents numerical simulations that validate the theoretical results and illustrate the practical advantages of our proposed methodology.

This paper is organized as follows.
The next section presents the necessary mathematical background of Bernstein polynomials and the Fourier series.
Sec.~\ref{Sec:prop_sol} presents the mixed Bernstein-Fourier approximants and proves the uniform convergence of this approximation, its derivative, and its integral.
Sec.~\ref{sec:primal} formulates the motion planning problem using optimal control theory and the approximated version of this problem using the mixed Bernstein-Fourier basis functions, including an analysis that proves the feasibility and consistency of the approximated problem.
After proving the covector mapping theorem in Sec.~\ref{sec:dual}, 
we present simulation results in Sec.~\ref{sec:sim} to illustrate the performance of the proposed solution in \revthree{three} examples, followed by concluding remarks.

\section{Mathematical Background}
\label{sec:math}
This section describes the required mathematical background on using Bernstein and Fourier approximants for function approximation.
We first present the notation used in this paper. 
Bold letters denote vectors: $\textbf{x} = [x_{1},\dots,x_{m}]^{T}$.
The symbol $\mathcal{C}^{r}$ denotes the space of functions with $r$ continuous derivatives, and $\mathcal{C}_{m}^{r}$ denotes the space of $m$-dimensional vector-valued functions in $\mathcal{C}^{r}$.
\textcolor{black}{
$\Vert \cdot \Vert : \mathbb{R}^{m} \rightarrow \mathbb{R}$ denotes the standard Euclidean norm of a vector, i.e., $\Vert \textbf{x} \Vert = \sqrt{x_{1}^{2}+\cdots+x_{m}^{2}}$.
For functions $\textbf{f}(t)$, bounds of the form $\Vert \textbf{f}(t) \Vert \le \epsilon$ are understood to hold uniformly for all $t \in [0, t_f]$.}
The modulus of continuity is defined for any \rev{continuous} function $\textbf{f} : [0,t_{f}] \rightarrow \mathbb{R}^{m_{\textbf{f}}}$ as
\begin{equation}
	W_{\textbf{f}}(\delta) = \max_{\substack{t_{1},t_{2}\in[0,t_{f}]\\ \abs{t_{1} - t_{2}}\leq \delta}}  \Vert \textbf{f}(t_{1}) - \textbf{f}(t_{2}) \Vert
\end{equation}
where $m_{\textbf{f}}$ is the dimension of the function vector $\textbf{f}(\cdot)$.
If $\textbf{f} \in \mathcal{C}_{m_{\textbf{f}}}^{0}$ on $ [0,t_{f}]$, then $\lim\limits_{\delta \rightarrow 0} W_{\textbf{f}}(\delta) = 0$.
\textcolor{black}{Note that the use of max is justified because continuous functions on the compact interval $[0,t_f]$ attain their extrema.}
\subsection{Bernstein Approximation}
The Bernstein basis polynomials of order $n$ are defined as
\begin{equation}
	b_{k,n}(t) = \binom{n}{k} \frac{t^{k} (t_{f}-t)^{n-k}}{t_{f}^{n}}, \quad \forall t \in [0,t_{f}]
\end{equation}
where 
\begin{equation}
	\binom{n}{k} = \frac{n!}{k! (n-k)!}.
\end{equation}
We denote equidistant time nodes $t_{k} = \frac{k}{n}t_{f}, k=0,\dots,n$, where $t_{f} > 0$ is given.
We then use these basis functions to introduce the $n^{\text{th}}$ approximation of a function
\begin{equation}
	\textbf{f}(t) \approx \textbf{f}_{n}(t) = \sum_{k=0}^{n} \bar{\textbf{f}}_{k,n} b_{k,n}(t) , \quad t \in [0,t_{f}]
	\label{eq:br_app}
\end{equation}
where $\bar{\textbf{f}}_{k,n} \in \mathbb{R}^{m_{\textbf{f}}}, k=0,\dots,n$, are the Bernstein coefficients, satisfying $\bar{\textbf{f}}_{k,n} = \textbf{f}(t_{k})$ for all $k=0,\dots,n$.
Note that the first and last coefficients of a Bernstein polynomial are its boundary points, that is,
$\textbf{f}(0) = \bar{\textbf{f}}_{0,n}$ and $\textbf{f}(t_{f}) = \bar{\textbf{f}}_{n,n}$.
We can also approximate the time derivative of the function using the following differentiation property
\begin{align}
	\dot{\textbf{f}}_{n}(t) = \sum_{k=0}^{n-1} \frac{n}{t_{f}} (\bar{\textbf{f}}_{k+1,n} - \bar{\textbf{f}}_{k,n}) b_{k,n-1}(t)  =\frac{1}{t_{f}} \sum_{k=0}^{n-1} \sum_{i=0}^{n} \bar{\textbf{f}}_{i,n} \textbf{D}_{i,k}  b_{k,n-1}(t) 
	\label{eq:dyn_ber}
\end{align}
where $\textbf{D}_{i,k} $ is the $(i,k)$ element of the differentiation matrix, satisfying
\begin{equation}
	\textbf{D} =  \begin{bmatrix}
		-n     & 0      & \cdots & 0  \\
		n      & -n     & \cdots & 0  \\
		\vdots & \ddots & \ddots &    \\
		0      & \cdots & n      & -n \\
		0      & \cdots & 0      & n
	\end{bmatrix} \in \mathbb{R}^{(n+1)\times n}.
\end{equation}

\begin{lemma}[Uniform Convergence, see~\rev{\cite{cichella_optimal_2021}}]
	\label{lem:B_con_mod}
	Let $\textup{\textbf{f}} \in \mathcal{C}^{0}_{m_{\textup{\textbf{f}}}}$ on $[0,t_{f}]$. Its Bernstein approximation, defined in Eq.~\eqref{eq:br_app}, satisfies for any $n \in \mathbb{N}$ \textcolor{black}{and all $t \in [0, t_f]$:}
	\begin{equation}
		\Vert \textup{\textbf{f}}(t) - \textup{\textbf{f}}_{n}(t) \Vert \leq C_{0} W_{\textup{\textbf{f}}}\Big(\frac{1}{\sqrt{n}}\Big)
	\end{equation}
	where $C_{0}$ is a positive constant independent of $n$ and $W_{\textup{\textbf{f}}}(\cdot)$ is the modulus of continuity of $\textup{\textbf{f}}$ in $[0,t_{f}]$.
	
	If $\textup{\textbf{f}} \in \mathcal{C}^{1}_{m_{\textup{\textbf{f}}}}$ on $[0,t_{f}]$, then
	\begin{equation}
		\Vert \dot{\textup{\textbf{f}}}(t) - \dot{\textup{\textbf{f}}}_{n}(t) \Vert \leq C_{1} W_{\dot{\textup{\textbf{f}}}}\Big(\frac{1}{\sqrt{n}}\Big)
	\end{equation}
	where $C_{1}$ is a positive constant independent of $n$ and $W_{\dot{\textup{\textbf{f}}}}(\cdot)$ is the modulus of continuity of $\dot{\textup{\textbf{f}}}$ in $[0,t_{f}]$.
\end{lemma}

Integrating the approximated state vector is readily computed as
\begin{equation}
	\int_{0}^{t_{f}} \textbf{f}_{n}(t)dt = \frac{t_{f}}{n+1} \sum_{k=0}^{n} \bar{\textbf{f}}_{k,n}.
	\label{eq:int_ber}
\end{equation}
\begin{lemma}[Uniform Convergence of the Integral, see~\cite{cichella_optimal_2021}]
	\label{lem:B_int}
	If $\textup{\textbf{f}} \in \mathcal{C}^{0}_{m_{\textup{\textbf{f}}}}$ on $[0,t_{f}]$, then
	\begin{equation}
		\left\| \int_{0}^{t_{f}} \textup{\textbf{f}}(t)dt - \frac{t_{f}}{n+1} \sum_{k=0}^{n} \bar{\textup{\textbf{f}}}_{k,n} \right\| \leq C_{I} W_{\textup{\textbf{f}}}\Big(\frac{1}{\sqrt{n}}\Big)
	\end{equation}
	where $C_{I}$ is a positive constant independent of $n$.
\end{lemma}
\subsection{Fourier Series}
A function $\textbf{p}(t)$ is $t_{f}$-periodic if $\textbf{p}(t) = \textbf{p}(t+t_{f})$ for all $t\in \mathbb{R}$; on the interval $[0,t_{f}]$ studied here, periodicity reduces to the endpoint condition $\textbf{p}(0) = \textbf{p}(t_{f})$.
\rev{For $r \in \mathbb{N}$, we denote by $\mathcal{P}^{r}_{m}$ the class of functions $\textbf{p} \in \mathcal{C}^{r}_{m}$ on $[0,t_{f}]$ whose derivatives match at the endpoints, i.e., $\textbf{p}^{(j)}(0) = \textbf{p}^{(j)}(t_{f})$ for all $j = 0,\dots,r$; equivalently, the $t_{f}$-periodic extension of $\textbf{p}$ possesses $r$ continuous derivatives on $\mathbb{R}$. In particular, every constant function belongs to $\mathcal{P}^{r}_{m}$ for all $r$.}
The Fourier series of a periodic function $\textbf{p}(t) \in \mathcal{C}^{r}_{m_{\textbf{p}}}$ on $[0,t_{f}]$ is constructed as
\begin{equation}
	\textbf{S}_{n}(t) = \frac{\textbf{a}_{0}}{2} + \sum_{k=1}^{n} \left[ \textbf{a}_{k} \cos \left( \frac{2 \pi k t}{t_{f}}\right) + \textbf{b}_{k} \sin \left( \frac{2 \pi k t}{t_{f}}\right) \right]
	\label{eq:for_app}
\end{equation}
whose coefficients satisfy:
\begin{align}
	\textbf{a}_{k} & = \frac{2}{t_{f}} \int_{0}^{t_{f}} \textbf{p}(t) \cos\left( \frac{2 \pi k t}{t_{f}}\right) dt, \quad \forall k=0,\dots,n, \\
	\textbf{b}_{k} & = \frac{2}{t_{f}} \int_{0}^{t_{f}} \textbf{p}(t) \sin\left( \frac{2 \pi k t}{t_{f}}\right) dt, \quad \forall k=1,\dots,n.
\end{align}
\textcolor{black}{Note that $\textbf{b}_{0}$ is zero for sine terms, so we index $\textbf{b}_{k}$ starting from $k=1$ while $\textbf{a}_{k}$ includes the DC component $k=0$.}

\begin{lemma}(Uniform Convergence of Periodic Function, see \rev{Chap.~I, Sec.~3, Thm.~IX and Cor.~IV} in~\cite{jackson_theory_1930_chap1})
	\label{lem:F_con_mod}
	If $\textup{\textbf{p}}(t)\in \rev{\mathcal{P}^{r}_{m_{\textup{\textbf{p}}}}}$, where $r \in \mathbb{N}$, then its Fourier approximation, defined in Eq.~\eqref{eq:for_app}, satisfies for any \rev{integer $n \geq 2$}
	\begin{equation}
		\Vert \textup{\textbf{p}}(t) - \textup{\textbf{S}}_{n}(t) \Vert \leq \frac{A_{0}\log n}{n^{r}} W_{\textup{\textbf{p}}^{(r)}}\Big(\rev{\frac{t_{f}}{n}}\Big)
	\end{equation}
	where $A_{0}$ is a positive constant independent of $n$ and \textcolor{black}{$W_{\textup{\textbf{p}}^{(r)}}(\cdot)$ is the modulus of continuity of the $r$-th derivative $\textup{\textbf{p}}^{(r)}$} in $[0,t_{f}]$.

    \rev{Moreover,}
    \begin{equation}
		\Vert \dot{\textup{\textbf{p}}}(t) - \dot{\textup{\textbf{S}}}_{n}(t) \Vert \leq \frac{A_{1}\log n}{n^{r-1}} W_{\textup{\textbf{p}}^{(r)}}\Big(\rev{\frac{t_{f}}{n}}\Big)
	\end{equation}
	where $A_{1}$ is a positive constant independent of $n$\rev{; the constants $A_{0}$ and $A_{1}$ depend only on $r$, $t_{f}$, and $m_{\textup{\textbf{p}}}$}.
\end{lemma}
\begin{remark}
	\label{rem:jackson_scaling}
	\rev{Lemma~\ref{lem:F_con_mod} is the period-$2\pi$ estimate of~\cite{jackson_theory_1930_chap1} transported to $[0,t_{f}]$ by the rescaling $s = 2\pi t/t_{f}$, which maps $[0,t_{f}]$ onto $[0,2\pi]$. The rescaled function $\tilde{\textup{\textbf{p}}}(s) \dfn \textup{\textbf{p}}(t_{f}s/2\pi)$, equivalently $\tilde{\textup{\textbf{p}}}(2\pi t/t_{f}) = \textup{\textbf{p}}(t)$, is $2\pi$-periodic and has the same Fourier coefficients as $\textup{\textbf{p}}$. Its moduli satisfy $W_{\tilde{\textup{\textbf{p}}}^{(r)}}(2\pi/n) = (t_{f}/2\pi)^{r}\,W_{\textup{\textbf{p}}^{(r)}}(t_{f}/n)$, the $t_{f}$-dependent factors being absorbed into $A_{0}$ and $A_{1}$. The endpoint matching in $\mathcal{P}^{r}_{m_{\textup{\textbf{p}}}}$ and the restriction $n \geq 2$ are the conditions under which that estimate holds.}
\end{remark}
\begin{remark}
    The integral of periodic functions over their period, $\textup{\textbf{P}}(t) \dfn \int_{t}^{t+t_{f}}\textup{\textbf{p}}(t')dt'$, are constant for any $t\in \mathbb{R}$:
	\begin{equation}
		\frac{d\textup{\textbf{P}}(t)}{dt} = \textup{\textbf{p}}(t+t_{f}) - \textup{\textbf{p}}(t) = \textup{\textbf{0}}.
	\end{equation}
	Therefore, the integral of the Fourier series is an exact approximation of $\textup{\textbf{P}}(0)$:
	\begin{equation}
		\label{eq:F_int}
		\int_{0}^{t_{f}} \textup{\textbf{S}}_{n}(t)dt = \frac{\textup{\textbf{a}}_{0}t_{f}}{2} = \frac{t_{f}}{2}\frac{2}{t_{f}} \int_{0}^{t_{f}} \textup{\textbf{p}}(t) dt = \textup{\textbf{P}}(0).
	\end{equation}
\end{remark}
%

\section{Mixed Bernstein-Fourier Basis for Approximation Theory}
\label{Sec:prop_sol}
This section shows how the mixed Bernstein-Fourier basis functions can be used to approximate continuous functions.
\subsection{Foundations of the Mixed Bernstein-Fourier Basis}
We first note that for any periodic function, $\textbf{p}(t)$, we can use the following relation, which trivially holds for any $t\in[0,t_{f}]$: 
\begin{equation}
	\textbf{f}(t) = (\textbf{f}(t) - \textbf{p}(t)) + \textbf{p}(t) = \textbf{g}(t) + \textbf{p}(t)
	\label{eq:per}
\end{equation}
where $\textbf{g}(t) \dfn \textbf{f}(t) - \textbf{p}(t)$. 
This also implies that $m_{\textbf{f}} = m_{\textbf{p}} = m_{\textbf{g}}$.
\begin{remark}
    When the function $\textup{\textbf{f}}(t)$ does not exhibit periodic behavior, e.g., it is a straight line, the periodic component can be chosen as \rev{any constant value function, including} the zero function $\textup{\textbf{p}}(t)\equiv \textup{\textbf{0}}$.
    In this case, the mixed Bernstein-Fourier approximation reduces to the classical Bernstein-only approximation.
\end{remark}
Then, we can present the mixed approximation using both Bernstein polynomials and Fourier series
\begin{align}
	\textbf{T}_{n}(t) = \sum_{k=0}^{n_{\textbf{g}}} \bar{\textbf{g}}_{k} b_{k,n_{\textbf{g}}}(t) + \frac{\textbf{a}_{0}}{2} + \sum_{k=1}^{n_{\textbf{p}}} \left[ \textbf{a}_{k} \cos \left( \frac{2 \pi k t}{t_{f}}\right) + \textbf{b}_{k} \sin \left( \frac{2 \pi k t}{t_{f}}\right) \right]
	\label{eq:FB_app}
\end{align}
where $n_{\textbf{g}}\in \mathbb{N}$ and $n_{\textbf{p}}\in \mathbb{N}$ are the orders of approximation of the nonperiodic and periodic parts, respectively, $n \dfn \min(n_{\textbf{g}},n_{\textbf{p}})$, and the coefficients satisfy
\begin{align}
	\label{eq:g_k}
	\bar{\textbf{g}}_{k} & = \textbf{f}\Big(\frac{k}{n_{\textbf{g}}}t_{f}\Big) - \textbf{p}\Big(\frac{k}{n_{\textbf{g}}}t_{f}\Big), \quad \forall k=0,\dots,n_{\textbf{g}},\\
	\label{eq:a_k}
	\textbf{a}_{k} & = \frac{2}{t_{f}} \int_{0}^{t_{f}} \textbf{p}(t) \cos\left( \frac{2 \pi k t}{t_{f}}\right) dt, \quad \forall k=0,\dots,n_{\textbf{p}}, \\
	\label{eq:b_k}
	\textbf{b}_{k} & = \frac{2}{t_{f}} \int_{0}^{t_{f}} \textbf{p}(t) \sin\left( \frac{2 \pi k t}{t_{f}}\right) dt, \quad \forall k=1,\dots,n_{\textbf{p}}.
\end{align}
\begin{remark}
	At the end points, the mixed approximation satisfies
	\begin{align}
		\textup{\textbf{T}}_{n}(0) & = \bar{\textup{\textbf{g}}}_{0} + \frac{\textup{\textbf{a}}_{0}}{2} + \sum_{k=1}^{n_{\textup{\textbf{p}}}} \textup{\textbf{a}}_{k}, \\
		\textup{\textbf{T}}_{n}(t_{f}) & = \bar{\textup{\textbf{g}}}_{n_{\textbf{g}}} + \frac{\textup{\textbf{a}}_{0}}{2} + \sum_{k=1}^{n_{\textup{\textbf{p}}}} \textup{\textbf{a}}_{k}.
	\end{align}
\end{remark}
\subsection{Mathematical Properties and Convergence}
\rev{Given a decomposition~\eqref{eq:per} with $\textbf{p} \in \mathcal{P}^{r}_{m_{\textbf{f}}}$, $r \in \mathbb{N}$, we define} the approximation error upper bound of the approximated function $\textbf{f}$ and its derivative to be:
\begin{align}
	\delta_{\textbf{f}}^{n}\revthree{(\textbf{p})} &\dfn C_{0}W_{\textbf{g}}\Big(\frac{1}{\sqrt{n_{\textbf{g}}}}\Big) +  \frac{A_{0}\log n_{\textbf{p}}}{n_{\textbf{p}}^{r}} W_{\textbf{p}^{(r)}}\Big(\rev{\frac{t_{f}}{n_{\textbf{p}}}}\Big), \label{eq:delta_f}\\
	\delta_{\dot{\textbf{f}}}^{n}\revthree{(\textbf{p})} &\dfn C_{1}W_{\dot{\textbf{g}}}\Big(\frac{1}{\sqrt{n_{\textbf{g}}}}\Big) 
	+  \frac{A_{1}\log n_{\textbf{p}}}{n_{\textbf{p}}^{r-1}} W_{{\textbf{p}^{(r)}}}\Big(\rev{\frac{t_{f}}{n_{\textbf{p}}}}\Big). \label{eq:delta_fd}
\end{align}
\revthree{The argument $\textbf{p}$ makes explicit that these bounds depend on the chosen decomposition~\eqref{eq:per}, through $\textbf{g} = \textbf{f} - \textbf{p}$ and $\textbf{p}^{(r)}$.}
\begin{theorem}
	\label{th:FB_conv1}
	Let $\textup{\textbf{f}}(t) \in \mathcal{C}_{m_{\textup{\textbf{f}}}}^{0}$ on $[0,t_{f}]$, \rev{let $\textup{\textbf{p}} \in \mathcal{P}^{r}_{m_{\textup{\textbf{f}}}}$, $r \in \mathbb{N}$, with $\textup{\textbf{g}} = \textup{\textbf{f}} - \textup{\textbf{p}}$ as defined in Eq.~\eqref{eq:per},} and \rev{let} its approximation be calculated according to Eq\rev{s}.~\eqref{eq:FB_app}\rev{--\eqref{eq:b_k}}. 
	Then\rev{,}
	\begin{equation}
		\Vert \textup{\textbf{f}}(t) - \textup{\textbf{T}}_{n}(t) \Vert \leq \delta_{\textup{\textbf{f}}}^{n}\revthree{(\textup{\textbf{p}})}
	\end{equation}
 	for any $n_{\textup{\textbf{g}}} \in \mathbb{N}$ \rev{and any integer $n_{\textup{\textbf{p}}} \geq 2$}.
\end{theorem}
\begin{proof}
	Consider:
	\begin{align}
		 \Vert \textbf{f}(t) &- \textbf{T}_{n}(t) \Vert = 
		  \Vert \textbf{g}(t) + \textbf{p}(t) - \textbf{T}_{n}(t)\Vert \notag \\
		& = \Vert \textbf{g}(t) - \sum_{k=0}^{n_{\textbf{g}}} \bar{\textbf{g}}_{k} b_{k}(t) + \textbf{p}(t) - \frac{\textbf{a}_{0}}{2} - \sum_{k=1}^{n_{\textbf{p}}} \left[ \textbf{a}_{k} \cos \Big( \frac{2 \pi k t}{t_{f}}\Big) + \textbf{b}_{k} \sin \left( \frac{2 \pi k t}{t_{f}}\right) \right] \Vert \notag \\
		& \leq \Vert \textbf{g}(t) - \sum_{k=0}^{n_{\textbf{g}}} \bar{\textbf{g}}_{k} b_{k}(t) \Vert + \Vert\textbf{p}(t) - \frac{\textbf{a}_{0}}{2}  - \sum_{k=1}^{n_{\textbf{p}}} \left[ \textbf{a}_{k} \cos \left( \frac{2 \pi k t}{t_{f}}\right) + \textbf{b}_{k} \sin \left( \frac{2 \pi k t}{t_{f}}\right) \right] \Vert
		\label{eq:FB_proof}
	\end{align}
	where the first equality holds due to the relation in Eqs.~\eqref{eq:per} and~\eqref{eq:FB_app} and the inequality holds by the triangle inequality.
	Using Lemma~\ref{lem:B_con_mod}, we get that
	\begin{equation}
		\Vert \textbf{g}(t) - \sum_{k=0}^{n_{\textbf{g}}} \bar{\textbf{g}}_{k} b_{k}(t) \Vert \leq C_{0}W_{\textbf{g}}\Big(\frac{1}{\sqrt{n_{\textbf{g}}}}\Big)
	\end{equation}
	where $\bar{\textbf{g}}_{k}$ are calculated according to~\eqref{eq:g_k}.
	Similarly, using Lemma~\ref{lem:F_con_mod}, we get that
	\begin{align}
		\Vert\textbf{p}(t) - \frac{\textbf{a}_{0}}{2}  - 
		\sum_{k=1}^{n_{\textbf{p}}} \left[ \textbf{a}_{k} \cos \left( \frac{2 \pi k t}{t_{f}}\right) + \textbf{b}_{k} \sin \left( \frac{2 \pi k t}{t_{f}}\right) \right] \Vert  \leq \frac{A_{0}\log n_{\textbf{p}}}{n_{\textbf{p}}^{r}} W_{\textbf{p}^{(r)}}\Big(\rev{\frac{t_{f}}{n_{\textbf{p}}}}\Big)
	\end{align}
	where $\textbf{a}_{k}$ and $\textbf{b}_{k}$ are calculated according to~\eqref{eq:a_k} and~\eqref{eq:b_k}, respectively, and the result follows.
\end{proof}
The following example shows the advantages of using the mixed approximation when there is a known periodic component.
\begin{exmp}
    Consider $f(t) = t + \sin{(2\pi t - 0.5)}$ on $[0,1]$.
    In this case, we readily identify the periodic part as $p(t) = \sin{(2\pi t - 0.5)}$ and the non-periodic part as $g(t) = t$.
    We approximate $f(t)$ with a mixture of first order Bernstein and Fourier approximants,
    \begin{equation}
        T_{1}(t) = g_{1} t + a_{1} \cos{(2\pi t)} + b_{1} \sin{(2\pi t)}
    \end{equation}
    where $g_{1} = 1, a_{1} = - \sin{0.5}$, and $b_{1} = \cos{0.5}$. Note that this approximation has a total of five coefficients, but $g_{0} = a_{0} = 0$.
    This approximation is exact, i.e., $T_{1}(t) = f(t)$ for any $t\in[0,1]$.

    Fig.~\ref{fig:aprox} illustrates the benefit of using mixed Bernstein-Fourier basis functions to approximate $f(t)$, rather than using Bernstein polynomials alone.
    \textcolor{black}{The figure compares the mixed approximation (dashed red line) against Bernstein approximations of $5^{\text{th}}$ (magenta), $20^{\text{th}}$ (green), and $100^{\text{th}}$ (black) orders for the function $f(t)=t + \sin{(2\pi t - 0.5)}$ (black line).}
    Adding the Fourier series to the approximant exploits the periodic behavior of the sine function, greatly improving the quality of the approximation. 
    Note that using only a Fourier series approximant in this case does not uniformly converge because $f(t)$ is not a periodic function.
    \begin{figure}[htb]
    	\centering
    	\includegraphics[width=3.25in]{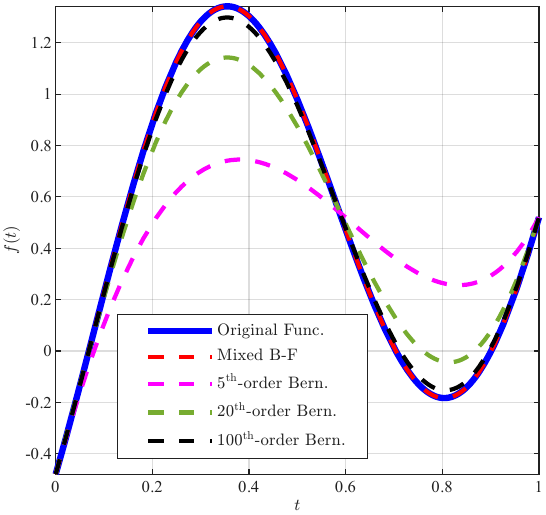}
    	\caption{\textcolor{black}{Approximation of a periodic function using mixed Bernstein-Fourier (B-F) versus Bernstein-only approximants.}}
    	\label{fig:aprox}
    \end{figure}
\end{exmp}

The derivative of the mixed approximation is:
\begin{align}
	\dot{\textbf{T}}_{n}(t)  =\frac{1}{t_{f}}\sum_{j=0}^{n_{\textbf{g}}-1} \sum_{k=0}^{n_{\textbf{g}}} \bar{\textbf{g}}_{k} \textbf{D}_{k,j}  b_{j,n-1}(t) + \frac{2 \pi}{t_{f}} \sum_{k=1}^{n_{\textbf{p}}}  k \left[\textbf{b}_{k} \cos \left( \frac{2 \pi k t}{t_{f}}\right) -  \textbf{a}_{k} \sin \left( \frac{2 \pi k t}{t_{f}}\right) \right].
	\label{eq:dFB_app}
\end{align}
\begin{corollary}
	\label{cor:diff}
	\rev{Under the hypotheses of Thm.~\ref{th:FB_conv1}, assume additionally that} $\textup{\textbf{f}}(t) \in \mathcal{C}_{m_{\textup{\textbf{f}}}}^{1}$ on $[0,t_{f}]$\rev{, so that $\textup{\textbf{g}} \in \mathcal{C}_{m_{\textup{\textbf{f}}}}^{1}$, since $\mathcal{P}^{r}_{m_{\textup{\textbf{f}}}} \subset \mathcal{C}^{1}_{m_{\textup{\textbf{f}}}}$ for $r \geq 1$.} Then, for any $n_{\textup{\textbf{g}}} \in \mathbb{N}$ \rev{and any integer $n_{\textup{\textbf{p}}} \geq 2$}
	\begin{align}
		\Vert \dot{\textup{\textbf{f}}}(t) - \dot{\textup{\textbf{T}}}_{n}(t)  \Vert \leq \delta_{\dot{\textup{\textbf{f}}}}^{n}\revthree{(\textup{\textbf{p}})}.
	\end{align}
\end{corollary}


We also need to approximate the integral of functions; we use $n_{\textbf{g}}+1$ time steps as done in~\eqref{eq:int_ber} for the Bernstein-based approximation, yielding the following relation
\begin{align}
	\label{eq:BF_int}
	\int_{0}^{t_{f}} \textbf{T}_{n}(t)dt = \frac{t_{f}}{n_{\textbf{g}}+1} \sum_{k=0}^{n_{\textbf{g}}} \bar{\textbf{g}}_{k} + \frac{\textbf{a}_{0}t_{f}}{2} = w \sum_{k=0}^{n_{\textbf{g}}} \left(\bar{\textbf{g}}_{k} + \frac{\textbf{a}_{0}}{2}\right)
\end{align}
where $w = \frac{t_{f}}{n_{\textbf{g}}+1}$.

\begin{theorem}
	\label{lem:int}
	If $\textup{\textbf{f}}(t) \in \mathcal{C}^{0}_{m_{\textup{\textbf{f}}}}$ on $[0,t_{f}]$ \rev{and $\textup{\textbf{p}} \in \mathcal{C}^{0}_{m_{\textup{\textbf{f}}}}$ in Eq.~\eqref{eq:per}, so that $\textup{\textbf{g}} = \textup{\textbf{f}} - \textup{\textbf{p}} \in \mathcal{C}^{0}_{m_{\textup{\textbf{f}}}}$}, then
	\begin{align}
		\left\Vert \int_{0}^{t_{f}} \textup{\textbf{f}}(t)dt - w \sum_{k=0}^{n_{\textup{\textbf{g}}}} \left(\bar{\textup{\textbf{g}}}_{k} + \frac{\textup{\textbf{a}}_{0}}{2}\right) \right\Vert \leq C_{I} \rev{W_{\textup{\textbf{g}}}\Big(\frac{1}{\sqrt{n_{\textup{\textbf{g}}}}}\Big)}.
	\end{align}
\end{theorem}
\begin{proof}
	Using the relations from~\eqref{eq:per} and~\eqref{eq:BF_int} and applying the triangle inequality, we get
	\begin{align}
		&\left\Vert \int_{0}^{t_{f}} \textbf{f}(t)dt - w \sum_{k=0}^{n_{\textbf{g}}} \left(\bar{\textbf{g}}_{k} + \frac{\textbf{a}_{0}}{2}\right) \right\Vert \leq
		\left\Vert \int_{0}^{t_{f}} \textbf{g}(t)dt - \frac{t_{f}}{n_{\textbf{g}}+1} \sum_{k=0}^{n_{\textbf{g}}} \bar{\textbf{g}}_{k} \right\Vert + \left\Vert \int_{0}^{t_{f}} \textbf{p}(t)dt - \frac{\textbf{a}_{0}t_{f}}{2} \right\Vert
	\end{align}
	and the result follows by applying Lemma~\ref{lem:B_int} \rev{to $\textup{\textbf{g}}$, whose coefficients~\eqref{eq:g_k} are its nodal samples,} and Eq.~\eqref{eq:F_int}\rev{, which renders the second term identically zero}.
\end{proof}
\rev{Note that the bound of Thm.~\ref{lem:int} is independent of $n_{\textup{\textbf{p}}}$: by Eq.~\eqref{eq:F_int}, the Fourier block of~\eqref{eq:FB_app} contributes zero quadrature error at any order.}
\subsection{Approximation Using Regulated Least Squares}
\label{sec:LS}
Whereas some continuous functions contain obviously periodic components, others may have hidden periodic behavior.
Thus, we can formulate a least squares (LS) problem to optimally find the coefficients of a mixed basis that approximates the given function or data points. We now illustrate this method for the case of a scalar function, but generalization to vector-valued functions is trivial.
We use equidistant time nodes $t_{k} = \frac{k}{n_{t}} t_{f}$ for all $k = 0,\dots,n_{t}$, where $1 / n_{t}$ is the spacing in time.
Consider a given function or data using these time nodes:
\begin{equation}
	\textup{\textbf{y}} = 
	\begin{bmatrix}
		y_{0} & \cdots & y_{n_{t}}
	\end{bmatrix}^{T} \in \mathbb{R}^{n_{t}+1}.
\end{equation}
Define the Bernstein-Fourier basis matrix:
\begin{equation}
	B = 
	\begin{bmatrix}
		b_{0,n_{\textbf{g}}}(t_{0}) & \cdots & b_{n_{\textbf{g}},n_{\textbf{g}}}(t_0) & \frac{1}{2} & c_{1}(t_{0}) & \cdots & c_{n_{\textbf{p}}}(t_{0}) & s_{1}(t_{0}) & \cdots & s_{n_{\textbf{p}}}(t_{0}) \\
		b_{0,n_{\textbf{g}}}(t_{1}) & \cdots & b_{n_{\textbf{g}},n_{\textbf{g}}}(t_1) & \frac{1}{2} & c_{1}(t_{1}) & \cdots & c_{n_{\textbf{p}}}(t_{1}) & s_{1}(t_{1}) & \cdots & s_{n_{\textbf{p}}}(t_{1}) \\
		\vdots & & \vdots & \vdots & \vdots & & \vdots & \vdots & & \vdots\\
		b_{0,n_{\textbf{g}}}(t_{n_{t}}) & \cdots & b_{n_{\textbf{g}},n_{\textbf{g}}}(t_{n_{t}}) & \frac{1}{2} & c_{1}(t_{n_{t}}) & \cdots & c_{n_{\textbf{p}}}(t_{n_{t}}) & s_{1}(t_{n_{t}}) & \cdots & s_{n_{\textbf{p}}}(t_{n_{t}})
	\end{bmatrix}
\end{equation}
where $B \in \mathbb{R}^{(n_{t}+1)\times(n_{\text{tot}})}$,  $n_{\text{tot}} = n_{\textbf{g}}+2n_{\textbf{p}}+2$, $c_{j}(t_{k}) \dfn \cos \left( \frac{2 \pi j t_{k}}{t_{f}}\right)$ and $s_{j}(t_{k}) \dfn \sin \left( \frac{2 \pi j t_{k}}{t_{f}}\right)$, and the decision vector:
\begin{equation}
	\textbf{d} = 
	\begin{bmatrix}
		\bar{{g}}_{0} & \dots & \bar{{g}}_{n_{\textbf{g}}} & {a}_{0} & \dots & {a}_{n_{\textbf{p}}} & {b}_{1} & \dots & {b}_{n_{\textbf{p}}}
	\end{bmatrix}^{T}  \in \mathbb{R}^{n_{\text{tot}}}.
\end{equation}

Thus, we aim to solve the following LS problem~\cite{boyd_introduction_2018_chap15}:
\begin{equation}
	\color{black}{\min_{\textbf{d}} \Vert B\textbf{d}-\textup{\textbf{y}} \Vert^2}
\end{equation}
which\rev{, for full-column-rank $B$,} has the known solution: 
\begin{equation}
	\textbf{d}^{*} = (B^{T}B)^{-1} B^{T} \textup{\textbf{y}}.
	\label{eq:LS_sol}
\end{equation}

\rev{The mixed basis matrix $B$ is rank-deficient by at least one at every order: by the partition-of-unity identity $\sum_{k=0}^{n_{\textbf{g}}} b_{k,n_{\textbf{g}}}(t) = 1$, the Bernstein columns sum to twice the constant column associated with $a_{0}$, so $B^{T}B$ is singular and the closed form~\eqref{eq:LS_sol} does not apply. One may remove this redundancy by omitting the $a_{0}$ column; more generally, and to mitigate the ill-conditioning that grows with the approximation order, we add a Tikhonov regularizer (see, e.g., \cite{boyd_introduction_2018}), which yields a unique solution, penalizes large coefficients, and enhances numerical stability without relying on the explicit inversion of $B^{T}B$.}
The regulated LS problem is:
\begin{equation}
	\color{black}{\min_{\textbf{d}} \Vert B\textbf{d}-\textup{\textbf{y}} \Vert^2 + \lambda \Vert \textbf{d} \Vert^2}
\end{equation}
where $\lambda$ is the regularization parameter.
The solution is known to be:
\begin{equation}
	\textbf{d}^{*} = (B^{T}B + \lambda I)^{-1} B^{T} \textup{\textbf{y}}
	\label{eq:reg_LS_sol}
\end{equation}
where $I$ is the identity matrix.
\rev{Note that the error bounds of Thm.~\ref{th:FB_conv1}, Corollary~\ref{cor:diff}, and Thm.~\ref{lem:int} apply to the approximant with coefficients defined by Eqs.~\eqref{eq:g_k}--\eqref{eq:b_k}. 
The regulated LS estimate~\eqref{eq:reg_LS_sol} is a numerical surrogate, used here for data fitting and for initializing the NLP of Sec.~\ref{sec:primal}, for which no such bounds are claimed.}
\begin{exmp}
	We consider a case of data fitting using the mixed Bernstein-Fourier approximants.
	Figure~\ref{fig:LS_data} presents data obtained from  
    uniform sampling at a frequency of $100$~Hz.
	While this plot clearly shows oscillations, the functional form of these oscillations is not straightforward.
	Therefore, we use LS and regulated LS techniques to find coefficients for different approximants, namely for Bernstein-only, Fourier-only, and the proposed mixed basis approximants.
	\begin{figure}[htb]
		\centering
		\includegraphics[width=3.25in]{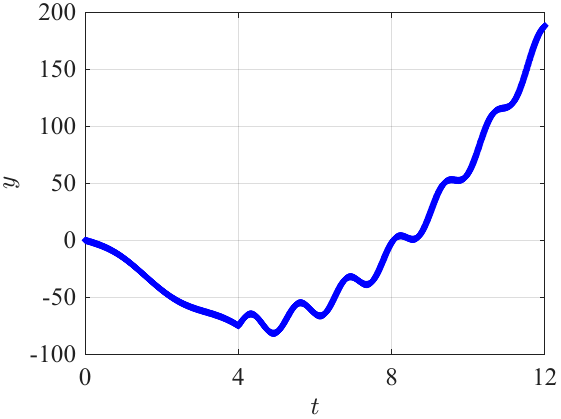}
		\caption{The data to be approximated.}
		\label{fig:LS_data}
	\end{figure}
	
	The Bernstein-only LS problem is solved by applying the solutions for the LS problem in~\eqref{eq:LS_sol} or the regulated LS problem in~\eqref{eq:reg_LS_sol} to the first $n_{\textup{\textbf{g}}} + 1$ columns of matrix $B$ and the first $n_{\textup{\textbf{g}}} + 1$ entries of $\textup{\textbf{d}}$.
	Similarly, we can solve the Fourier-only LS problem using the last $2n_{\textup{\textbf{\textbf{p}}}} + 1$ columns of matrix $B$ and the last $2n_{\textup{\textbf{\textbf{p}}}} + 1$ entries of $\textup{\textbf{d}}$.
	For simplicity of exposition, we set $n = n_{\textup{\textbf{g}}} = n_{\textup{\textbf{p}}}$.
	Hence, the dimensions for the decision vector for different approximants are
	\begin{equation}
		\dim(\textup{\textbf{d}}) = 
		\begin{cases}
			n + 1, & \rm{Bernstein-only} \\
			2n + 1, & \rm{Fourier-only} \\
			3n + 2, & \rm{Mixed\;Bernstein-Fourier}
		\end{cases}.
	\end{equation}

	We observe that mixed Bernstein-Fourier approximants have the smallest RMS errors, approximately three times lower than Bernstein-only approximants and ten times smaller than Fourier-only approximants.
    \textcolor{black}{Fig.~\ref{fig:LS_RMS} compares these errors for Bernstein-only (red), Fourier-only \rev{(blue)}, and mixed basis (black) approximants. The dots represent the LS solution, while squares represent the regulated LS solution ($\lambda = 10^{-14}$).}
    Moreover, we note that the regularization term improves the solution's numerical stability, especially for the Bernstein-only case.
    Even a very small regularization value makes the resulting LS solution more reliable and consistent.
    \begin{figure}[htb]
    	\centering
    	\includegraphics[width=3.25in]{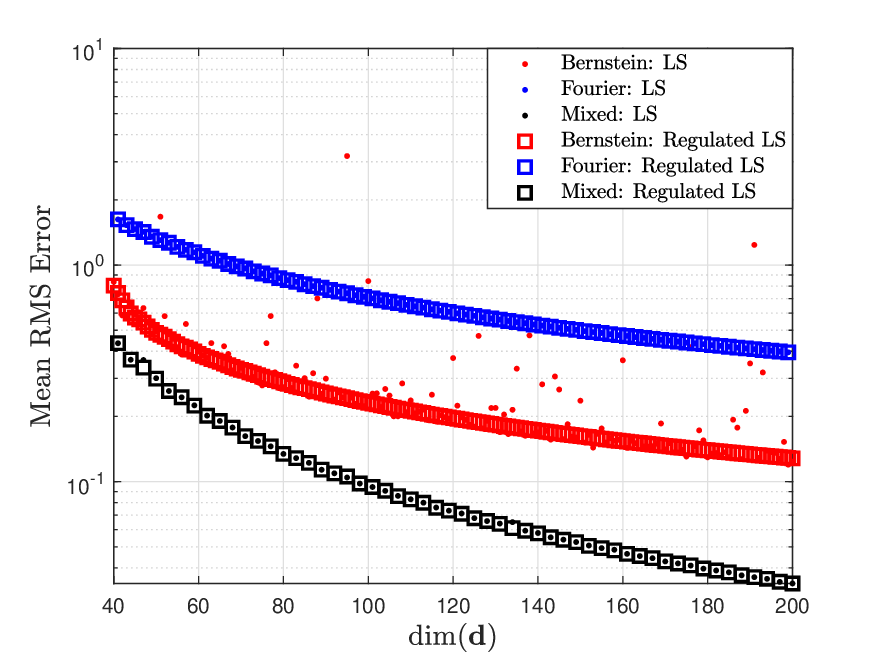}
    	\caption{\textcolor{black}{Mean RMS error comparison for different approximation bases.}}
    	\label{fig:LS_RMS}
    \end{figure}
\end{exmp}
In the ensuing section, we formulate an approximated version of the optimal motion planning problem using mixed Bernstein-Fourier approximants.
The regulated LS technique described above is used to find the Bernstein-Fourier coefficients for a given trajectory, which are needed to provide an initial guess to the NLP optimizer.

\section{Optimal Motion Planning Problem}
\label{sec:primal}
This section presents the optimal primal motion planning problem using a normalized time scale, as presented for the Bernstein-only case in~\cite{cichella_consistency_2022}.
We first generally formulate the problem and then present the approximated problem using the mixed Bernstein-Fourier approximants.
We then prove the feasibility and consistency of the approximated problem.
\subsection{Problem Formulation}
\label{sec:prob_for}
The optimal motion planning problem can be formulated as follows.
\begin{problem}[Problem $P$]
	\label{prob:P}
	Find $\textup{\textbf{x}}(t) : [0,1] \rightarrow \mathbb{R}^{m_{\textup{\textbf{x}}}}$ and $\textup{\textbf{u}}(t) : [0,1] \rightarrow \mathbb{R}^{m_{\textup{\textbf{u}}}}$ that minimize
	\begin{align}
		J(\textup{\textbf{x}}(t),\textup{\textbf{u}}(t)) &= E(\textup{\textbf{x}}(0),\textup{\textbf{x}}(1))
		 + \int_{0}^{1} F(\textup{\textbf{x}}(t),\textup{\textbf{u}}(t)) dt
	\end{align}
	subject to
	\textup{\begin{align}
		\label{eq:dyn_P}
		&\dot{\textbf{x}} = \textbf{f}(\textbf{x}(t),\textbf{u}(t)), \quad \forall t\in [0,1] \\
		\label{eq:eq_P}
		&\textbf{e}(\textbf{x}(0),\textbf{x}(1)) = 0 \\
		\label{eq:in_P}
		&\textbf{h}(\textbf{x}(t),\textbf{u}(t)) \leq 0, \quad \forall t\in [0,1]
	\end{align}}
	where
	\textup{$J: \mathbb{R}^{m_{\textbf{x}}} \times \mathbb{R}^{m_{\textbf{u}}} \rightarrow \mathbb{R}$,
	$E: \mathbb{R}^{m_{\textbf{x}}} \times \mathbb{R}^{m_{\textbf{x}}}  \rightarrow \mathbb{R}$,
	$F: \mathbb{R}^{m_{\textbf{x}}} \times \mathbb{R}^{m_{\textbf{u}}} \rightarrow \mathbb{R}$,
	$\textbf{f}: \mathbb{R}^{m_{\textbf{x}}} \times \mathbb{R}^{m_{\textbf{u}}} \rightarrow \mathbb{R}^{m_{\textbf{x}}}$, 
	$\textbf{e}: \mathbb{R}^{m_{\textbf{x}}} \times \mathbb{R}^{m_{\textbf{x}}} \rightarrow \mathbb{R}^{m_{\textbf{e}}}$, and
	$\textbf{h}: \mathbb{R}^{m_{\textbf{x}}} \times \mathbb{R}^{m_{\textbf{u}}} \rightarrow \mathbb{R}^{m_{\textbf{h}}}$}.
\end{problem}
The solution to this problem is generally intractable; therefore, we aim to approximate it.
We use equidistant time nodes $t_{k} = \frac{k}{n_{t}}, k=0,\dots,n_{t}$, where $n_{t} \in \mathbb{N}$.
The mixed approximants are comprised of both Bernstein polynomials and Fourier series to approximate the state and control trajectories:
\begin{align}
	\label{eq:x_app}
	\textbf{x}_n&(t)  = \sum_{k=0}^{n_{B}} \bar{\textbf{c}}_{x,k} b_{k,n}(t) + \frac{\textbf{a}_{x,0}}{2} + \sum_{k=1}^{n_{F}} \left[ \textbf{a}_{x,k} \cos \left( 2 \pi k t\right) + \textbf{b}_{x,k} \sin \left( 2 \pi k t\right) \right] \\
	\label{eq:u_app}
	\textbf{u}_n&(t)  =\sum_{k=0}^{n_{B}} \bar{\textbf{c}}_{u,k} b_{k,n}(t) + \frac{\textbf{a}_{u,0}}{2}  + \sum_{k=1}^{n_{F}} \left[ \textbf{a}_{u,k} \cos \left( 2 \pi k t\right) + \textbf{b}_{u,k} \sin \left( 2 \pi k t\right) \right],
\end{align}
where $\textbf{x}_n(t) : [0,1] \rightarrow \mathbb{R}^{m_{\textbf{x}}}$, $\textbf{u}_n(t) : [0,1] \rightarrow \mathbb{R}^{m_{\textbf{u}}}$, $n_{B} \in \mathbb{N}$ and $n_{F} \in \mathbb{N}$ are the order of approximation of the Bernstein and Fourier basis functions, respectively, and $n \dfn \min( n_{t},n_{B},n_{F}) $. 
Note that the maximal dimension of the decision vector is $(n_{B} + 2 n_{F} + 2)(m_{\textbf{x}} + m_{\textbf{u}})$, which is independent of $n_{t}$.
\begin{remark}
	\label{rem:nt}
	Choosing $n_{t}$ is crucial for the performance of the suggested method.
	Although the dimension of the decision vector is independent of $n_{t}$, it still affects the performance.
	If we set $n_{t}$ to be too large, we will have numerical problems as the matrix of the basis functions becomes ill-conditioned.
	In contrast, if $n_{t}$ is too small, we can face sampling problems such as aliasing~\cite{ross_low-thrust_2007}. 
	Therefore, we can easily overcome this challenge by obeying the guidelines of the Nyquist-Shannon sampling theorem and ensuring that $n_{t} > 2 n_{F}$. \rev{This is a sampling condition, imposed at each fixed order to resolve the Fourier block on the node grid, and it is all that the primal results of this section require. The dual analysis of Sec.~\ref{sec:dual} constrains instead how $n_{t}$ must grow relative to $n_{B}$ and $n_{F}$ as the approximation is refined, and is correspondingly stronger.}
	An appropriate initial value can be $n_{t} = n_{B}$, as implemented in the Bernstein-only case~\cite{cichella_optimal_2021,cichella_consistency_2022}.
\end{remark}

With these approximants, we can formulate the following NLP problem.
\begin{problem}[Problem $P_n$]
	\label{prob:Pn}
	Find $\textup{\textbf{x}}_n(t_{k})$ and $\textup{\textbf{u}}_n(t_{k})$ \rev{(equivalently, the coefficients of their mixed approximants)} that minimize
	\textup{\begin{align}
		J_n(\textbf{x}_n(t_{k}),\textbf{u}_n(t_{k})) = E(\textbf{x}_n(0),\textbf{x}_n(1)) +w  \sum_{k=0}^{n_{t}} F(\textbf{x}_n(t_{k}),\textbf{u}_n(t_{k}))
	\end{align}}
	subject to
	\textup{\begin{align}
		\label{eq:dyn_Pn}
		&\Vert\dot{\textbf{x}}_n(t_{k}) - \textbf{f}(\textbf{x}_n(t_{k}),\textbf{u}_n(t_{k}))\Vert \leq \delta_{P}^{n}, \; \forall k=0,\dots,n_{t}\\
		\label{eq:eq_Pn}
		&\Vert\textbf{e}(\textbf{x}_n(0),\textbf{x}_n(1)) \Vert \leq \delta_{P}^{n} \\
		\label{eq:in_Pn}
		&\textbf{h}(\textbf{x}_n(t_{k}),\textbf{u}_n(t_{k})) \leq \delta_{P}^{n} \cdot \textbf{1}, \; \forall k=0,\dots,n_{t}
	\end{align}}
	where $w = \frac{1}{n_{t} + 1}$.
\end{problem}

\subsection{Analysis}
\label{sec:analysis}
This section addresses the feasibility and consistency issues of Problem~$P_{n}$, following the lines of analysis presented in~\cite{cichella_optimal_2021} for approximations using only Bernstein polynomials.
We first make the following assumptions regarding Problem~$P$.
\begin{assumption}
	\label{ass:lip}
	The functions $E, F, \textup{\textbf{f}}, \textup{\textbf{e}}$, and $\textup{\textbf{h}}$ are Lipschitz continuous with respect to their arguments.
\end{assumption}
\begin{assumption}
	\label{ass:exis}
	Problem~$P$ admits an optimal solution \textup{$(\textbf{x}^{*}(t),\textbf{u}^{*}(t))$} that satisfies \textup{$\textbf{x}^{*}(t)\in \mathcal{C}^{1}_{m_{\textbf{x}}}$} and \textup{$\textbf{u}^{*}(t)\in \mathcal{C}^{0}_{m_{\textbf{u}}}$ on $[0,1]$}.
\end{assumption}

\rev{Fix any decomposition of a feasible pair \revthree{$(\textbf{x},\textbf{u})$, with \textup{$\textbf{x}\in\mathcal{C}^{1}_{m_{\textbf{x}}}$} and \textup{$\textbf{u}\in\mathcal{C}^{0}_{m_{\textbf{u}}}$} on $[0,1]$,} into periodic and non-periodic parts, \textup{$\textbf{x} = \textbf{x}_{\textbf{g}} + \textbf{x}_{\textbf{p}}$} and \textup{$\textbf{u} = \textbf{u}_{\textbf{g}} + \textbf{u}_{\textbf{p}}$}, with \textup{$\textbf{x}_{\textbf{p}} \in \mathcal{P}^{r}_{m_{\textbf{x}}}$} ($r \geq 2$) and \textup{$\textbf{u}_{\textbf{p}} \in \mathcal{P}^{\tilde{r}}_{m_{\textbf{u}}}$} ($\tilde{r} \geq 1$) arbitrary; such a \revthree{decomposition, which we call admissible,} always exists, e.g., $\textbf{x}_{\textbf{p}} \equiv \textbf{0}$, $\textbf{u}_{\textbf{p}} \equiv \textbf{0}$. Let \textup{$\delta_{\textbf{x}}^{n}\revthree{\pX}$, $\delta_{\dot{\textbf{x}}}^{n}\revthree{\pX}$, $\delta_{\textbf{u}}^{n}\revthree{\pU}$} denote the associated error bounds of Eqs.~\eqref{eq:delta_f}--\eqref{eq:delta_fd}, evaluated for $(\textbf{x}_{\textbf{g}},\textbf{x}_{\textbf{p}},n_{B},n_{F},r)$ and $(\textbf{u}_{\textbf{g}},\textbf{u}_{\textbf{p}},n_{B},n_{F},\tilde{r})$.}
The following theorem summarizes the feasibility result of the primal problem.
\begin{theorem}[Feasibility]
	\label{th:feas}
	Let \textup{$(\textbf{x}(t),\textbf{u}(t))$} be a feasible solution to Problem $P$ \revthree{with \textup{$\textbf{x}(t)\in\mathcal{C}^{1}_{m_{\textbf{x}}}$} and \textup{$\textbf{u}(t)\in\mathcal{C}^{0}_{m_{\textbf{u}}}$} on $[0,1]$ (e.g., the optimal solution of Assumption~\ref{ass:exis})} with the decomposition and error bounds $\delta_{\textbf{x}}^{n}\revthree{\pX}$, $\delta_{\dot{\textbf{x}}}^{n}\revthree{\pX}$, $\delta_{\textbf{u}}^{n}\revthree{\pU}$ introduced above.
	Let
	\begin{align}
		\label{eq:del_P_n}
		\delta_{P}^{n}\revthree{\pXU} = C_{P} \max\{\delta_{\dot{\textbf{x}}}^{n}\revthree{\pX} , \delta_{\textbf{x}}^{n}\revthree{\pX} , \delta_{\textbf{u}}^{n}\revthree{\pU} \} 
	\end{align}
	where $C_{P}$ is a \revthree{sufficiently large} positive constant independent of $n_{B}$, $n_{F}$, and $n_{t}$.
	Then, Problem $P_{n}$ \revthree{with $\delta_{P}^{n}\pXU$} admits a feasible solution \textup{$(\textbf{x}_{n},\textbf{u}_{n})$} for \revthree{all $n_{B}, n_{t} \in \mathbb{Z}^{+}$ and $n_{F} \geq 2$}\rev{, and $\delta_{P}^{n}\revthree{\pXU} \rightarrow 0$ as $n \rightarrow \infty$}.
\end{theorem}
The proof is deferred to Appendix~\ref{sec:app1}.
Proving the consistency of the solution requires an additional assumption that ensures the convergence of the approximated solution.
\begin{assumption}
	\label{ass:conv}
	\revthree{For an admissible decomposition $(\textbf{x}_{\textbf{p}},\textbf{u}_{\textbf{p}})$ of the optimal solution \textup{$(\textbf{x}^{*}(t),\textbf{u}^{*}(t))$} of Assumption~\ref{ass:exis}, let \textup{$\{(\textbf{x}_{n}(t),\textbf{u}_{n}(t))\}$} be a sequence of optimal solutions to Problem~$P_{n}$ with $\delta_{P}^{n}\pXU$.} There exist \textup{$\textbf{x}^{\infty}(t)\in\mathcal{C}^{1}_{m_{\textbf{x}}}$} and \textup{$\textbf{u}^{\infty}(t)\in\mathcal{C}^{0}_{m_{\textbf{u}}}$} on $[0,1]$ such that\rev{, uniformly on $[0,1]$,}
	\begin{equation}
		\lim\limits_{n\rightarrow\infty} (\textbf{x}_{n}(t),\rev{\dot{\textbf{x}}_{n}(t),}\textbf{u}_{n}(t)) = (\textbf{x}^{\infty}(t),\rev{\dot{\textbf{x}}^{\infty}(t),}\textbf{u}^{\infty}(t)).
	\end{equation}
\rev{Assumption~\ref{ass:conv} is the analogue, for the present setting, of Assumption~1 in~\cite{gong_connections_2008}, which postulates a uniform accumulation point of the sequence $(\dot{\textbf{x}}_{N}(\cdot),\textbf{u}_{N}(\cdot))$ with continuous limits; see also Assumption~3 of~\cite{cichella_optimal_2021} in the Bernstein-only setting. The \revthree{mode of} convergence postulated here is exhibited, e.g., by the feasible sequence constructed in the proof of Thm.~\ref{th:feas}.}
\end{assumption}
\begin{theorem}
	\label{th:consis}
	\revthree{Let $(\textbf{x}_{\textbf{p}},\textbf{u}_{\textbf{p}})$ be an admissible decomposition of the optimal solution \textup{$(\textbf{x}^{*}(t),\textbf{u}^{*}(t))$} of Assumption~\ref{ass:exis}, and let} \textup{$\{(\textbf{x}_{n}(t),\textbf{u}_{n}(t))\}$} be a sequence of optimal solutions to Problem $P_{n}$\revthree{, with $\delta_{P}^{n}\pXU$ given by Eq.~\eqref{eq:del_P_n},} satisfying Assumption~\ref{ass:conv}. Then \textup{$(\textbf{x}^{\infty}(t),\textbf{u}^{\infty}(t))$} is an optimal solution of Problem~$P$.
\end{theorem}
The proof is deferred to Appendix~\ref{sec:app2}.

\section{Necessary Conditions for Optimal Motion Planning}
\label{sec:dual}
This section presents the theoretical results related to the dual problems of Problem $P$ and Problem $P_{n}$, following the guidelines of~\cite{cichella_consistency_2022}.
\subsection{Problem Formulation}
\subsubsection{Applying PMP to Problem P}
First, we derive the first-order necessary conditions for Problem $P$.
Let $\bm{\lambda}(t) : [0,1] \rightarrow \mathbb{R}^{m_{\textbf{x}}}$ be the trajectory costate, and let $\bm{\mu}(t) : [0,1] \rightarrow \mathbb{R}^{m_{\textbf{h}}}$ and $\bm{\nu} \in \mathbb{R}^{m_{\textbf{e}}}$ be the multipliers for the inequality and equality constraints, respectively.
We define the Lagrangian of the Hamiltonian~\cite{hartl_survey_1995}:
\begin{equation}
	\mathcal{L}(\textbf{x}(t),\textbf{u}(t),\bm{\lambda}(t),\bm{\mu}(t)) = 
	\mathcal{H}(\textbf{x}(t),\textbf{u}(t),\bm{\lambda}(t)) + \bm{\mu}^{T}(t)\textbf{h}(\textbf{x}(t),\textbf{u}(t)),
\end{equation}
where the Hamiltonian is 
\begin{equation}
	\mathcal{H}(\textbf{x}(t),\textbf{u}(t),\bm{\lambda}(t)) = F(\textbf{x}(t),\textbf{u}(t))
	+ \bm{\lambda}^{T}(t)\textbf{f}(\textbf{x}(t),\textbf{u}(t)).
\end{equation}
The dual problem is formulated as follows~\cite{hartl_survey_1995}.
\begin{problem}[Problem $P_{\bm{\lambda}}$]
	\label{prob:P_lamb}
	Determine \textup{$\textbf{x}(t)$, $\textbf{u}(t)$, $\bm{\lambda}(t)$, $\bm{\mu}(t)$}, and $\bm{\nu}$ that for all $t \in [0,1]$ satisfy Eqs.~\eqref{eq:dyn_P}--\eqref{eq:in_P} and 
	\textup{\begin{align}
		\label{eq:mu_h_0}
		& \bm{\mu}^{T}(t)\textbf{h}(\textbf{x}(t),\textbf{u}(t)) = 0, \quad \bm{\mu}(t) \geq 0, \\
		\label{eq:lamb_dot}
		& \dot{\bm{\lambda}}^{T}(t) + \mathcal{L}_{\textbf{x}}(\textbf{x}(t),\textbf{u}(t),\bm{\lambda}(t),\bm{\mu}(t)) = 0, \\
		\label{eq:closure1}
		& \bm{\lambda}^{T}(0) = - \bm{\nu}^{T}\textbf{e}_{\textbf{x}(0)}(\textbf{x}(0),\textbf{x}(1)) - 
		E_{\textbf{x}(0)}(\textbf{x}(0),\textbf{x}(1)),\\
		\label{eq:closure2}
		& \bm{\lambda}^{T}(1) =  \bm{\nu}^{T}\textbf{e}_{\textbf{x}(1)}(\textbf{x}(0),\textbf{x}(1)) + 
		E_{\textbf{x}(1)}(\textbf{x}(0),\textbf{x}(1)), \\
		\label{eq:L_u}
		& \mathcal{L}_{\textbf{u}}(\textbf{x}(t),\textbf{u}(t),\bm{\lambda}(t),\bm{\mu}(t)) = 0.
	\end{align}}
\end{problem}
The subscripts are used to denote partial derivatives, e.g., $\textbf{h}_{\textbf{x}}(\textbf{x}(t),\textbf{u}(t)) = \frac{\partial}{\partial \textbf{x}} \textbf{h}(\textbf{x}(t),\textbf{u}(t))$.
\subsubsection{KKT Conditions of Problem \texorpdfstring{$P_{n}$}{Pn}}
We derive the necessary conditions of Problem~$P_{n}$. 
We first introduce the approximation of the trajectory costates and the inequality multipliers:
\begin{align}
	\label{eq:lamb_app}
	\bm{\lambda}_n&(t)  = \sum_{k=0}^{n_{B}} \bar{\textbf{c}}_{\lambda,k} b_{k,n}(t) + \frac{\textbf{a}_{\lambda,0}}{2}  + \sum_{k=1}^{n_{F}} \left[ \textbf{a}_{\lambda,k} \cos \left( 2 \pi k t\right) + \textbf{b}_{\lambda,k} \sin \left( 2 \pi k t\right) \right] \\
	\label{eq:mu_app}
	\bm{\mu}_n&(t)  = \sum_{k=0}^{n_{B}} \bar{\textbf{c}}_{\mu,k} b_{k,n}(t) + \frac{\textbf{a}_{\mu,0}}{2} + \sum_{k=1}^{n_{F}} \left[ \textbf{a}_{\mu,k} \cos \left( 2 \pi k t\right) + \textbf{b}_{\mu,k} \sin \left( 2 \pi k t\right) \right],
\end{align}
where $\bm{\lambda}_n(t) : [0,1] \rightarrow \mathbb{R}^{m_{\textbf{x}}}$, and $\bm{\mu}_n(t) : [0,1] \rightarrow \mathbb{R}^{m_{\textbf{h}}}$ are the mixed Bernstein-Fourier approximations of $\bm{\lambda}(t)$ and $\bm{\mu}(t)$, respectively.

Using these notations, we present the Lagrangian of Problem~$P_{n}$
\begin{align}
	\mathcal{L}_{n} & = E(\textbf{x}_n(0),\textbf{x}_n(1)) + w \sum_{k=0}^{n_{t}} F(\textbf{x}_n(t_{k}),\textbf{u}_n(t_{k})) + \bar{\bm{\nu}}^{T} \textbf{e}(\textbf{x}_n(0),\textbf{x}_n(1)) \notag \\
	& + \sum_{k=0}^{n_{t}} \bm{\lambda}_n^{T}(t_{k})\big(- \dot{\textbf{x}}_n(t_{k}) + \textbf{f}(\textbf{x}_n(t_{k}),\textbf{u}_n(t_{k}))\big)
	 + \sum_{k=0}^{n_{t}} \bm{\mu}_n^{T}(t_{k}) \textbf{h}(\textbf{x}_n(t_{k}),\textbf{u}_n(t_{k})) 
\end{align}
where $\bar{\bm{\nu}} \in \mathbb{R}^{m_{\textbf{e}}}$.
Now, we can formulate the dual problem of Problem~$P_{n}$.
\begin{problem}[Problem $P_{\bm{n\lambda}}$]
	Find $\textup{\textbf{x}}_n(t_{k})$, $\textup{\textbf{u}}_n(t_{k})$, $\bm{\lambda}_n(t_{k})$ , $\bm{\mu}_n(t_{k})$, $\bar{\bm{\nu}}$ that satisfy the primal feasibility conditions in Eqs.~\eqref{eq:dyn_Pn}--\eqref{eq:in_Pn}, the complementary slackness and dual feasibility conditions
	\textup{\begin{align}
		\label{eq:in_eq_Pnl}
		 \left\Vert \bm{\mu}_n^{T}(t_{k}) \textbf{h}(\textbf{x}_n(t_{k}),\textbf{u}_n(t_{k})) \right\Vert \leq \frac{\delta_{D}^{\revthree{n}}}{n_{t}}, 
		 \quad \bm{\mu}_n(t_{k}) \geq - \frac{\delta_{D}^{\revthree{n}}}{n_{t}} \bm{1}, \quad \forall k=0,\dots,n_{t},
	\end{align}}
	the stationarity conditions
	\textup{\begin{align}
		&\left\Vert \frac{\partial \mathcal{L}_{n}}{\partial \bar{\textbf{c}}_{x,k}} \right\Vert \leq \delta_{D}^{\revthree{n}}, \quad
		\left\Vert \frac{\partial \mathcal{L}_{n}}{\partial \bar{\textbf{c}}_{u,k}} \right\Vert \leq \delta_{D}^{\revthree{n}}, \quad \rev{k=0,\dots,n_{B},}  \notag\\
		& \left\Vert \frac{\partial \mathcal{L}_{n}}{\partial {\textbf{a}}_{x,k}} \right\Vert \leq \delta_{D}^{\revthree{n}}, \quad 
		\left\Vert \frac{\partial \mathcal{L}_{n}}{\partial {\textbf{a}}_{u,k}} \right\Vert \leq \delta_{D}^{\revthree{n}}, \quad \rev{k=0,\dots,n_{F},}  \notag\\
		\label{eq:stat}
		& \left\Vert \frac{\partial \mathcal{L}_{n}}{\partial {\textbf{b}}_{x,k}} \right\Vert \leq \delta_{D}^{\revthree{n}}, \quad
		\left\Vert \frac{\partial \mathcal{L}_{n}}{\partial {\textbf{b}}_{u,k}} \right\Vert \leq \delta_{D}^{\revthree{n}}, \quad \rev{k=1,\dots,n_{F},}
	\end{align}}
	and the closure conditions
	\textup{\begin{align}
		\label{eq:close1}
		& \left\Vert \frac{\bm{\lambda}_{n}^{T}(0)}{w} + \bar{\bm{\nu}}^{T}\textbf{e}_{\textbf{x}(0)}(\textbf{x}_{n}(0),\textbf{x}_{n}(1)) + 
		E_{\textbf{x}(0)}(\textbf{x}_{n}(0),\textbf{x}_{n}(1)) \right\Vert \leq \delta_{D}^{\revthree{n}}, \\
		\label{eq:close2}
		& \left\Vert \frac{\bm{\lambda}_{n}^{T}(1)}{w} - \bar{\bm{\nu}}^{T}\textbf{e}_{\textbf{x}(1)}(\textbf{x}_{n}(0),\textbf{x}_{n}(1)) - E_{\textbf{x}(1)}(\textbf{x}_{n}(0),\textbf{x}_{n}(1)) \right\Vert  \leq \delta_{D}^{\revthree{n}}, 
	\end{align}}
	where $ \delta_{D}^{\revthree{n}}$ is a small positive number that depends on $n$ and satisfies $\lim_{n\rightarrow\infty} \delta_{D}^{\revthree{n}} = 0$\rev{; an explicit tolerance realizing this requirement\revthree{, which also depends on the chosen decomposition,} and the condition on the node count under which it vanishes, are given below}.
\end{problem}
\begin{remark}
	The closure conditions must be added to Problem~$P_{n\lambda}$ to obtain a consistent approximation of Problem $P_{\lambda}$.
	These conditions come from the transversality conditions in Eqs.~\eqref{eq:closure1}--\eqref{eq:closure2}, which do not appear in the KKT conditions of Problem~$P_{n}$.
	More information regarding these conditions can be found in~\cite{fahroo_costate_2001,gong_connections_2008}.
\end{remark}
\subsection{Analysis of the Dual Problems}
This section analyzes the feasibility and consistency of Problem~$P_{n\lambda}$.
We first show the existence of a solution to this problem, and then we investigate its convergence properties when $n\rightarrow\infty$.
This yields the covector mapping for the mixed Bernstein-Fourier approximants summarized in \revthree{Fig.}~\ref{fig:Mapping}, which relates the solutions to the different Problems presented in this paper.
\begin{figure}[htb]
	\centering
	\includegraphics[width=3.25in]{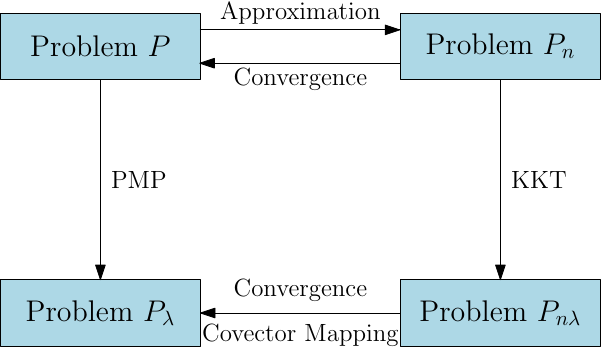}
	\caption{Schematic of the \revthree{covector mapping} for mixed Bernstein-Fourier approximants.}
	\label{fig:Mapping}
\end{figure}
\revthree{An optimal solution of Problem~$P_{n\lambda}$ denotes a solution of its conditions whose primal pair $(\textbf{x}_{n},\textbf{u}_{n})$ is optimal for Problem~$P_{n}$. Similarly, an optimal solution of Problem~$P_{\lambda}$ denotes a solution of its conditions whose primal pair $(\textbf{x},\textbf{u})$ is optimal for Problem~$P$.}
We begin by making the following assumptions regarding Problem~$P_{\lambda}$:
\begin{assumption}
	\label{ass:lip_d}
	The functions $E, F, \textup{\textbf{f}}, \textup{\textbf{e}}$, and \textup{$\textbf{h}$} are continuously differentiable with respect to their arguments, and their gradients are Lipschitz continuous over the domain.
\end{assumption}
\begin{assumption}
	\label{ass:exis_d}
	Problem~$P_{\lambda}$ admits an optimal solution \textup{$(\textbf{x}^{*}(t),\textbf{u}^{*}(t),\bm{\lambda}^{*}(t),\bm{\mu}^{*}(t),\bm{\nu}^{*})$ that satisfies $\textbf{x}^{*}(t)\in \mathcal{C}^{1}_{m_{\textbf{x}}}$, $\textbf{u}^{*}(t)\in \mathcal{C}^{0}_{m_{\textbf{u}}}$, $\bm{\lambda}^{*}(t)\in \mathcal{C}^{1}_{m_{\textbf{x}}}$, $\bm{\mu}^{*}(t)\in \mathcal{C}^{0}_{m_{\textbf{h}}}$ on $[0,1]$, and $\bm{\nu}^{*} \in \mathbb{R}^{m_{\textbf{e}}}$}.
\end{assumption}
\begin{assumption}
	\label{ass:state}
	The inequality constraint function \textup{$\textbf{h}(\textbf{x}(t),\textbf{u}(t))$} does not contain pure state constraints. 
\end{assumption}
Note that these assumptions are more restrictive than the ones in Sec.~\ref{sec:primal}. 
Therefore, some problems which can be solved directly cannot be analyzed using the results of this section.
For example, the numerical problem in Sec.~\ref{sec:observe} can be solved directly. It violates Assumption~\ref{ass:state}, however, so we cannot verify that the solution satisfies the necessary conditions of Problem $P_{\lambda}$ because the costates become discontinuous.

\rev{Extend the primal decomposition of Thm.~\ref{th:feas} to the multipliers \revthree{of a feasible solution $(\textbf{x},\textbf{u},\bm{\lambda},\bm{\mu},\bm{\nu})$ of Problem~$P_{\lambda}$ with \textup{$\bm{\lambda}\in\mathcal{C}^{1}_{m_{\textbf{x}}}$} and \textup{$\bm{\mu}\in\mathcal{C}^{0}_{m_{\textbf{h}}}$} on $[0,1]$, whose primal pair is decomposed as in Sec.~\ref{sec:analysis}}: let \textup{$\bm{\lambda}_{\textbf{p}} \in \mathcal{P}^{r_{\lambda}}_{m_{\textbf{x}}}$} ($r_{\lambda} \geq 2$) and \textup{$\bm{\mu}_{\textbf{p}} \in \mathcal{P}^{r_{\mu}}_{m_{\textbf{h}}}$} ($r_{\mu} \geq 1$) be arbitrary, with \textup{$\bm{\lambda}_{\textbf{g}} \dfn \bm{\lambda} - \bm{\lambda}_{\textbf{p}}$} and \textup{$\bm{\mu}_{\textbf{g}} \dfn \bm{\mu} - \bm{\mu}_{\textbf{p}}$}\revthree{; the resulting decomposition $(\textbf{x}_{\textbf{p}},\textbf{u}_{\textbf{p}},\bm{\lambda}_{\textbf{p}},\bm{\mu}_{\textbf{p}})$ is again called admissible}, and let \textup{$\delta_{\bm{\lambda}}^{n}\revthree{\pL}$, $\delta_{\dot{\bm{\lambda}}}^{n}\revthree{\pL}$, $\delta_{\bm{\mu}}^{n}\revthree{\pM}$} be the corresponding bounds of Eqs.~\eqref{eq:delta_f}--\eqref{eq:delta_fd}, with $r_{\lambda}$, $r_{\mu}$ in place of $r$; the primal quantities \textup{$\delta_{\textbf{x}}^{n}\revthree{\pX}$, $\delta_{\dot{\textbf{x}}}^{n}\revthree{\pX}$, $\delta_{\textbf{u}}^{n}\revthree{\pU}$} are as in Thm.~\ref{th:feas}.}

\begin{theorem}[Feasibility of Problem $P_{n\lambda}$]
	\label{th:feas_d}
	Let \textup{$(\textbf{x}(t),\textbf{u}(t),\bm{\lambda}(t),\bm{\mu}(t),\bm{\nu})$} be a feasible solution to Problem $P_{\lambda}$ \revthree{with \textup{$\textbf{x}(t)\in\mathcal{C}^{1}_{m_{\textbf{x}}}$, $\textbf{u}(t)\in\mathcal{C}^{0}_{m_{\textbf{u}}}$, $\bm{\lambda}(t)\in\mathcal{C}^{1}_{m_{\textbf{x}}}$}, and \textup{$\bm{\mu}(t)\in\mathcal{C}^{0}_{m_{\textbf{h}}}$} on $[0,1]$ (e.g., the optimal solution of Assumption~\ref{ass:exis_d})} with the multiplier decomposition and error bounds \textup{$\delta_{\bm{\lambda}}^{n}\revthree{\pL}$, $\delta_{\dot{\bm{\lambda}}}^{n}\revthree{\pL}$, $\delta_{\bm{\mu}}^{n}\revthree{\pM}$} introduced above.
	Let $\delta_{P}^{\revthree{n}}\revthree{\pXU}$ be as in~\eqref{eq:del_P_n} \revthree{for the primal pair $(\textbf{x},\textbf{u})$}, and
	\textup{\begin{align}
		\label{eq:del_D}
		\delta_{D}^{\revthree{n}}\revthree{\pXULM} = C_{D}\max\Big\{&\delta_{\textbf{x}}^{n}\revthree{\pX} , \delta_{\textbf{u}}^{n}\revthree{\pU} , \delta_{\bm{\lambda}}^{n}\revthree{\pL} , \delta_{\bm{\mu}}^{n}\revthree{\pM} , \frac{1}{\sqrt{n}} \rev{, \delta_{\dot{\bm{\lambda}}}^{n}\revthree{\pL} ,} \notag \\
		&\rev{W_{\textbf{u}_{\textbf{g}}}\Big(\frac{1}{\sqrt{n}}\Big) , W_{\bm{\mu}_{\textbf{g}}}\Big(\frac{1}{\sqrt{n}}\Big) , \frac{n_{B}^{2}}{n_{t}} , \frac{n_{F}^{2}}{n_{t}}} \Big\}
	\end{align}}
	where $C_{D}$ is a \revthree{sufficiently large} positive constant independent of $n_{B}$, $n_{F}$, and $n_{t}$.
	Problem $P_{n\lambda}$ \revthree{with $\delta_{P}^{n}\pXU$ and $\delta_{D}^{n}\pXULM$} then admits a feasible solution \textup{$(\textbf{x}_{n},\textbf{u}_{n},\bm{\lambda}_{n},\bm{\mu}_{n},\bar{\bm{\nu}})$} for \revthree{all $n_{B}, n_{t} \in \mathbb{Z}^{+}$ and $n_{F} \geq 2$}\rev{; if, moreover, $n_{B}^{2}/n_{t} \rightarrow 0$ and $n_{F}^{2}/n_{t} \rightarrow 0$ as $n \rightarrow \infty$, then $\delta_{P}^{n}\revthree{\pXU}, \delta_{D}^{n}\revthree{\pXULM} \rightarrow 0$}.
\end{theorem}
The proof is deferred to Appendix~\ref{sec:app3}. 
\begin{remark}
	\label{rem:coupling}
	\rev{The coupling $n_{B}^{2}/n_{t}, n_{F}^{2}/n_{t} \rightarrow 0$ is a quadrature-resolution requirement rather than a structural one: the dual stationarity conditions are enforced as sums over the $n_{t}$ sample nodes, and because their integrands contain the basis derivatives, matching these sums to the continuous conditions incurs errors of order $n_{B}^{2}/n_{t}$ and $n_{F}^{2}/n_{t}$. Under this coupling the tolerance~\eqref{eq:del_D} vanishes, and thereby realizes the requirement $\lim_{n \rightarrow \infty}\delta_{D}^{\revthree{n}} = 0$ imposed in the definition of Problem $P_{n\lambda}$.}

	\rev{This condition and the sampling condition of Remark~\ref{rem:nt} act on the same three orders but differ in kind. The sampling condition, $n_{t} > 2 n_{F}$, is imposed at each fixed order: it resolves the Fourier block on the node grid and governs the choice of $n_{t}$ when Problem $P_{n}$ is solved at a given order. The present condition instead constrains the joint growth of $n_{t}$, $n_{B}$, and $n_{F}$ as $n \rightarrow \infty$, and it is required only for the dual convergence guarantees of this section; the primal results of Sec.~\ref{sec:primal} need only the sampling condition. It is the stronger of the two and subsumes the sampling condition asymptotically, since $n_{F}^{2}/n_{t} \rightarrow 0$ implies $n_{t} > 2 n_{F}$ for all sufficiently large $n$.}

	\rev{The coupling is mild in practice: since the decision-vector dimension is independent of $n_{t}$, it is enough to let $n_{t}$ grow faster than $\max(n_{B},n_{F})^{2}$, refining the sampling grid at no cost in optimization variables.}
\end{remark}
Proving the consistency of the solution requires an additional assumption that ensures the convergence of the approximated solution.
\begin{assumption}
	\label{ass:conv_d}
	\revthree{For an admissible decomposition $(\textbf{x}_{\textbf{p}},\textbf{u}_{\textbf{p}},\bm{\lambda}_{\textbf{p}},\bm{\mu}_{\textbf{p}})$ of the optimal solution of Assumption~\ref{ass:exis_d}, let \textup{$\{(\textbf{x}_{n}(t),\allowbreak\textbf{u}_{n}(t),\allowbreak\bm{\lambda}_{n}(t),\allowbreak\bm{\mu}_{n}(t),\allowbreak\bar{\bm{\nu}}_{n})\}$} be a sequence of optimal solutions to Problem~$P_{n\lambda}$ with $\delta_{P}^{n}\pXU$ and $\delta_{D}^{n}\pXULM$.} There exist \textup{$\textbf{x}^{\infty}(t)\in\mathcal{C}^{1}_{m_{\textbf{x}}}$, $\textbf{u}^{\infty}(t)\in\mathcal{C}^{0}_{m_{\textbf{u}}}$,  $\bm{\lambda}^{\infty}(t)\in\mathcal{C}^{1}_{m_{\textbf{x}}}$, $\bm{\mu}^{\infty}(t)\in\mathcal{C}^{0}_{m_{\textbf{h}}}$ on $[0,1]$}, and \textup{$\bar{\bm{\nu}}^{\infty}\in \mathbb{R}^{m_{\textbf{e}}}$}, such that\rev{, uniformly on $[0,1]$,}
	\textup{\begin{equation}
		\lim\limits_{n\rightarrow\infty} \left(\textbf{x}_{n}(t),\rev{\dot{\textbf{x}}_{n}(t),}\textbf{u}_{n}(t),\frac{\bm{\lambda}_{n}(t)}{w},\frac{\bm{\mu}_{n}(t)}{w},\bar{\bm{\nu}}_{\revthree{n}}\right) = (\textbf{x}^{\infty}(t),\rev{\dot{\textbf{x}}^{\infty}(t),}\textbf{u}^{\infty}(t),\bm{\lambda}^{\infty}(t),\bm{\mu}^{\infty}(t),\bar{\bm{\nu}}^{\infty}).
	\end{equation}}
\end{assumption}
\rev{Assumption~\ref{ass:conv_d} extends Assumption~\ref{ass:conv} to the dual variables. Note that no convergence of the multipliers' derivatives is postulated: only the regularity of the limits, $\bm{\lambda}^{\infty} \in \mathcal{C}^{1}_{m_{\textbf{x}}}$ and $\bm{\mu}^{\infty} \in \mathcal{C}^{0}_{m_{\textbf{h}}}$, enters the analysis.} \revthree{The reason is structural. The dynamics are imposed pointwise at the nodes through Eq.~\eqref{eq:dyn_Pn}, so passing them to the limit (Step~1 of the proof of Thm.~\ref{th:consis}) relies on the uniform convergence of $\dot{\textbf{x}}_{n}$. The adjoint equation, in contrast, has no nodal counterpart in Problem~$P_{n\lambda}$: it is encoded in the stationarity conditions~\eqref{eq:stat}, where $\bm{\lambda}_{n}$ multiplies the time derivatives of the basis functions. It is therefore recovered in weak form, by passing to the limit in~\eqref{eq:stat} and integrating by parts, which uses $\bm{\lambda}^{\infty} \in \mathcal{C}^{1}_{m_{\textbf{x}}}$ but no convergence of $\dot{\bm{\lambda}}_{n}$, and is then made pointwise by the completeness of the Fourier system (Appendix~\ref{sec:app4}).}
\begin{theorem}
\label{th:consis_d}
\revthree{Let $(\textbf{x}_{\textbf{p}},\textbf{u}_{\textbf{p}},\bm{\lambda}_{\textbf{p}},\bm{\mu}_{\textbf{p}})$ be an admissible decomposition of the optimal solution of Assumption~\ref{ass:exis_d}, let the orders satisfy $n_{B}^{2}/n_{t} \rightarrow 0$ and $n_{F}^{2}/n_{t} \rightarrow 0$ as $n \rightarrow \infty$, and let \textup{$\{(\textbf{x}_{n}(t),\allowbreak\textbf{u}_{n}(t),\allowbreak\bm{\lambda}_{n}(t),\allowbreak\bm{\mu}_{n}(t),\allowbreak\bar{\bm{\nu}}_{n})\}$} be a sequence of optimal solutions to Problem~$P_{n\lambda}$, with $\delta_{P}^{n}\pXU$ and $\delta_{D}^{n}\pXULM$ given by Eqs.~\eqref{eq:del_P_n} and~\eqref{eq:del_D}, satisfying Assumption~\ref{ass:conv_d}. Then \textup{$(\textbf{x}^{\infty}(t),\allowbreak\textbf{u}^{\infty}(t),\allowbreak\bm{\lambda}^{\infty}(t),\allowbreak\bm{\mu}^{\infty}(t),\allowbreak\bar{\bm{\nu}}^{\infty})$} is an optimal solution of Problem~$P_{\lambda}$.}
\end{theorem}
The proof is deferred to Appendix~\ref{sec:app4}. 

\begin{revthreeblock}
\begin{remark}[Covector mapping]
	\label{rem:CMT}
	Theorems~\ref{th:feas_d} and~\ref{th:consis_d} are the two directions of the covector mapping of Fig.~\ref{fig:Mapping}, the mixed-basis counterpart of~\cite[Thm.~3]{cichella_consistency_2022}. By Thm.~\ref{th:feas_d}, every feasible solution of Problem~$P_{\lambda}$ with the regularity stated there, in particular the optimal solution of Assumption~\ref{ass:exis_d}, yields, for all $n_{B}, n_{t} \in \mathbb{Z}^{+}$ and $n_{F} \geq 2$, feasible solutions of Problem~$P_{n\lambda}$ with the tolerances~\eqref{eq:del_P_n} and~\eqref{eq:del_D} computed for its decomposition, built from its mixed approximants with $w$-scaled multipliers; these tolerances vanish under the coupling of Remark~\ref{rem:coupling}. Conversely, by Thm.~\ref{th:consis_d}, under this coupling and with these tolerances computed for the optimal solution of Assumption~\ref{ass:exis_d}, whenever the scaled sequence $(\textbf{x}_{n},\textbf{u}_{n},\bm{\lambda}_{n}/w,\bm{\mu}_{n}/w,\bar{\bm{\nu}}_{n})$ of optimal solutions of Problem~$P_{n\lambda}$ converges in the sense of Assumption~\ref{ass:conv_d}, its limit is an optimal solution of Problem~$P_{\lambda}$. In this sense, discretization and dualization commute in the limit.
\end{remark}
\end{revthreeblock}

\section{Numerical Examples}
\label{sec:sim}

This section presents illustrative examples demonstrating the versatility and efficacy of the proposed Mixed Bernstein-Fourier solution.
\textcolor{black}{
The first example, a disturbance-rejection scenario, serves as a theoretical validation ground to demonstrate the method's unique capability to handle combined transient and periodic behaviors using uniform sampling.}
The second example addresses a nonlinear observer trajectory optimization problem, highlighting superior constraint satisfaction.
\textcolor{black}{
Finally, we present an engineering application, the autonomous mine countermeasures (MCM) search.}

\subsection{\textcolor{black}{Illustrative Example: Disturbance Rejection}}
\label{sec:slew_shake}

\textcolor{black}{To validate the properties of the mixed basis, we consider a disturbance rejection problem designed to test the decoupling of transient and periodic dynamics. The system is a double integrator subject to a sinusoidal disturbance field $d(t)$:}
\begin{equation}
    \textcolor{black}{\ddot{x}(t) = u(t) + d(t), \quad d(t) = A \sin(\omega t)}
\end{equation}
\textcolor{black}{with $A=5.0$ and $\omega=\pi$ rad/s, i.e., frequency of 0.5 Hz. The objective is to drive the state from rest at the origin, $x_0=0, v_0=0$, to a target state, $x_f=10, v_f=0$, over $T=10$ s while minimizing the control effort $J = \revthree{\frac{1}{2}} \int_{0}^{T} u(t)^2 dt$.}

\subsubsection{\textcolor{black}{Comparative Analysis}}
\textcolor{black}{We compared the proposed mixed Bernstein-Fourier approach against a Bernstein-only solver and a Chebyshev pseudospectral method with Clenshaw-Curtis quadrature.
To ensure a fair comparison of computational efficiency, all three formulations were implemented within the CasADi framework~\cite{Andersson2019}. This ensures that the reported differences in runtime and node counts are attributable to the mathematical properties of the approximation methods rather than implementation details of the solver.
The results, summarized in Table~\ref{tab:slew_comparison}, demonstrate the advantages of the mixed basis.}

\begin{table}[ht]
\centering
\caption{\textcolor{black}{Performance Comparison for the Disturbance Rejection Example.}}
\label{tab:slew_comparison}
\begin{tabular}{lccc}
\hline
\textbf{Method} & \textbf{Variables} & \textbf{CPU, ms} & \textbf{Max Error} \\ \hline
\textcolor{black}{Bernstein-Fourier} & \textcolor{black}{7} & \textcolor{black}{10} & \textcolor{black}{$5.5 \times 10^{-3}$} \\
\textcolor{black}{Bernstein-only} & \textcolor{black}{501} & \textcolor{black}{50} & \textcolor{black}{$1.1 \times 10^{-1}$} \\
\textcolor{black}{Pseudospectral} & \textcolor{black}{27} & \textcolor{black}{11} & \textcolor{black}{$7.5 \times 10^{-3}$} \\ \hline
\end{tabular}
\end{table}

\textcolor{black}{The Bernstein-Fourier method achieves the lowest maximum error, $5.5 \times 10^{-3}$, 
with the highest computational efficiency, utilizing only 7 optimization variables. \rev{The DC coefficient $a_{0}$ is omitted here, since the zero-mean disturbance induces no constant offset in the periodic component; the representation therefore uses $n_{B}+1=5$ Bernstein coefficients together with $\{a_{1},b_{1}\}$, for $7$ variables in total.} 
The approximation orders were selected to leverage the decoupled nature of the basis: the Bernstein order, $n_B=4$, captures the smooth, transient baseline trajectory, while 
the Fourier component, $n_F=1$ (a single harmonic), \rev{is placed at the disturbance frequency, the fifth harmonic over the horizon $[0,T]$, rather than at the fundamental mode}, matching the spectral content of the periodic disturbance. This separation allows for a parsimonious representation where $n_B$ and $n_F$ can be adjusted independently to match the distinct dynamical features of the problem. 
Critically, this high accuracy is attained using purely uniform temporal discretization, a regime in which polynomial methods typically 
suffer from the Runge phenomenon.
In comparison, the pseudospectral method, despite employing non-uniform Chebyshev nodes specifically to avoid such instabilities, requires nearly four times as many variables (27) to achieve a similar error of $7.5 \times 10^{-3}$, with comparable computation time.
The Bernstein-only approach fails to match the accuracy of the mixed method, exhibiting an error of $1.1 \times 10^{-1}$, two orders of magnitude larger, while requiring 501 variables and five times the CPU time.
These results validate that the mixed basis effectively decouples the transient and periodic dynamics, achieving \textit{spectral-like accuracy with uniform sampling} without the computational burden of high-order approximations.
The resulting optimal trajectories are shown in Fig.~\ref{fig:slew_primal}.}

\begin{figure}[ht]
	\begin{center}
		\begin{subfigure}{0.49\textwidth}
			\begin{center}
				\includegraphics[width=\linewidth]{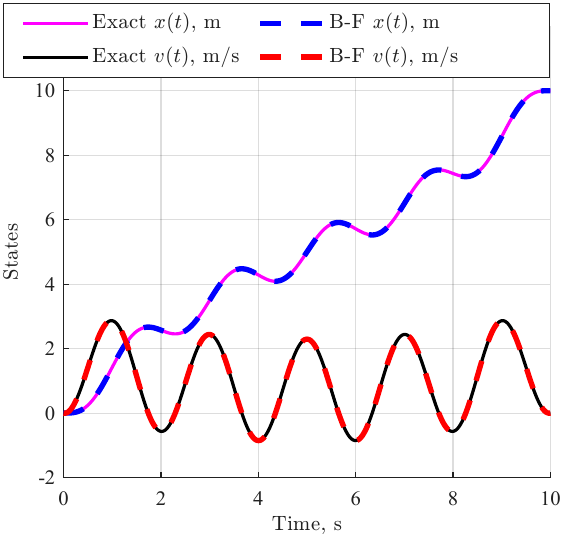}
				\caption{Optimal states.}
				\label{fig:slew_pos}
			\end{center}
		\end{subfigure}
		~
		\begin{subfigure}{0.49\textwidth}
			\begin{center}
				\includegraphics[width=\linewidth]{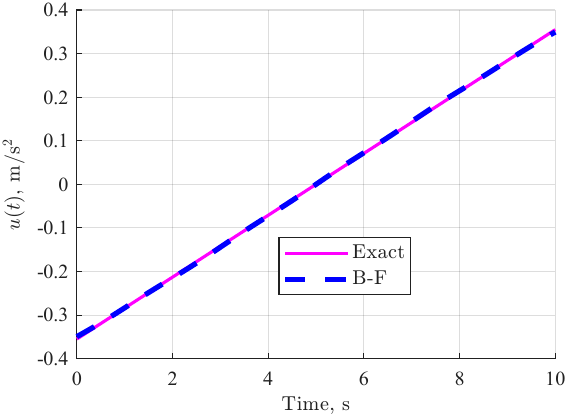}
				\caption{Optimal control effort $u(t)$.}
				\label{fig:slew_ctrl}
			\end{center}
		\end{subfigure}
		\caption{\textcolor{black}{Primal solution comparison. The mixed basis achieves smooth control reconstruction with minimal coefficients.}}
		\label{fig:slew_primal}
	\end{center}
\end{figure}

\subsubsection{\textcolor{black}{Dual Results}}
\textcolor{black}{The explicit time dependence of the disturbance $d(t)$ renders the system non-autonomous, meaning the Hamiltonian $\mathcal{H}(t)$ is not constant. Since the standard PMP necessary condition, $\mathcal{H} = \text{const}$, applies only to autonomous systems, direct verification of optimality using the covector mapping is not possible in this form. To address this, we employ a state-augmentation technique to recover an autonomous structure. We introduce an auxiliary state variable $x_{n+1}(t) = \tau(t)$ such that $\dot{\tau} = 1$ and $\tau(0)=0$. This transformation maps the original non-autonomous dynamics into an autonomous system in a higher-dimensional state space.}

According to the PMP, the augmented Hamiltonian $\mathcal{H}_{aug}$ for this autonomous system is defined as:
\begin{equation}
    \mathcal{H}_{aug} = \mathcal{H}(x, u, t) + \lambda_\tau \cdot 1
\end{equation}
where $\lambda_\tau$ is the costate associated with time. For an optimal trajectory, the autonomous Hamiltonian $\mathcal{H}_{aug}$ must be constant. This yields a Hamiltonian consistency condition:
\begin{equation}
    \lambda_\tau(t) = -\mathcal{H}(t).
\end{equation}

Figure~\ref{fig:slew_dual} presents the validation of this condition. We computed the Hamiltonian $\mathcal{H}(t)$ derived from the Bernstein-Fourier discrete solution and compared it against the exact time costate $\lambda_\tau(t)$ obtained by explicitly integrating the augmented dual dynamics $\dot{\lambda}_\tau = -\partial \mathcal{H} / \partial t$ using  MATLAB\textsuperscript{\textregistered}'s built-in \textit{bvp4c} solver. The approximated overlap between $-\mathcal{H}$ and the analytical $\lambda_\tau$ confirms that the proposed method correctly captures the dual physics and satisfies the extended covector mapping theorem, even in the presence of strong time-varying disturbances.

\begin{figure}[ht]
	\begin{center}
		\begin{subfigure}{0.46\textwidth}
			\begin{center}
				\includegraphics[width=\linewidth]{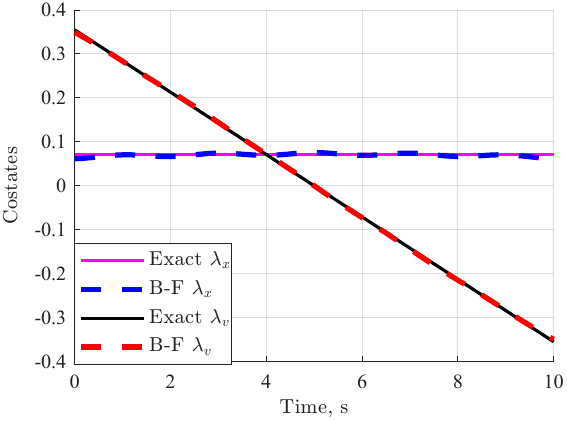}
				\caption{\textcolor{black}{Optimal costates $\lambda_x, \lambda_v$.}}
				\label{fig:slew_lam}
			\end{center}
		\end{subfigure}
		~
		\begin{subfigure}{0.49\textwidth}
			\begin{center}
				\includegraphics[width=\linewidth]{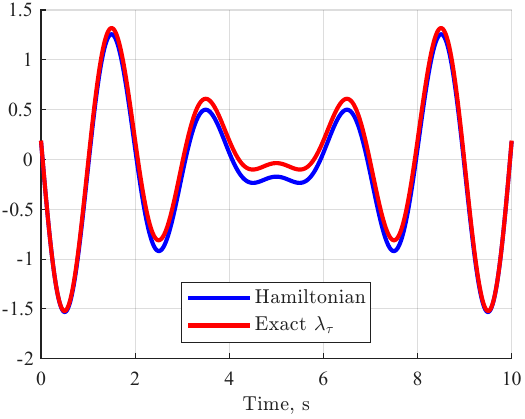}
				\caption{\textcolor{black}{Hamiltonian consistency check.}}
				\label{fig:slew_ham}
			\end{center}
		\end{subfigure}
		\caption{\textcolor{black}{Dual variable validation. The method recovers the exact costates and verifies the Hamiltonian consistency condition $\lambda_\tau \approx -\mathcal{H}(t)$.}}
		\label{fig:slew_dual}
	\end{center}
\end{figure}

\subsection{Optimization of Observer Trajectory}
\label{sec:observe}
We address the problem of finding optimal trajectories for localizing a stationary target using an unmanned aerial vehicle (UAV) equipped with bearing-only measurements in three-dimensional scenarios, as presented in~\cite{ponda_trajectory_2009}.
\subsubsection{Scenario Details}
The UAV is assumed to maintain a fixed altitude and constant speed.
Hence, it has the following equations of motion:
\begin{subequations}
	\begin{align}
			\dot{x}_{A}(t) & = V \cos \psi(t), \\
			\dot{y}_{A}(t) & = V \sin \psi(t) ,\\
			\dot{z}_{A}(t) & = 0, 
		\end{align}
	\label{eq:UAV_EOM}
\end{subequations}
where $(x_{A},y_{A},z_{A})$ is the UAV's position, $\psi$ is its heading angle, and $V$ is its speed.
The UAV is assumed to have an onboard vision system that provides the following elevation angle, $\alpha$, and azimuth angle, $\beta$, from the UAV to the target:
\begin{equation}
	\begin{bmatrix}
		\alpha_{k} \\ \beta_{k}
	\end{bmatrix} = \textbf{g}_{k} + \bm{\omega}_{k} = 
	\begin{bmatrix}
		\tan^{-1} \left( \frac{\Delta z_{k}}{\sqrt{\Delta x_{k}^{2} + \Delta y_{k}^{2}}}  \right) \\
		\tan^{-1} \left( \frac{\Delta y_{k}}{\Delta x_{k}} \right) - \psi_{k}
	\end{bmatrix} + \bm{\omega}_{k}
\end{equation}
where $(\Delta x_{k},\Delta y_{k},\Delta z_{k})$ is the position of the target, which is assumed to be time-invariant, relative to the UAV at $t_{k}$, $\{\bm{\omega}_{k}\}_{k=1}^{f}$ is a zero-mean white Gaussian noise sequence with a known constant covariance matrix $\Sigma$, and $t_{f}$ is the given final time of the scenario.

The main objective of this example is to generate a trajectory that maximizes the information about the target's location obtained from the bearing-only measurements.
We also add a 100 m no-fly zone around the target to reduce the detection probability of the UAV. 
Note that this constraint \textcolor{black}{is a pure state constraint, which} violates Assumption~\ref{ass:state}. \textcolor{black}{While this prevents applying the covector mapping (which assumes mixed constraints), numerical results demonstrate that the proposed method effectively handles such constraints via NLP relaxation.} Thus, we \textcolor{black}{defer the validation of the dual variables to the other examples and focus here on the primal solution quality.}
Selecting an objective function to maximize the information gathered along a trajectory is nontrivial because it is difficult to quantify this information.
We follow the practice of~\cite{ponda_trajectory_2009}, therefore, and use the trace of the Cramér–Rao lower bound (CRLB) to improve the data collected by the UAV regarding the target.
The CRLB is computed as the inverse of the celebrated Fisher information matrix (FIM), which satisfies in this case:
\begin{equation}
	J_{k} = \sum_{k=0}^{f} (G_{k})^{T} \Sigma^{-1} G_{k}.
\end{equation}
where $G_{k}$ is the Jacobian matrix of $\textbf{g}_{k}$.

In addition to the no-fly zone described above, the other scenario parameters are as follows.
Its duration is set to be 60 s.
The UAV starts its motion at $\begin{bmatrix}0 & 0 & 20 \end{bmatrix}^{T}$ m with a constant speed of 30 m/s, heading towards the x-axis.
The stationary target is located at $\begin{bmatrix}400 & 0 & 0 \end{bmatrix}^{T}$~m.
The UAV's vision system measures the angles to detected objects within its field of view at a frequency of 10 Hz, which yields 600 equally spaced time steps for this scenario.
\textcolor{black}{This fixed hardware sampling rate motivates the use of the mixed Bernstein-Fourier approach, which inherently supports uniform time discretization. Pseudospectral methods, while efficient, rely on non-uniform node distributions, e.g., Legendre-Gauss-Lobatto (LGL), that require post-processing interpolation to match the sensor's 10 Hz sampling rate, thereby introducing potential approximation errors.}
The measurement noise covariance matrices are assumed to be time-invariant, satisfying
$\Sigma = \text{diag} \{ (3\cdot10^{-3})^{2} , (3\cdot10^{-3})^{2} \}$.
\subsubsection{Results}
We now present the numerical solutions to the above problem.
We use MATLAB\textsuperscript{\textregistered}'s built-in \textit{fmincon} function as an off-the-shelf numerical solver to solve the NLP problem.
We set $n_{B} = 70$ and $n_{F} = 12$, i.e., we use $70^{th}$ order Bernstein polynomials and $12^{th}$ order Fourier series to construct candidate UAV trajectories using $n_{t} = n_{B} = 70$. 
Because we need to determine the horizontal position, $(x_{A},y_{A})$, and the heading of the UAV at each time step, the state vector has three entries. Therefore, the decision vector has 288 decision variables (71 from the Bernstein part and 25 from the Fourier part of the approximation, for each of the three states) to be determined by the solver.
For comparison, we also used a regular Bernstein-based solver, as presented in~\cite{cichella_optimal_2021}.
The current study indicates that Bernstein approximants of at least $282^{nd}$ order are required for the \textit{fmincon} function to converge when the number of function evaluations is limited to a million. 
Note that this is almost three times more than the number of variables used by the proposed mixed approximation, as the total length of the decision vector is 846\rev{: the three coefficients fixed by the given initial conditions through the Bernstein endpoint-value property are known a priori and are therefore excluded from the count. For the mixed approximant, in contrast, the initial condition couples all coefficients of each state linearly and fixes no single coefficient a priori, so all 288 coefficients remain decision variables}.
\textcolor{black}{In terms of computation time, the mixed method converged in approximately 20 s when initialized with zero Fourier coefficients (``cold start''), demonstrating the solver's ability to identify periodic structures from scratch. When initialized with the regulated LS solution (``warm start''), the time dropped to 3.5 s. This compares favorably to the Bernstein-only baseline, which took 8 s for the lower-order (but infeasible) solution, and 160 s for a high-order ($n_B=1000$) solution that still violated the constraint by 4 m.}
Additionally, Bernstein-only approximants achieve safety constraints only asymptotically. \textcolor{black}{When the strict 100~m no-fly zone was applied to the Bernstein-only solver, the optimization failed to converge to a feasible solution (maximum violation $\approx 12$ m with $n_B=282$).} Consequently, \textcolor{black}{to allow for a valid cost comparison,} we extended the no-flight zone radius to 112~m for these approximants.

We initialize both solutions using the same Bernstein coefficients, meaning that the Fourier part of the mixed solution is inactive. 
It corresponds to a path that directly flies toward the target until it reaches the no-flight zone's boundary and then encircles the target's position.
\textcolor{black}{Figure~\ref{fig:Traj} presents a top view of the optimal trajectories. The black line denotes the mixed Bernstein-Fourier approximation. The green line denotes the Bernstein-only approximation, which required three times more decision variables and resulted in a $7\%$ performance degradation. The black line illustrates the consistency of the high-order mixed approximation. The red asterisk and dashed line denote the target and no-fly zone, respectively.}
For all trajectories, the UAV flies toward the target's location and starts orbiting the no-flight zone, similar to the initial guess. 
However, when using the mixed Bernstein-Fourier approximation, the UAV is able to follow the no-fly zone boundary much closer.
Moreover, the mixed Bernstein-Fourier approximation produces a value of $6.76\cdot10^{-4}$ for the trace of the CRLB, while the Bernstein-based approximation produces a value of $7.23\cdot10^{-4}$, a nearly $7\%$ performance degradation despite using three times more decision variables.
Finally, to show the consistency of the mixed Bernstein-Fourier approximation given by Th.~\ref{th:consis}, we include a solution based on a significantly higher order of approximation, with $n_{B} = 600$ and $n_{F} = 300$.
For this case, the mixed Bernstein-Fourier approximation produces a value of $6.64\cdot10^{-4}$ for the trace of the CRLB, a performance improvement of only $1.8\%$ compared to the lower order mixed approximation.
\begin{figure}[htb]
	\centering
	\includegraphics[width=3.25in]{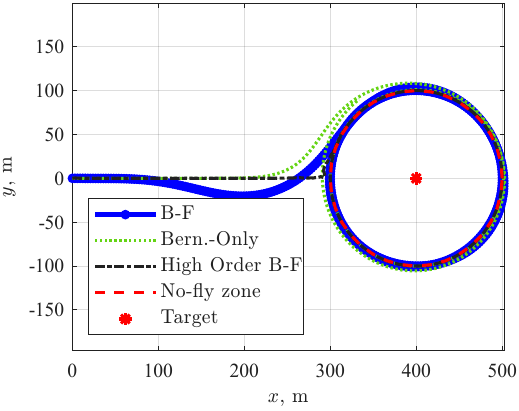}
	\caption{\textcolor{black}{Top-down view of optimal UAV trajectories using different numerical approximations.}}
	\label{fig:Traj}
\end{figure}

Fig.~\ref{fig:Pos} presents the optimal solution for the horizontal positioning of the UAV using the mixed Bernstein-Fourier basis functions.
The optimal trajectory consists of both periodic and non-periodic components. In the first ten seconds, there is a transient response which is non-periodic. After ten seconds, however,  a periodic behavior emerges, exhibiting circular motion with period $T = 2\pi r / v = 20.94$ s.
This periodic behavior matches the known results in~\cite{ponda_trajectory_2009} for target localization using bearing-only measurements. Note that in~\cite{hung_range-based_2020}, similar periodic behaviors appeared for moving target localization using range-only measurements.
\begin{figure}[htb]
	\centering
	\includegraphics[width=3.25in]{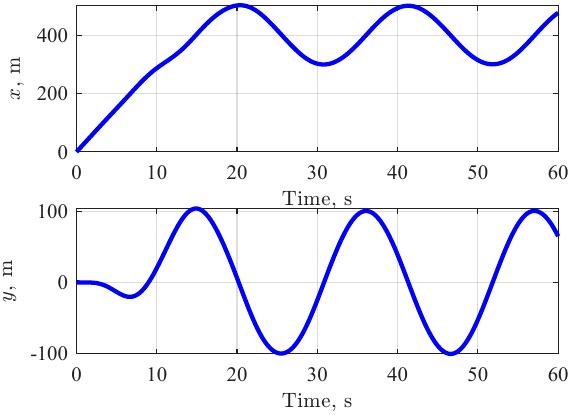}
	\caption{Optimal UAV motion in the x and y coordinate directions.}
	\label{fig:Pos}
\end{figure}
\subsection{\textcolor{black}{Autonomous Mine Countermeasures (MCM) Search}}
\label{sec:AUV}
\textcolor{black}{
To demonstrate the practical utility of the approach for periodic robotic missions, we consider an optimal motion planning problem for an Autonomous Underwater Vehicle (AUV) conducting an MCM search~\cite{kragelund_generalized_2021}.}

\subsubsection{\textcolor{black}{Problem Formulation}}
\textcolor{black}{
This scenario represents a time-critical engineering problem governed by strict hardware and operational constraints. Unlike traditional MCM operations that utilize active sonar, which can alert hostile forces or trigger acoustic mines, this mission requires a covert operation using a passive optical sensor, e.g., a high-resolution downward-facing camera.}

\textcolor{black}{
This sensor choice imposes a critical discretization constraint: optical systems operate at fixed frame rates. To accurately model the probability of detection, $P_{det}$, the optimization nodes must be aligned with the physical exposure times. \rev{Non-uniform node distributions, such as the LGL nodes of pseudospectral methods, cluster near the interval endpoints,} misaligning with the camera frames and necessitating post-processing interpolation that degrades the fidelity of the detection integral. The proposed Bernstein-Fourier method naturally supports the required uniform temporal sampling.}

\textcolor{black}{The AUV is modeled as a planar kinematic system:
\begin{subequations}
	\begin{align}
		\dot{x}(t) &= v(t) \cos \psi(t), \\
		\dot{y}(t) &= v(t) \sin \psi(t),
	\end{align}
	\label{eq:AUV_EOM}
\end{subequations}
where the states $(x, y)$ are the position, and the control inputs are the speed, $v$, and heading angle, $\psi$.
The speed is bounded by $v \in [v_{\min}, v_{\max}]$.
To maintain feasible turn commands, the turn rate is derived and constrained as a path constraint:
\begin{equation}
	\dot{\psi} = \frac{\dot{x}\ddot{y} - \dot{y}\ddot{x}}{v^2}, \quad \dot{\psi}^2 \leq u_{max}^2.
	\label{eq:AUV_turnrate}
\end{equation}
Unlike the optimal observer trajectory example, this turn rate constraint involves both states and control inputs, satisfying Assumption~\ref{ass:state}\rev{. Equivalently, treating the heading $\psi$ as a state and the turn rate $\omega = \dot\psi$ as a control casts Eq.~\eqref{eq:AUV_turnrate} as a mixed constraint $\omega^{2} - u_{\max}^{2} \le 0$ that depends on the control, so no pure-state constraint arises. The constraint is enforced pointwise at the sample nodes of Problem~$P_{n}$, so no rational representation of the mixed approximant is ever formed; since $v \geq v_{\min} > 0$ on the feasible set, Eq.~\eqref{eq:AUV_turnrate} is Lipschitz there, as the analysis of Secs.~\ref{sec:primal}--\ref{sec:dual} requires.}}
\textcolor{black}{The objective is to maximize coverage while regulating speed:
\begin{equation}
    J = \int_{0}^{t_f} \big[ - D(x,y) +   \rho_v (v - V_{nom})^2 \big] \, dt,
    \label{eq:AUV_cost}
\end{equation}
where the instantaneous detection rate is $D(x,y) = \sum_{j=1}^{M} w_j \gamma_j$, with $\gamma_j = \lambda_0 \exp(-d_j^2 / 2\sigma^2)$ representing the Gaussian sensor response at cell $j$, $d_j$ being the distance from the AUV to cell center $j$, $\sigma$ the sensor standard deviation, and $\lambda_0$ the base detection rate. $\rho_v$ is the speed penalty coefficient.}

\subsubsection{\textcolor{black}{Results and Validation}}
\textcolor{black}{The scenario considers a $100 \times 100$ m search area discretized into a $20 \times 20$ grid ($M = 400$ cells) with uniform prior probability. The AUV operates at nominal speed $V_{nom} = 2$ m/s with bounds $[1.8, 2.2]$ m/s and maximum turn rate $u_{max} = 0.5$ rad/s. Sensor parameters are $\sigma = 10$ m, $\alpha = 0.9$, and $\lambda_0 = 0.3$. The speed penalty coefficient is $\rho_v = 0.01$. \rev{The states $(x, y)$ are each parameterized with $n_B = 40$ Bernstein and $n_F = 10$ Fourier modes, so that each coordinate carries 62 coefficients and the decision vector has 124 entries in total.}}

\textcolor{black}{
We use MATLAB\textsuperscript{\textregistered}'s built-in \textit{fmincon} function as an off-the-shelf numerical solver to solve the NLP problem.
The optimizer is initialized with a lawnmower pattern, a standard coverage trajectory consisting of parallel survey legs connected by 180° turns. However, because the scenario is limited in duration, the lawnmower pattern cannot complete a thorough search.
This pattern exhibits inherent periodicity in the cross-track direction, making it well-suited for the mixed Bernstein-Fourier parameterization: the Fourier modes capture the oscillatory sweeping motion while the Bernstein polynomials handle the along-track progression.}

\textcolor{black}{Figure~\ref{fig:mcm_trajectory} compares the initial lawnmower pattern, yielding a 56.2\% coverage, with the optimized trajectory, yielding 82.7\% coverage, representing a 26.5 percentage point improvement. The optimizer reshapes the path to maximize time spent in high-value regions while maintaining kinematic feasibility: speed remains within 5\% of nominal (1.94--2.10 m/s) and the turn rate constraint is satisfied throughout.}
\rev{For comparison, a Bernstein-only parameterization initialized with the same lawnmower pattern, with orders $n_B = 200$ and $n_t = n_B + 1$ collocation nodes~\cite{cichella_optimal_2021}, attains the same coverage of 82.7\%. The mixed basis thus achieves this coverage at a fivefold lower approximation order ($n_B = 40$ versus $200$), with native uniform sampling matched to the sensor frame rate.}

\begin{figure}[ht]
	\begin{center}
		\textcolor{black}{
		\begin{subfigure}{0.49\textwidth}
			\begin{center}
				\includegraphics[width=\linewidth]{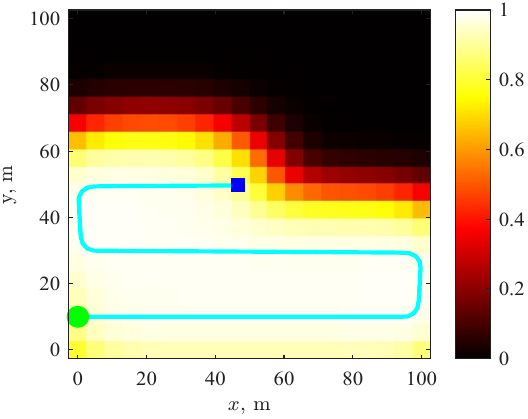}
				\caption{\textcolor{black}{Initial lawnmower (56.2\% coverage).}}
				\label{fig:mcm_initial}
			\end{center}
		\end{subfigure}
		~
		\begin{subfigure}{0.49\textwidth}
			\begin{center}
				\includegraphics[width=\linewidth]{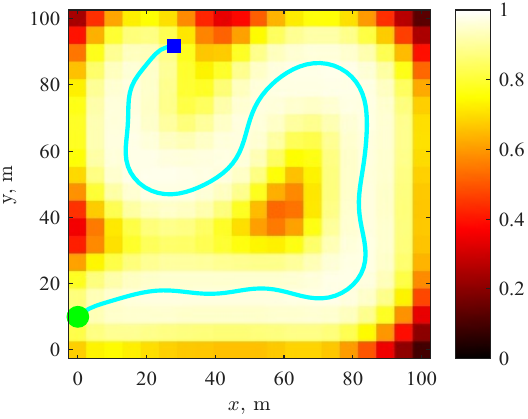}
				\caption{\textcolor{black}{Optimized trajectory (82.7\% coverage).}}
				\label{fig:mcm_optimized}
			\end{center}
		\end{subfigure}
		}
		\caption{\textcolor{black}{MCM coverage comparison with detection probability shown by color.}}
		\label{fig:mcm_trajectory}
	\end{center}
\end{figure}

The Lagrangian of the Hamiltonian for this system is:
\begin{equation}
	\mathcal{L} = -D(x,y) + \rho_v(v - V_{nom})^2 + \lambda_x v\cos\psi + \lambda_y v\sin\psi + \mu(\dot{\psi}^2 - u_{max}^2)
	\label{eq:AUV_hamiltonian}
\end{equation}
where $\lambda_x$ and $\lambda_y$ are the costates. The adjoint residuals are:
\begin{equation}
	r_{\lambda_{x}} = \dot{\lambda}_x - \frac{\partial D}{\partial x}, \quad 
    r_{\lambda_{y}} = \dot{\lambda}_y - \frac{\partial D}{\partial y},
	\label{eq:AUV_adjoint}
\end{equation}
where $\lambda_x(t_f) = \lambda_y(t_f) = 0$.
The stationarity conditions are $\partial \mathcal{L} / \partial v = 0$, and $\partial \mathcal{L} / \partial \psi = 0$.
Table~\ref{tab:mcm_cmt} summarizes \rev{a direct check of} the first-order necessary conditions \rev{(stationarity and adjoint residuals) for the computed solution}. 
Figure~\ref{fig:mcm_cmt} presents these residuals over time and confirms the satisfaction of the necessary conditions.

\begin{table}[ht]
\centering
\caption{CMT verification metrics for the MCM trajectory.}
\label{tab:mcm_cmt}
\begin{tabular}{lcc}
\hline
\textbf{Condition} & \textbf{Metric} & \textbf{Value} \\ \hline
Stationarity $\partial \mathcal{L} / \partial v$ & RMS & 0.0048 \\
Stationarity $\partial \mathcal{L} / \partial \psi$ & RMS & 0.0075 \\
Adjoint equations $(\lambda_x, \lambda_y)$ & $\| r_\lambda \|$ RMS & $1.3 \times 10^{-5}$ \\
\hline
\end{tabular}
\end{table}

\begin{figure}[htbp]
    \centering
    \textcolor{black}{\includegraphics[width=3.25in]{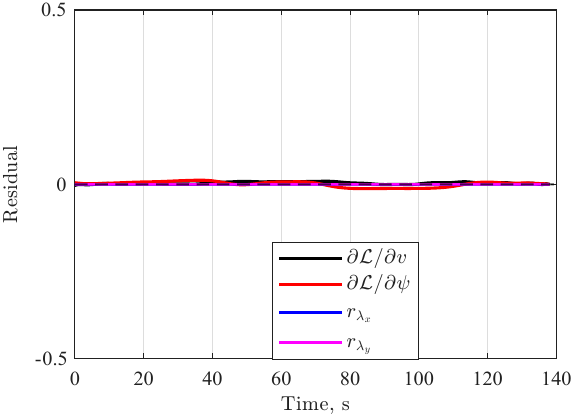}}
    \caption{\textcolor{black}{Stationarity and adjoint residuals for MCM verification.}}
    \label{fig:mcm_cmt}
\end{figure}
\section{Conclusions}
\label{sec:concl}
This paper introduced a mixed Bernstein-Fourier approximation methodology that combines Bernstein polynomials and Fourier series for optimal trajectory generation. 
We established theoretical results, including explicit proofs of uniform convergence for approximated functions, their derivatives, and integrals, accompanied by detailed error bound analyses demonstrating the robustness and accuracy of the method. 
A regulated least squares formulation was developed to determine approximation coefficients, enhancing practical flexibility and numerical stability.

Beyond these theoretical contributions, we applied the mixed Bernstein-Fourier approach to optimal motion planning problems, demonstrating significant advantages over traditional methods for systems exhibiting periodic behavior. 
By proving feasibility and consistency of the approximated optimal control solutions, we established theoretical guarantees on accuracy. 
We further extended the covector mapping theorem to this setting, providing rigorous assurances for dual variable approximations crucial for validating optimality conditions derived from PMP.

\textcolor{black}{Numerical examples reinforced these theoretical insights. The disturbance rejection example demonstrated that the mixed approach achieves spectral-like accuracy with purely uniform temporal discretization, a regime in which pseudospectral methods require non-uniform nodes to avoid the Runge phenomenon. The observer trajectory and MCM search examples illustrated effective handling of periodic dynamics and constraints, achieving high precision with significantly fewer variables than Bernstein-only approximations, while CMT verification confirmed satisfaction of the necessary optimality conditions.}

\textcolor{black}{Future work will extend this framework to handle unknown or time-varying periods by incorporating frequency-scaling parameters directly into the optimization, thereby mitigating the dependence on the fixed time horizon. Additionally, while this work focused on uniform sampling for onboard implementation, the least-squares formulation (Sec.~\ref{sec:LS}) suggests inherent robustness to sensor noise and missing data, which warrants further investigation in stochastic scenarios.}
\section*{Acknowledgment}
The authors would like to thank the Office of Naval Research Science of Autonomy Program under Grant No.\ N0001425GI01545 and Consortium for Robotics \revthree{and} Unmanned Systems Education and Research at the Naval Postgraduate School for sponsoring this research.
This project was also supported in part by an appointment to the NRC Research Associateship Program at the Naval Postgraduate School, administered by the Fellowships Office of the National Academies of Sciences, Engineering, and Medicine.

\appendix
\section{Proof of \texorpdfstring{\revthree{Thm.}}{Thm.}~\ref{th:feas}}
\label{sec:app1}
\begin{proof}
	Our objective is to show that there exist Bernstein and Fourier coefficients such that the constraints of Problem $P_{n}$ are satisfied with $\delta_{P}^{\revthree{n}}\revthree{\pXU}$, where there exist $\textbf{x}(t)$ and $\textbf{u}(t)$ which comprise a feasible solution to Problem $P$ \revthree{as in Thm.~\ref{th:feas}, for instance the optimal solution of Assumption~\ref{ass:exis}; the construction uses only its feasibility and its regularity, $\textbf{x}\in\mathcal{C}^{1}_{m_{\textbf{x}}}$ and $\textbf{u}\in\mathcal{C}^{0}_{m_{\textbf{u}}}$, and not its optimality}.
	Define equispaced time nodes $t_{k} \dfn \frac{k}{n_{t}}$ such that $\textbf{x}_{k}$ and $\textbf{u}_{k}$ are defined according to \eqref{eq:x_app} and \eqref{eq:u_app}, respectively, whose construction is readily achieved using the coefficients presented in Eqs.~\eqref{eq:g_k}--\eqref{eq:b_k}.
	Similarly to~\eqref{eq:per}, \revthree{this feasible solution} can be presented as:
	\begin{align}
		\label{eq:decomp_xu}
		\textbf{x}(t)  = \textbf{x}_{\textbf{g}}(t) + \textbf{x}_{\textbf{p}}(t), \quad
		\textbf{u}(t)  = \textbf{u}_{\textbf{g}}(t) + \textbf{u}_{\textbf{p}}(t)
	\end{align}
	where the subscript $\textbf{p}$ denotes the periodic part inside the function, and $\textbf{g}$ denotes the nonperiodic part of it.
	
	We first show that the dynamic constraint is satisfied. We add the nulled expression $\textbf{f}(\textbf{x}(t_{k}),\textbf{u}(t_{k})) - \dot{\textbf{x}}(t_{k})$ from Eq.~\eqref{eq:dyn_P} to get that
	\begin{align}
		\label{eq:proof_dyn_P}
		&\Vert\dot{\textbf{x}}_{k} - \textbf{f}(\textbf{x}_{k},\textbf{u}_{k})\Vert \leq
		\Vert\dot{\textbf{x}}_{k} - \dot{\textbf{x}}(t_{k}) \Vert + \Vert \textbf{f}(\textbf{x}(t_{k}),\textbf{u}(t_{k})) -  \textbf{f}(\textbf{x}_{k},\textbf{u}_{k})\Vert \leq \notag \\		&C_{1}W_{\dot{\textbf{x}}_{\textbf{g}}}\Big(\frac{1}{\sqrt{n_{B}}}\Big)
		+  \frac{A_{1}\log n_{F}}{n_{F}^{r-1}} W_{{\textbf{x}}_{\textbf{p}}^{(r)}}\Big(\rev{\frac{1}{n_{F}}}\Big) + L_{\textbf{f}_{\textbf{x}}}\left[ C_{0}W_{\textbf{x}_{\textbf{g}}}\Big(\frac{1}{\sqrt{n_{B}}}\Big) \rev{+} \frac{A_{0}\log n_{F}}{n_{F}^{r}} W_{\textbf{x}_{\textbf{p}}^{(r)}}\Big(\rev{\frac{1}{n_{F}}}\Big) \right] \notag\\
        & + L_{\textbf{f}_{\textbf{u}}}\left[ \bar{C}_{0}W_{\textbf{u}_{\textbf{g}}}\Big(\frac{1}{\sqrt{n_{B}}}\Big) \rev{+} \frac{\bar{A}_{0}\log n_{F}}{n_{F}^{\rev{\tilde{r}}}} W_{\textbf{u}_{\textbf{p}}^{(\rev{\tilde{r}})}}\Big(\rev{\frac{1}{n_{F}}}\Big)\right] = \delta_{\dot{\textbf{x}}}^{n}\revthree{\pX} + L_{\textbf{f}_{\textbf{x}}} \delta_{\textbf{x}}^{n}\revthree{\pX} + L_{\textbf{f}_{\textbf{u}}} \delta_{\textbf{u}}^{n}\revthree{\pU}
	\end{align}
	where $A_{0}$, $\bar{A}_{0}$, $A_{1}$, $C_{0}$, $\bar{C}_{0}$, and $C_{1}$ are positive constants independent of $n_{B}$, $n_{F}$, and $n_{t}$; 
	$W_{\textbf{x}_{\textbf{g}}}(\cdot)$, $W_{\dot{\textbf{x}}_{\textbf{g}}}(\cdot)$, $\textcolor{black}{W_{\textbf{x}_{\textbf{p}}^{(r)}}(\cdot)}$,  $W_{\textbf{u}_{\textbf{g}}}(\cdot)$, and $\textcolor{black}{W_{\textbf{u}_{\textbf{p}}^{(\revthree{\tilde{r}})}}(\cdot)}$,   are moduli of continuity for $\textbf{x}_{\textbf{g}}(t)$, $\dot{\textbf{x}}_{\textbf{g}}(t)$, $\textcolor{black}{\textbf{x}_{\textbf{p}}^{(r)}(t)}$, $\textbf{u}_{\textbf{g}}(t)$, and $\textcolor{black}{\textbf{u}_{\textbf{p}}^{(\revthree{\tilde{r}})}(t)}$,  respectively; and
	$L_{\textbf{f}_{\textbf{x}}}$ and $L_{\textbf{f}_{\textbf{u}}}$ are the Lipschitz constants of $\textbf{f}$ with respect to $\textbf{x}$ and $\textbf{u}$, respectively.
	We set $C_{P} > (1 + L_{\textbf{f}_{\textbf{x}}} + L_{\textbf{f}_{\textbf{u}}})$ such that Eq.~\revthree{\eqref{eq:dyn_Pn}} holds.
	The second inequality in~\eqref{eq:proof_dyn_P} holds using Thm.~\ref{th:FB_conv1} and Corollary~\ref{cor:diff}.
	\rev{Under the hypotheses of Thm.~\ref{th:feas}, $W_{\textbf{x}_{\textbf{g}}}(1/\sqrt{n_{B}})$, $W_{\dot{\textbf{x}}_{\textbf{g}}}(1/\sqrt{n_{B}})$, and $W_{\textbf{u}_{\textbf{g}}}(1/\sqrt{n_{B}})$ vanish by continuity of $\textbf{x}_{\textbf{g}}$, $\dot{\textbf{x}}_{\textbf{g}}$, and $\textbf{u}_{\textbf{g}}$; $W_{\textbf{x}_{\textbf{p}}^{(r)}}(1/n_{F})$ and $W_{\textbf{u}_{\textbf{p}}^{(\tilde{r})}}(1/n_{F})$ vanish by continuity of $\textbf{x}_{\textbf{p}}^{(r)}$ and $\textbf{u}_{\textbf{p}}^{(\tilde{r})}$; and $\log n_{F}/n_{F}^{r-1} \rightarrow 0$ since $r \geq 2$. 
    Hence every component of Eq.~\eqref{eq:del_P_n} vanishes as $n = \min(n_{t},n_{B},n_{F}) \rightarrow \infty$.}
	
	Similarly, the inequality constraint satisfies
	\begin{align}
		\Vert\textbf{h}(\textbf{x}_n(t_{k}),\textbf{u}_n(t_k)) - \textbf{h}(\textbf{x}(t_{k}),\textbf{u}(t_{k}))\Vert \leq  L_{\textbf{h}_{\textbf{x}}}\delta_{\textbf{x}}^{n}\revthree{\pX} + L_{\textbf{h}_{\textbf{u}}}\delta_{\textbf{u}}^{n}\revthree{\pU}
	\end{align}
	where $L_{\textbf{h}_{\textbf{x}}}$ and $L_{\textbf{h}_{\textbf{u}}}$ are the Lipschitz constant of $\textbf{h}$ with respect to $\textbf{x}$ and $\textbf{u}$, respectively.
	Hence,
	\begin{align}
		\textbf{h}(\textbf{x}_n(t_{k}),\textbf{u}_n(t_k))  &\leq \textbf{h}(\textbf{x}(t_{k}),\textbf{u}(t_{k})) +  (L_{\textbf{h}_{\textbf{x}}}\delta_{\textbf{x}}^{n}\revthree{\pX} + L_{\textbf{h}_{\textbf{u}}}\delta_{\textbf{u}}^{n}\revthree{\pU})\cdot \textbf{1} \notag \\
		&\leq  (L_{\textbf{h}_{\textbf{x}}}\delta_{\textbf{x}}^{n}\revthree{\pX} + L_{\textbf{h}_{\textbf{u}}}\delta_{\textbf{u}}^{n}\revthree{\pU}) \cdot \textbf{1}
	\end{align}
	where the second inequality holds by using~\eqref{eq:in_P} and by setting $C_{P} > L_{\textbf{h}_{\textbf{x}}} + L_{\textbf{h}_{\textbf{u}}}$.
	
	Similarly, for the equality constraint, we add the nulled expression $\textbf{e}(\textbf{x}(0),\textbf{x}(1))$ from Eq.~\eqref{eq:eq_P} to the constraints in Eq.~\eqref{eq:eq_Pn} to get that
	\begin{align}
		\Vert\textbf{e}(\textbf{x}_n(0),\textbf{x}_n(1)) - \textbf{e}(\textbf{x}(0),\textbf{x}(1))\Vert &\leq \revthree{L_{\textbf{e}}\big(\Vert\textbf{x}_{n}(0)-\textbf{x}(0)\Vert + \Vert\textbf{x}_{n}(1)-\textbf{x}(1)\Vert\big)} \notag \\
		&\leq \revthree{2}L_{\textbf{e}} \left[ C_{0}W_{\textbf{x}_{\textbf{g}}}\Big(\frac{1}{\sqrt{n_{B}}}\Big) +  \frac{A_{0}\log n_{F}}{n_{F}^{r}} W_{\textbf{x}_{\textbf{p}}^{(r)}}\Big(\rev{\frac{1}{n_{F}}}\Big) \right] \notag \\
		&= \revthree{2}L_{\textbf{e}} \delta_{\textbf{x}}^{n}\revthree{\pX}
	\end{align}
	where $L_{\textbf{e}}$ is the Lipschitz constant of $\textbf{e}$ \revthree{with respect to each of its arguments}. 
	Then, the constraint holds as we set $C_{P} > \revthree{2}L_{\textbf{e}}$. \revthree{Hence all constraints of Problem~$P_{n}$ hold for every $C_{P} > \max\{1 + L_{\textbf{f}_{\textbf{x}}} + L_{\textbf{f}_{\textbf{u}}},\, L_{\textbf{h}_{\textbf{x}}} + L_{\textbf{h}_{\textbf{u}}},\, 2L_{\textbf{e}}\}$, a threshold independent of $n_{B}$, $n_{F}$, and $n_{t}$.}
\end{proof}
\section{Proof of \texorpdfstring{\revthree{Thm.}}{Thm.}~\ref{th:consis}}
\label{sec:app2}
\begin{proof}
	The proof is divided into three steps: 
	1) we show that $(\textbf{x}^{\infty}(t),\textbf{u}^{\infty}(t))$ is a feasible solution to Problem~$P$;
	2) we prove that 
	\begin{equation}
		\label{eq:I_con}
		\lim\limits_{n\rightarrow\infty } J_{n} (\textbf{x}_{n}(t),\textbf{u}_{n}(t)) = J(\textbf{x}^{\infty}(t),\textbf{u}^{\infty}(t));
	\end{equation}
	and 3) we show that $J(\textbf{x}^{\infty}(t),\textbf{u}^{\infty}(t)) = J(\textbf{x}^{*}(t),\textbf{u}^{*}(t))$.
	\paragraph*{Step 1)} 
	If $(\textbf{x}^{\infty}(t),\textbf{u}^{\infty}(t))$ satisfy the constraints of Problem~$P$, then it is a feasible solution.
	First, we show it satisfies the dynamic constraint, that is,
	\begin{equation}
		\dot{\textbf{x}}^{\infty}(t) - \textbf{f}(\textbf{x}^{\infty}(t),\textbf{u}^{\infty}(t)) = \textbf{0}.
	\end{equation}
	\begin{revblock}Fix $t^{\prime} \in [0,1]$ arbitrary and let $t_{k} = \frac{k}{n_{t}}$, $k=0,\dots,n_{t}$, denote the constraint nodes of Problem~$P_{n}$. For each $n$, let $t_{k(n)}$ be a node nearest $t^{\prime}$, so that $\vert t_{k(n)} - t^{\prime}\vert \leq \frac{1}{2n_{t}} \rightarrow 0$ as $n \rightarrow \infty$. By the triangle inequality,
	\begin{align}
		\Vert \dot{\textbf{x}}^{\infty}(t^{\prime}) &- \textbf{f}(\textbf{x}^{\infty}(t^{\prime}),\textbf{u}^{\infty}(t^{\prime})) \Vert \leq{} \Vert \dot{\textbf{x}}^{\infty}(t^{\prime}) - \dot{\textbf{x}}^{\infty}(t_{k(n)}) \Vert + \Vert \dot{\textbf{x}}^{\infty}(t_{k(n)}) - \dot{\textbf{x}}_{n}(t_{k(n)}) \Vert  \notag\\
        &+\Vert \dot{\textbf{x}}_{n}(t_{k(n)}) - \textbf{f}(\textbf{x}_{n}(t_{k(n)}),\textbf{u}_{n}(t_{k(n)})) \Vert + \Vert \textbf{f}(\textbf{x}_{n}(t_{k(n)}),\textbf{u}_{n}(t_{k(n)})) - \textbf{f}(\textbf{x}^{\infty}(t_{k(n)}),\textbf{u}^{\infty}(t_{k(n)})) \Vert \notag \\
		&+ \Vert \textbf{f}(\textbf{x}^{\infty}(t_{k(n)}),\textbf{u}^{\infty}(t_{k(n)})) - \textbf{f}(\textbf{x}^{\infty}(t^{\prime}),\textbf{u}^{\infty}(t^{\prime})) \Vert.
		\label{eq:five_term}
	\end{align}
	As $n \rightarrow \infty$: the first term vanishes by continuity of $\dot{\textbf{x}}^{\infty}$ and $t_{k(n)} \rightarrow t^{\prime}$; the second is bounded by $\max_{t\in[0,1]}\Vert \dot{\textbf{x}}_{n}(t) - \dot{\textbf{x}}^{\infty}(t)\Vert$ and vanishes by Assumption~\ref{ass:conv}; the third is bounded by $\delta_{P}^{n}\revthree{\pXU}$ by Eq.~\eqref{eq:dyn_Pn}, since $t_{k(n)}$ is a constraint node, and $\delta_{P}^{n}\revthree{\pXU} \rightarrow 0$ by Thm.~\ref{th:feas}; the fourth is bounded by $L_{\textbf{f}_{\textbf{x}}} \max_{t\in[0,1]}\Vert \textbf{x}_{n} - \textbf{x}^{\infty}\Vert + L_{\textbf{f}_{\textbf{u}}} \max_{t\in[0,1]}\Vert \textbf{u}_{n} - \textbf{u}^{\infty}\Vert$ by Assumption~\ref{ass:lip} and vanishes by Assumption~\ref{ass:conv}; the fifth vanishes by Lipschitz continuity of $\textbf{f}$ (Assumption~\ref{ass:lip}) together with continuity of $\textbf{x}^{\infty}$ and $\textbf{u}^{\infty}$ at $t^{\prime}$. The left-hand side is independent of $n$; hence it equals zero. Since $t^{\prime}$ was arbitrary, $(\textbf{x}^{\infty}(t),\textbf{u}^{\infty}(t))$ satisfy the dynamic constraint.\end{revblock}
	
	\rev{The inequality constraint in~\eqref{eq:in_P} follows from the same five-term estimate applied componentwise, with Eq.~\eqref{eq:in_Pn} in place of Eq.~\eqref{eq:dyn_Pn}.}
	
	The equality constraint in~\eqref{eq:eq_P} we get by using Assumption~\ref{ass:conv} and Eq.~\eqref{eq:eq_Pn} that
	\begin{equation}
		\textbf{e}(\textbf{x}^{\infty}(0),\textbf{x}^{\infty}(1)) = 
		\lim\limits_{n\rightarrow\infty }
		\textbf{e}(\textbf{x}_{n}(0),\textbf{x}_{n}(1)) = \textbf{0}
	\end{equation}
	as $\lim_{n\rightarrow\infty} \delta_{P}^{n}\revthree{\pXU} = 0$.
	\paragraph*{Step 2)} 
	To show that~\eqref{eq:I_con} is satisfied, we first need to show that the running cost satisfies
	\begin{align}
		\lim\limits_{n\rightarrow\infty } w \sum_{k=0}^{n_{t}} F(\textbf{x}_{n}(t_{k}),\textbf{u}_{n}(t_{k}))  = \int_{0}^{1} F(\textbf{x}^{\infty}(t),\textbf{u}^{\infty}(t))dt.
	\end{align}
	\begin{revblock}To establish this, split
	\begin{align}
		\Big\vert w \sum_{k=0}^{n_{t}} F(\textbf{x}_{n}(t_{k}),\textbf{u}_{n}(t_{k})) - \int_{0}^{1} F(\textbf{x}^{\infty},\textbf{u}^{\infty})dt \Big\vert \leq{}& w \sum_{k=0}^{n_{t}} \big\vert F(\textbf{x}_{n}(t_{k}),\textbf{u}_{n}(t_{k})) - F(\textbf{x}^{\infty}(t_{k}),\textbf{u}^{\infty}(t_{k})) \big\vert \notag \\
		&+ \Big\vert w \sum_{k=0}^{n_{t}} \varphi(t_{k}) - \int_{0}^{1} \varphi(t)\,dt \Big\vert
	\end{align}
	where $\varphi \dfn F(\textbf{x}^{\infty}(\cdot),\textbf{u}^{\infty}(\cdot))$ is a fixed continuous function. Since $w(n_{t}+1) = 1$, the first term is bounded by $L_{F_{\textbf{x}}} \max_{t\in[0,1]}\Vert \textbf{x}_{n} - \textbf{x}^{\infty}\Vert + L_{F_{\textbf{u}}} \max_{t\in[0,1]}\Vert \textbf{u}_{n} - \textbf{u}^{\infty}\Vert$ by Assumption~\ref{ass:lip} and vanishes by Assumption~\ref{ass:conv}; the second term is bounded by $C_{I}W_{\varphi}(1/\sqrt{n_{t}})$ by Lemma~\ref{lem:B_int} applied to $\varphi$ (with $t_{f}=1$ and $N=n_{t}$), and vanishes as $n_{t}\rightarrow\infty$.\end{revblock}
	
	The following relation regarding the terminal cost also holds using Assumptions~\ref{ass:lip} and~\ref{ass:conv}
	\begin{equation}
		\lim\limits_{n\rightarrow\infty} E(\textbf{x}_{n}(0),\textbf{x}_{n}(1)) = E(\textbf{x}^{\infty}(0),\textbf{x}^{\infty}(1)).
	\end{equation}
	\paragraph*{Step 3)} 
	To show the result of this step we first construct approximants to $(\textbf{x}^{*}(t),\textbf{u}^{*}(t))$, the optimal solution of Problem~$P$, using the relations in Eqs.~\eqref{eq:g_k}--\eqref{eq:b_k} \revthree{with the decomposition fixed in Thm.~\ref{th:consis}, i.e., the feasible solution of Problem~$P_{n}$ constructed in the proof of Thm.~\ref{th:feas} for $(\textbf{x},\textbf{u}) = (\textbf{x}^{*},\textbf{u}^{*})$}. We denote them by $(\tilde{\textbf{x}}_n^{*}(t),\tilde{\textbf{u}}_n^{*}(t))$.
	Then, using the same arguments as in the previous step \rev{with Assumption~\ref{ass:conv} replaced by the uniform convergence of $(\tilde{\textbf{x}}_n^{*},\dot{\tilde{\textbf{x}}}{}_n^{*},\tilde{\textbf{u}}_n^{*})$ to $(\textbf{x}^{*},\dot{\textbf{x}}^{*},\textbf{u}^{*})$, guaranteed by Thm.~\ref{th:FB_conv1} and Corollary~\ref{cor:diff} for the decomposition of Thm.~\revthree{\ref{th:consis}}}, we can show that 
	\begin{equation}
    \label{eq:97}
		\lim\limits_{n\rightarrow\infty } J_{n} (\tilde{\textbf{x}}_n^{*}(t),\tilde{\textbf{u}}_n^{*}(t)) = J(\textbf{x}^{*}(t),\textbf{u}^{*}(t)).
	\end{equation}
	
	Having already shown the feasibility of the solution $(\textbf{x}^{\infty}(t),\textbf{u}^{\infty}(t))$ to Problem~$P$, we now use the result of the second step to obtain
	\begin{align}
    \label{eq:98}
		J(\textbf{x}^{*}(t),\textbf{u}^{*}(t)) \leq J(\textbf{x}^{\infty}(t),\textbf{u}^{\infty}(t)) = \lim\limits_{n\rightarrow\infty} J_{n} (\textbf{x}_{n}(t),\textbf{u}_{n}(t)).
	\end{align}
	\rev{However, following the argument in the proof of Thm.~\ref{th:feas} applied to $(\textbf{x}^{*},\textbf{u}^{*})$, the solution furnished by Assumption~\ref{ass:exis}, for which $\delta_{P}^{n}\revthree{\pXU}$ in Eq.~\eqref{eq:del_P_n} is computed, the approximants $(\tilde{\textbf{x}}_n^{*},\tilde{\textbf{u}}_n^{*})$ are feasible for \revthree{the same} Problem~$P_{n}$ \revthree{as in Thm.~\ref{th:consis}}; optimality of $(\textbf{x}_{n},\textbf{u}_{n})$ for Problem~$P_{n}$ then yields}
	\begin{align}
    \label{eq:99}
        J(\textbf{x}^{\infty}(t),\textbf{u}^{\infty}(t)) \leq 
        \lim\limits_{n\rightarrow\infty} J_{n} (\tilde{\textbf{x}}_n^{*}(t),\tilde{\textbf{u}}_n^{*}(t)) = J(\textbf{x}^{*}(t),\textbf{u}^{*}(t)).
	\end{align}
    Thus, using expressions \revthree{\eqref{eq:98}} and~\eqref{eq:99}, we get
	\begin{align}
		J(\textbf{x}^{*}(t),\textbf{u}^{*}(t))  \leq J(\textbf{x}^{\infty}(t),\textbf{u}^{\infty}(t))  \leq J(\textbf{x}^{*}(t),\textbf{u}^{*}(t))
	\end{align}
	and the result follows by applying the squeeze theorem.
\end{proof}

\section{Proof of \texorpdfstring{\revthree{Thm.}}{Thm.}~\ref{th:feas_d}}
\label{sec:app3}
\begin{proof}
	Similarly to the proof of Th.~\ref{th:feas}, we construct a feasible solution for Problem~$P_{n\lambda}$ using an existing solution to Problem~$P_{\lambda}$.
	Let $(\textbf{x}(t),\textbf{u}(t),\bm{\lambda}(t),\bm{\mu}(t),\bm{\nu})$ be a feasible solution to Problem $P_{\lambda}$ \revthree{as in Thm.~\ref{th:feas_d}, for instance the optimal solution of Assumption~\ref{ass:exis_d}; the construction uses only its feasibility and its regularity, and not its optimality}. 
	Define equispaced time nodes $t_{k}\dfn \rev{\frac{k}{n_{t}}}$ for all $k=0,\dots,n_{t}$. The approximated states, $\textbf{x}_{n}(t)$,  and control inputs, $\textbf{u}_{n}(t)$, satisfy Eq.~\eqref{eq:x_app} and~\eqref{eq:u_app}, respectively, and the approximated costates, $\bm{\lambda}_{n}(t)$, and inequality multipliers, $\bm{\mu}_{n}(t)$, satisfy Eq.~\eqref{eq:lamb_app} and~\eqref{eq:mu_app}, respectively, \revthree{with $\bm{\lambda}_{n}/w$ and $\bm{\mu}_{n}/w$ the mixed approximants of $\bm{\lambda}$ and $\bm{\mu}$,} whose construction is readily achieved using the coefficients presented in Eqs.~\eqref{eq:g_k}--\eqref{eq:b_k}.
	The equality multipliers satisfy $\bm{\nu} = \bar{\bm{\nu}}$.
	
	Theorem~\ref{th:feas}\revthree{, applied to the primal pair $(\textbf{x},\textbf{u})$,} shows that $\textbf{x}_{n}(t)$ and $\textbf{u}_{n}(t)$ satisfy Eqs.~\eqref{eq:dyn_Pn}--\eqref{eq:in_Pn} \revthree{with $\delta_{P}^{n}\pXU$}.
	The remainder of this proof shows that $(\textbf{x}_{n}(t),\textbf{u}_{n}(t),\bm{\lambda}_{n}(t),\bm{\mu}_{n}(t),\bar{\bm{\nu}})$ also satisfy Eqs.~\eqref{eq:in_eq_Pnl}--\eqref{eq:close2} \revthree{with $\delta_{D}^{n}\pXULM$}.
	First, we define:
	\begin{align}
		\bar{\bm{\lambda}}_{n}(t_{k}) = \frac{\bm{\lambda}_{n}(t_{k})}{w}, \quad
		\bar{\bm{\mu}}_{n}(t_{k}) = \frac{\bm{\mu}_{n}(t_{k})}{w}
	\end{align}
	for all $k = 0,\dots,n_{t}$\revthree{, and likewise $\bar{\bm{\lambda}}_{n}(t) \dfn \bm{\lambda}_{n}(t)/w$ and $\bar{\bm{\mu}}_{n}(t) \dfn \bm{\mu}_{n}(t)/w$ for $t \in [0,1]$}.
	
	We start by showing that the conditions in~\eqref{eq:in_eq_Pnl} are satisfied.
	Using the definition above for $\bar{\bm{\mu}}_{n}(t_{k})$, we add and subtract the expression
	$w \bm{\mu}^{T}(t_{k}) ( \textbf{h}(\textbf{x}_{n}(t_{k}) , \textbf{u}_{n}(t_{k})) +  \textbf{h}(\textbf{x}(t_{k}) , \textbf{u}(t_{k})))$, so we get
	\begin{align}
		& \left\Vert \bm{\mu}_n^{T}(t_{k}) \textbf{h}(\textbf{x}_n(t_{k}),\textbf{u}_n(t_{k})) \right\Vert = \left\Vert w \bar{\bm{\mu}}_n^{T}(t_{k}) \textbf{h}(\textbf{x}_n(t_{k}),\textbf{u}_n(t_{k})) \right\Vert 
		\leq w \left\| \bm{\mu}^{T}(t_{k}) \textbf{h}(\textbf{x}(t_{k}) , \textbf{u}(t_{k})) \right\|
		\\
		& + w \left\| (\bar{\bm{\mu}}_n^{T}(t_{k}) - \bm{\mu}^{T}(t_{k})) \textbf{h}(\textbf{x}_n(t_{k}),\textbf{u}_n(t_{k})) \right\| 
		+ w \left\| \bm{\mu}^{T}(t_{k}) (\textbf{h}(\textbf{x}_{n}(t_{k}) , \textbf{u}_{n}(t_{k})) -  \textbf{h}(\textbf{x}(t_{k}) , \textbf{u}(t_{k}))) \right\|  \notag
	\end{align}
	where the inequality holds by using the triangle inequality. 
	Note that the first expression in the inequality is nulled using~\eqref{eq:mu_h_0}.
	Moreover, we get that
	\begin{align}
		&\left\| (\bar{\bm{\mu}}_n^{T}(t_{k}) - \bm{\mu}^{T}(t_{k})) \textbf{h}(\textbf{x}_n(t_{k}),\textbf{u}_n(t_{k})) \right\| \leq \delta^{n}_{\bm{\mu}}\revthree{\pM}\textbf{h}_{\max}
		\\
		&\left\| \bm{\mu}^{T}(t_{k}) (\textbf{h}(\textbf{x}_{n}(t_{k}) , \textbf{u}_{n}(t_{k})) -  \textbf{h}(\textbf{x}(t_{k}) , \textbf{u}(t_{k}))) \right\| \leq \bm{\mu}_{\max}(L_{\textbf{h}_{x}} \delta^{n}_{\textbf{x}}\revthree{\pX} + L_{\textbf{h}_{u}} \delta^{n}_{\textbf{u}}\revthree{\pU})
	\end{align}
	where $L_{\textbf{h}_{x}}$ and $L_{\textbf{h}_{u}}$ are the Lipschitz constant of $\textbf{h}$ with respect to $\textbf{x}$ and $\textbf{u}$, respectively, and \revthree{$\textbf{h}_{\max} \dfn \sup_{n}\max_{k}\Vert\textbf{h}(\textbf{x}_{n}(t_{k}),\textbf{u}_{n}(t_{k}))\Vert$ and $\bm{\mu}_{\max} \dfn \max_{t\in[0,1]}\Vert\bm{\mu}(t)\Vert$}, whose existence is guaranteed from Assumption~\ref{ass:lip}\revthree{, the continuity of $\textbf{x}$, $\textbf{u}$, and $\bm{\mu}$, and the uniform convergence of the approximants}.
	Thus, there exists a positive constant $\tilde{C}$ that is independent of $n$, such that
	\begin{equation}
		\left\Vert \bm{\mu}_n^{T}(t_{k}) \textbf{h}(\textbf{x}_n(t_{k}),\textbf{u}_n(t_{k})) \right\Vert \leq w \tilde{C} \max\Big\{ \delta_{\bm{\mu}}^{n}\revthree{\pM} , \delta_{\textbf{x}}^{n}\revthree{\pX}  , \delta_{\textbf{u}}^{n}\revthree{\pU} \Big\}, \quad \forall k=0,\dots,n_{t}
	\end{equation}
	where we choose $C_{D} > \tilde{C}$\revthree{; since $w < 1/n_{t}$ and $\delta_{D}^{n}\pXULM$ is given by~\eqref{eq:del_D}, this yields the first condition in~\eqref{eq:in_eq_Pnl}}.
	Similarly, we add and subtract $w\bm{\mu}(t_{k})$, and get
	\begin{align}
		\bm{\mu}_n(t_{k}) = w \bar{\bm{\mu}}_n(t_{k})\geq  w \bm{\mu}(t_{k}) - 
		w \left\| \bm{\mu}(t_{k}) - \bar{\bm{\mu}}_n(t_{k}) \right\| \bm{1}
		\geq - w \delta^{n}_{\bm{\mu}}\revthree{\pM}\bm{1},
	\end{align}
	which\revthree{, since $w < 1/n_{t}$ and $\delta^{n}_{\bm{\mu}}\pM$ is a component of~\eqref{eq:del_D}, gives the second condition in~\eqref{eq:in_eq_Pnl} whenever $C_{D} \geq 1$; hence,} Eq.~\eqref{eq:in_eq_Pnl} holds.
	
	\begin{revblock}The stationarity analysis transfers node sums to integrals; we record the elementary estimate used throughout. For continuous $g$ and continuously differentiable $\psi$ on $[0,1]$,
	\begin{align}
		\label{eq:transfer}
		\Big\vert w\sum_{k=0}^{n_{t}} g(t_{k})\psi(t_{k}) - \int_{0}^{1} g(t)\psi(t)dt \Big\vert
		\leq \Vert\psi\Vert_{\max} W_{g}\Big(\frac{1}{n_{t}}\Big) + \frac{\Vert g\Vert_{\max}\Vert\dot{\psi}\Vert_{\max}}{n_{t}} + \frac{2\Vert g\Vert_{\max}\Vert\psi\Vert_{\max}}{n_{t}}
	\end{align}
	where $\Vert\cdot\Vert_{\max}$ denotes the maximum over $[0,1]$: with $h \dfn g\psi$, replacing $w\sum_{k=0}^{n_{t}}h(t_{k})$ by the left Riemann sum $\frac{1}{n_{t}}\sum_{k=0}^{n_{t}-1}h(t_{k})$ costs at most $2\Vert h\Vert_{\max}/(n_{t}+1)$, the left Riemann sum differs from the integral by at most $W_{h}(1/n_{t})$, and 
    \begin{equation}
        W_{h}(\delta) \leq \Vert\psi\Vert_{\max}W_{g}(\delta) + \Vert g\Vert_{\max}\Vert\dot{\psi}\Vert_{\max}\delta.
    \end{equation}
	The mixed-basis weights satisfy, uniformly in $k$: $\Vert b_{k,n_{B}}\Vert_{\max} \leq 1$, $\Vert\dot{b}_{k,n_{B}}\Vert_{\max} \leq 2n_{B}$, and $\Vert\ddot{b}_{k,n_{B}}\Vert_{\max} \leq 4n_{B}^{2}$ (from $\dot{b}_{k,n} = n(b_{k-1,n-1} - b_{k,n-1})$); $\Vert c_{k}\Vert_{\max}, \Vert s_{k}\Vert_{\max} \leq 1$\revthree{, with $c_{k}(t) \dfn \cos(2\pi k t)$ and $s_{k}(t) \dfn \sin(2\pi k t)$}, $\Vert\dot{c}_{k}\Vert_{\max}, \Vert\dot{s}_{k}\Vert_{\max} \leq 2\pi n_{F}$, and $\Vert\ddot{c}_{k}\Vert_{\max}, \Vert\ddot{s}_{k}\Vert_{\max} \leq 4\pi^{2}n_{F}^{2}$; moreover $w\sum_{k=0}^{n_{t}}\vert\psi(t_{k})\vert \leq 1$ for $\psi \in \{b_{k,n_{B}}, c_{k}, s_{k}\}$ and $\int_{0}^{1}\vert b_{k,n_{B}}\vert dt = \frac{1}{n_{B}+1}$. \revthree{A subscript $\max$ denotes the maximum over $[0,1]$ of the norm of the indicated function, e.g., $F_{\textbf{x},\max} \dfn \Vert F_{\textbf{x}}(\textbf{x}(\cdot),\textbf{u}(\cdot))\Vert_{\max}$ and $\bm{\lambda}_{\max} \dfn \Vert\bm{\lambda}\Vert_{\max}$; for functions evaluated along the approximants, it also includes the supremum over $n$, e.g., $\textbf{f}_{\textbf{x},\max} \dfn \sup_{n}\Vert\textbf{f}_{\textbf{x}}(\textbf{x}_{n}(\cdot),\textbf{u}_{n}(\cdot))\Vert_{\max}$, similarly $\textbf{h}_{\textbf{x},\max}$, and $\bar{\bm{\lambda}}_{\max} \dfn \sup_{n}\Vert\bar{\bm{\lambda}}_{n}\Vert_{\max} \leq \bm{\lambda}_{\max} + \sup_{n}\delta_{\bm{\lambda}}^{n}\pL$; all are finite by the uniform convergence of the approximants and Assumption~\ref{ass:lip_d}.}
	By Assumption~\ref{ass:lip_d} the moduli of the composite integrands reduce to those of the solution, e.g., $W_{F_{\textbf{x}}(\textbf{x}(\cdot),\textbf{u}(\cdot))}(\delta) \leq L_{F_{\textbf{x}}}(W_{\textbf{x}}(\delta) + W_{\textbf{u}}(\delta))$; since $\textbf{x}, \bm{\lambda} \in \mathcal{C}^{1}$, $W_{\textbf{x}}(1/n_{t})$ and $W_{\bm{\lambda}}(1/n_{t})$ are $O(1/n_{t})$, while $W_{\textbf{u}}(1/n_{t}) \leq W_{\textbf{u}_{\textbf{g}}}(1/\sqrt{n}) + O(1/\sqrt{n})$ and $W_{\bm{\mu}}(1/n_{t}) \leq W_{\bm{\mu}_{\textbf{g}}}(1/\sqrt{n}) + O(1/\sqrt{n})$, using $1/n_{t} \leq 1/\sqrt{n}$ and the $\mathcal{C}^{1}$ regularity of $\textbf{u}_{\textbf{p}}$ and $\bm{\mu}_{\textbf{p}}$.
	Throughout, let
	\begin{align*}
		\eta_{n}\revthree{\pXULM} \dfn \max\big\{&\delta_{\textbf{x}}^{n}\revthree{\pX}, \delta_{\textbf{u}}^{n}\revthree{\pU}, \delta_{\bm{\lambda}}^{n}\revthree{\pL}, \delta_{\dot{\bm{\lambda}}}^{n}\revthree{\pL}, \delta_{\bm{\mu}}^{n}\revthree{\pM}, \\
		&W_{\textbf{u}_{\textbf{g}}}(1/\sqrt{n}), W_{\bm{\mu}_{\textbf{g}}}(1/\sqrt{n}), 1/\sqrt{n}, n_{B}^{2}/n_{t}, n_{F}^{2}/n_{t}\big\}
	\end{align*}
	collect the components of Eq.~\eqref{eq:del_D}\revthree{, so that $\delta_{D}^{n}\pXULM = C_{D}\,\eta_{n}\pXULM$}.\end{revblock}

	Next, we prove that the stationarity conditions hold.
	First, we consider
	\begin{align}
		&\left\Vert \frac{\partial \mathcal{L}_{n}}{\partial \bar{\textbf{c}}_{x,0}} \right\Vert  = \Bigg\| E_{\textbf{x}(0)}(\textbf{x}_n(0),\textbf{x}_n(1)) 
		+ w \sum_{k=0}^{n_{t}} F_{\textbf{x}}(\textbf{x}_n(t_{k}),\textbf{u}_n(t_{k})) b_{0,n}(t_{k})
		+ \bar{\bm{\nu}}^{T} \textbf{e}_{\textbf{x}(0)}(\textbf{x}_n(0),\textbf{x}_n(1)) 
		\notag \\
		& 
		+ \sum_{k=0}^{n_{t}} \bm{\lambda}_n^{T}(t_{k})\Big(\textbf{f}_{\textbf{x}}(\textbf{x}_n(t_{k}),\textbf{u}_n(t_{k}))b_{0,n}(t_{k})-\dot{b}_{0,n}(t_{k})\Big)
		+ \sum_{k=0}^{n_{t}} \bm{\mu}_n^{T}(t_{k}) \textbf{h}_{\textbf{x}}(\textbf{x}_n(t_{k}),\textbf{u}_n(t_{k})) b_{0,n}(t_{k})
		\bigg\| \notag \\
		& = \Bigg\| E_{\textbf{x}(0)}(\textbf{x}_n(0),\textbf{x}_n(1)) 
		+ w \sum_{k=0}^{n_{t}} F_{\textbf{x}}(\textbf{x}_n(t_{k}),\textbf{u}_n(t_{k})) b_{0,n}(t_{k})
		+ \bar{\bm{\nu}}^{T} \textbf{e}_{\textbf{x}(0)}(\textbf{x}_n(0),\textbf{x}_n(1)) 
		\notag \\
		& 
		+ w\sum_{k=0}^{n_{t}} \bar{\bm{\lambda}}_n^{T}(t_{k})\Big(\textbf{f}_{\textbf{x}}(\textbf{x}_n(t_{k}),\textbf{u}_n(t_{k}))b_{0,n}(t_{k})-\dot{b}_{0,n}(t_{k})\Big)
		+ w\sum_{k=0}^{n_{t}} \bar{\bm{\mu}}_n^{T}(t_{k}) \textbf{h}_{\textbf{x}}(\textbf{x}_n(t_{k}),\textbf{u}_n(t_{k})) b_{0,n}(t_{k})
		\bigg\| \notag
	\end{align}
	where in the second equality we substitute $\bm{\lambda}_n^{T}(t_{k}) = w \bar{\bm{\lambda}}_n^{T}(t_{k})$ and $\bm{\mu}_n^{T}(t_{k}) = w \bar{\bm{\mu}}_n^{T}(t_{k})$.
	
	\begin{revblock}We add and subtract $w\sum_{k=0}^{n_{t}}F_{\textbf{x}}(\textbf{x}(t_{k}),\textbf{u}(t_{k}))b_{0,n}(t_{k})$. The node-swap term satisfies
	\begin{align}
		\Big\Vert w\sum_{k=0}^{n_{t}}\big[F_{\textbf{x}}(\textbf{x}_{n}(t_{k}),\textbf{u}_{n}(t_{k})) - F_{\textbf{x}}(\textbf{x}(t_{k}),\textbf{u}(t_{k}))\big]b_{0,n}(t_{k})\Big\Vert \leq L_{F_{\textbf{x}}}(\delta_{\textbf{x}}^{n}\revthree{\pX} + \delta_{\textbf{u}}^{n}\revthree{\pU})
	\end{align}
	by Assumption~\ref{ass:lip_d} and $w\sum_{k}\vert b_{0,n}(t_{k})\vert \leq 1$, while the remaining transfer is bounded by Eq.~\eqref{eq:transfer} with the fixed function $g = F_{\textbf{x}}(\textbf{x}(\cdot),\textbf{u}(\cdot))$ and $\psi = b_{0,n}$, whose error terms are $W_{g}(1/n_{t}) + 2F_{\textbf{x},\max}n_{B}/n_{t} + 2F_{\textbf{x},\max}/n_{t}$. Reducing $W_{g}$ as above, we obtain
	\begin{align}
		\label{eq:int_F}
		\left\|w \sum_{k=0}^{n_{t}} F_{\textbf{x}}(\textbf{x}_n(t_{k}),\textbf{u}_n(t_{k})) b_{0,n}(t_{k}) -
		\int_{0}^{1} F_{\textbf{x}}(\textbf{x}(t) , \textbf{u}(t))b_{0,n}(t)dt \right\|\leq \bar{C}_{1}\,\eta_{n}\revthree{\pXULM}
	\end{align}
	where $\bar{C}_{1}$ is a positive constant independent of $n$, $n_{t}$, $n_{B}$, $n_{F}$, and of the row index (the identical bound holds for every value weight $b_{k,n_{B}}$, $c_{k}$, $s_{k}$).\end{revblock}
	\begin{revblock}For the derivative weight $\dot{b}_{0,n}$ the argument swap is performed on the continuous side. By Eq.~\eqref{eq:transfer} with $g = \bar{\bm{\lambda}}_{n}$, whose derivative is uniformly bounded, $\Vert\dot{\bar{\bm{\lambda}}}_{n}\Vert_{\max} \leq \Vert\dot{\bm{\lambda}}\Vert_{\max} + \delta_{\dot{\bm{\lambda}}}^{n}\revthree{\pL}$, by Corollary~\ref{cor:diff}, and $\psi = \dot{b}_{0,n}$, the node sum differs from $\int_{0}^{1}\bar{\bm{\lambda}}_{n}^{T}\dot{b}_{0,n}dt$ by at most $2n_{B}\Vert\dot{\bar{\bm{\lambda}}}_{n}\Vert_{\max}/n_{t} + 4n_{B}^{2}\bar{\bm{\lambda}}_{\max}/n_{t} + 4n_{B}\bar{\bm{\lambda}}_{\max}/n_{t}$; integrating by parts,
	\begin{align}
		\Big\Vert \int_{0}^{1}(\bar{\bm{\lambda}}_{n} - \bm{\lambda})^{T}(t)\dot{b}_{0,n}(t)dt\Big\Vert &\leq \delta_{\dot{\bm{\lambda}}}^{n}\revthree{\pL}\int_{0}^{1}\vert b_{0,n}\vert dt + 2\delta_{\bm{\lambda}}^{n}\revthree{\pL} \notag \\
		&\leq \delta_{\dot{\bm{\lambda}}}^{n}\revthree{\pL} + 2\delta_{\bm{\lambda}}^{n}\revthree{\pL},
	\end{align}
	so that
	\begin{align}
		\label{eq:int_lamb1}
		\left\|w\sum_{k=0}^{n_{t}} \bar{\bm{\lambda}}_n^{T}(t_{k})\dot{b}_{0,n}(t_{k})
		- \int_{0}^{1} \bm{\lambda}^{T}(t) \dot{b}_{0,n}(t) dt \right\| \leq \bar{C}_{2}\,\eta_{n}\revthree{\pXULM}.
	\end{align}
	The value-weight transfers
	\begin{align}
		\label{eq:int_lamb2}
		&\left\|w\sum_{k=0}^{n_{t}} \bar{\bm{\lambda}}_n^{T}(t_{k}) \textbf{f}_{\textbf{x}}(\textbf{x}_n(t_{k}),\textbf{u}_n(t_{k}))b_{0,n}(t_{k})
		- \int_{0}^{1} \bm{\lambda}^{T}(t) \textbf{f}_{\textbf{x}}(\textbf{x}(t),\textbf{u}(t))b_{0,n}(t) dt\right\| \notag \\
		& \leq \bar{C}_{3}\,\eta_{n}\revthree{\pXULM}
		\\
		\label{eq:int_mu}
		&\left\|w\sum_{k=0}^{n_{t}} \bar{\bm{\mu}}_n^{T}(t_{k}) \textbf{h}_{\textbf{x}}(\textbf{x}_n(t_{k}),\textbf{u}_n(t_{k})) b_{0,n}(t_{k})
		- \int_{0}^{1} \bm{\mu}^{T}(t) \textbf{h}_{\textbf{x}}(\textbf{x}(t),\textbf{u}(t)) b_{0,n}(t) dt\right\| \notag \\
		& \leq \bar{C}_{4}\,\eta_{n}\revthree{\pXULM}
	\end{align}
	follow exactly as Eq.~\eqref{eq:int_F}, with $g = \bm{\lambda}^{T}\textbf{f}_{\textbf{x}}(\textbf{x},\textbf{u})$ and $g = \bm{\mu}^{T}\textbf{h}_{\textbf{x}}(\textbf{x},\textbf{u})$ and node swaps costing $\textbf{f}_{\textbf{x},\max}\delta_{\bm{\lambda}}^{n}\revthree{\pL} + \bm{\lambda}_{\max}L_{\textbf{f}_{\textbf{x}}}(\delta_{\textbf{x}}^{n}\revthree{\pX}+\delta_{\textbf{u}}^{n}\revthree{\pU})$ and $\textbf{h}_{\textbf{x},\max}\delta_{\bm{\mu}}^{n}\revthree{\pM} + \bm{\mu}_{\max}L_{\textbf{h}_{\textbf{x}}}(\delta_{\textbf{x}}^{n}\revthree{\pX}+\delta_{\textbf{u}}^{n}\revthree{\pU})$, respectively.\end{revblock}
	Similarly, we can obtain that
	\begin{align}
		&\left\|E_{\textbf{x}(0)}(\textbf{x}_n(0),\textbf{x}_n(1)) - E_{\textbf{x}(0)}(\textbf{x}(0),\textbf{x}(1)) \right\| \leq \bar{C}_{5} \delta_{\textbf{x}}^{n}\revthree{\pX}, \\
		&\left\|\bar{\bm{\nu}}^{T} \textbf{e}_{\textbf{x}(0)}(\textbf{x}_n(0),\textbf{x}_n(1)) -  \bm{\nu}^{T} \textbf{e}_{\textbf{x}(0)}(\textbf{x}(0),\textbf{x}(1))\right\| \leq \bar{C}_{6} \delta_{\textbf{x}}^{n}\revthree{\pX}.
	\end{align}
	where $\bar{C}_{5}$ and $\bar{C}_{6}$ are positive constants independent of $n$.
	Using these relations, we can write the following inequality
	\begin{align}
		&\left\Vert \frac{\partial \mathcal{L}_{n}}{\partial \bar{\textbf{c}}_{x,0}} \right\Vert \leq \Bigg\| E_{\textbf{x}(0)}(\textbf{x}(0),\textbf{x}(1)) 
		+ \int_{0}^{1} F_{\textbf{x}}(\textbf{x}(t) , \textbf{u}(t))b_{0,n}(t)dt
		+ \bar{\bm{\nu}}^{T} \textbf{e}_{\textbf{x}(0)}(\textbf{x}(0),\textbf{x}(1)) 
		\notag \\
		& 
		- \int_{0}^{1} \bm{\lambda}^{T}(t) \dot{b}_{0,n}(t) dt
		+ \int_{0}^{1} \bm{\lambda}^{T}(t) \textbf{f}_{\textbf{x}}(\textbf{x}(t),\textbf{u}(t))b_{0,n}(t) dt \notag \\
		&+ \int_{0}^{1} \bm{\mu}^{T}(t) \textbf{h}_{\textbf{x}}(\textbf{x}(t),\textbf{u}(t)) b_{0,n}(t) dt
		\bigg\| 
		+ \bar{C}\,\rev{\eta_{n}\revthree{\pXULM}}
	\end{align}
	where $\bar{C} > 6 \max_{i=1,\dots,6}\{ \bar{C}_{i}\}$. 
	We apply integration by parts to obtain:
	\begin{align}
		\int_{0}^{1} \bm{\lambda}^{T}(t) \dot{b}_{0,n}(t) dt & = - \int_{0}^{1} \dot{\bm{\lambda}}^{T}(t) b_{0,n}(t) dt +\Big[ \bm{\lambda}^{T}(t) b_{0,n}(t)\Big]_{0}^{1} \notag \\
		& = - \int_{0}^{1} \dot{\bm{\lambda}}^{T}(t) b_{0,n}(t) dt - \bm{\lambda}^{T}(0)
	\end{align}
	where $b_{0,n}(0) = 1$ and $b_{0,n}(1) = 0$.
	We get that
	\begin{align}
		&\left\Vert \frac{\partial \mathcal{L}_{n}}{\partial \bar{\textbf{c}}_{x,0}} \right\Vert \leq \left\| E_{\textbf{x}(0)}(\textbf{x}(0),\textbf{x}(1)) + \bm{\lambda}^{T}(0)
		+ \bar{\bm{\nu}}^{T} \textbf{e}_{\textbf{x}(0)}(\textbf{x}(0),\textbf{x}(1)) \right\| 
		\notag \\
		& 
		+\left\| \int_{0}^{1} \Big( F_{\textbf{x}}(\textbf{x}(t) , \textbf{u}(t))
		+\dot{\bm{\lambda}}^{T}(t) 	+ \bm{\lambda}^{T}(t) \textbf{f}_{\textbf{x}}(\textbf{x}(t),\textbf{u}(t))
		+ \bm{\mu}^{T}(t) \textbf{h}_{\textbf{x}}(\textbf{x}(t),\textbf{u}(t))\Big)b_{0,n}(t) dt \right\| \notag \\
		& + \bar{C}\,\rev{\eta_{n}\revthree{\pXULM}}
	\end{align}
	The first term is nulled by applying Eq.~\eqref{eq:closure1} and the second term using~\eqref{eq:lamb_dot}\revthree{; hence $\Vert\partial\mathcal{L}_{n}/\partial\bar{\textbf{c}}_{x,0}\Vert \leq \bar{C}\,\eta_{n}\pXULM \leq \delta_{D}^{n}\pXULM$ whenever $C_{D} \geq \bar{C}$}.
	\rev{The same estimate holds for every Bernstein row $\bar{\textbf{c}}_{x,k}$, $k=0,\dots,n_{B}$, the weight bounds above are uniform in $k$, the boundary group at $t=0$ is replaced by that at $t=1$ for $k=n_{B}$ and is absent for $0<k<n_{B}$ since $b_{k,n_{B}}(0)=b_{k,n_{B}}(1)=0$, and for the $\textbf{u}$-block rows, which contain no $\dot{b}$ terms.}
	
	Similarly, we prove the stationarity conditions in~\eqref{eq:stat} for the Fourier part of the approximation.
	Consider:
	\begin{align}
		&\left\Vert \frac{\partial \mathcal{L}_{n}}{\partial \revthree{\textbf{a}}_{x,1}} \right\Vert  = \Bigg\| E_{\textbf{x}(0)}(\textbf{x}_n(0),\textbf{x}_n(1)) 
		+ E_{\textbf{x}(1)}(\textbf{x}_n(0),\textbf{x}_n(1)) 
		+ w \sum_{k=0}^{n_{t}} F_{\textbf{x}}(\textbf{x}_n(t_{k}),\textbf{u}_n(t_{k})) c_{1}(t_{k}) \notag \\
		& + \sum_{k=0}^{n_{t}} \bm{\lambda}_n^{T}(t_{k})\Big(\textbf{f}_{\textbf{x}}(\textbf{x}_n(t_{k}),\textbf{u}_n(t_{k}))c_{1}(t_{k})-\dot{c}_{1}(t_{k})\Big)
		+ \sum_{k=0}^{n_{t}} \bm{\mu}_n^{T}(t_{k}) \textbf{h}_{\textbf{x}}(\textbf{x}_n(t_{k}),\textbf{u}_n(t_{k})) c_{1}(t_{k}) \notag \\
		& + \bar{\bm{\nu}}^{T} \textbf{e}_{\textbf{x}(0)}(\textbf{x}_n(0),\textbf{x}_n(1))
		+ \bar{\bm{\nu}}^{T} \textbf{e}_{\textbf{x}(1)}(\textbf{x}_n(0),\textbf{x}_n(1)) 
		\bigg\| \notag \\
		& = \Bigg\| E_{\textbf{x}(0)}(\textbf{x}_n(0),\textbf{x}_n(1)) + E_{\textbf{x}(1)}(\textbf{x}_n(0),\textbf{x}_n(1)) 
		+ w \sum_{k=0}^{n_{t}} F_{\textbf{x}}(\textbf{x}_n(t_{k}),\textbf{u}_n(t_{k})) c_{1}(t_{k}) \notag \\
		& + w\sum_{k=0}^{n_{t}} \bar{\bm{\lambda}}_n^{T}(t_{k})\Big(\textbf{f}_{\textbf{x}}(\textbf{x}_n(t_{k}),\textbf{u}_n(t_{k}))c_{1}(t_{k})-\dot{c}_{1}(t_{k})\Big)
		+ w\sum_{k=0}^{n_{t}} \bar{\bm{\mu}}_n^{T}(t_{k}) \textbf{h}_{\textbf{x}}(\textbf{x}_n(t_{k}),\textbf{u}_n(t_{k})) c_{1}(t_{k}) \notag \\
		& + \bar{\bm{\nu}}^{T} \textbf{e}_{\textbf{x}(0)}(\textbf{x}_n(0),\textbf{x}_n(1)) 
		+ \bar{\bm{\nu}}^{T} \textbf{e}_{\textbf{x}(1)}(\textbf{x}_n(0),\textbf{x}_n(1)) 
		\bigg\| 
	\end{align}
	where $c_{j}(t_{k}) \dfn \cos \left( 2 \pi j t_{k}\right)$ and $\dot{c}_{j}(t_{k})$ is its time derivative, and in the second equality we substitute $\bm{\lambda}_n^{T}(t_{k}) = w \bar{\bm{\lambda}}_n^{T}(t_{k})$ and $\bm{\mu}_n^{T}(t_{k}) = w \bar{\bm{\mu}}_n^{T}(t_{k})$.
	\rev{By Eq.~\eqref{eq:transfer} exactly as in Eqs.~\eqref{eq:int_F}--\eqref{eq:int_mu}, with the continuous-side swap of Eq.~\eqref{eq:int_lamb1} for the $\dot{c}_{1}$ weight, using $\Vert\dot{c}_{j}\Vert_{\max} \leq 2\pi n_{F}$ and $\Vert\ddot{c}_{j}\Vert_{\max} \leq 4\pi^{2}n_{F}^{2}$,} we get that
	\begin{align}
		& \left\| w \sum_{k=0}^{n_{t}} F_{\textbf{x}}(\textbf{x}_n(t_{k}),\textbf{u}_n(t_{k})) c_{1}(t_{k}) -
		\int_{0}^{1} F_{\textbf{x}}(\textbf{x}(t) , \textbf{u}(t))c_{1}(t)dt \right\|\leq \bar{C}_{7}\,\rev{\eta_{n}\revthree{\pXULM}} \\
		&\left\|w\sum_{k=0}^{n_{t}} \bar{\bm{\lambda}}_n^{T}(t_{k})\dot{c}_{1}(t_{k})
		- \int_{0}^{1} \bm{\lambda}^{T}(t) \dot{c}_{1}(t) dt \right\| \leq \bar{C}_{8}\,\rev{\eta_{n}\revthree{\pXULM}}
		\\
		&\left\|w\sum_{k=0}^{n_{t}} \bar{\bm{\lambda}}_n^{T}(t_{k}) \textbf{f}_{\textbf{x}}(\textbf{x}_n(t_{k}),\textbf{u}_n(t_{k}))c_{1}(t_{k})
		- \int_{0}^{1} \bm{\lambda}^{T}(t) \textbf{f}_{\textbf{x}}(\textbf{x}(t),\textbf{u}(t))c_{1}(t) dt\right\| \notag \\
		& \leq \bar{C}_{9}\,\rev{\eta_{n}\revthree{\pXULM}}
		\\
		&\left\|w\sum_{k=0}^{n_{t}} \bar{\bm{\mu}}_n^{T}(t_{k}) \textbf{h}_{\textbf{x}}(\textbf{x}_n(t_{k}),\textbf{u}_n(t_{k})) c_{1}(t_{k})
		- \int_{0}^{1} \bm{\mu}^{T}(t) \textbf{h}_{\textbf{x}}(\textbf{x}(t),\textbf{u}(t)) c_{1}(t) dt\right\| \notag \\
		& \leq \bar{C}_{10}\,\rev{\eta_{n}\revthree{\pXULM}}
	\end{align}
	where $\bar{C}_{7}$, $\bar{C}_{8}$, $\bar{C}_{9}$, and $\bar{C}_{10}$ are positive constants independent of $n$.
	We can also obtain that
	\begin{align}
		&\left\|E_{\textbf{x}(1)}(\textbf{x}_n(0),\textbf{x}_n(1)) - E_{\textbf{x}(1)}(\textbf{x}(0),\textbf{x}(1)) \right\| \leq \bar{C}_{11} \delta_{\textbf{x}}^{n}\revthree{\pX}, \\
		&\left\|\bar{\bm{\nu}}^{T} \textbf{e}_{\textbf{x}(1)}(\textbf{x}_n(0),\textbf{x}_n(1)) -  \bm{\nu}^{T} \textbf{e}_{\textbf{x}(1)}(\textbf{x}(0),\textbf{x}(1))\right\| \leq \bar{C}_{12} \delta_{\textbf{x}}^{n}\revthree{\pX}.
	\end{align}
	where $\bar{C}_{11}$ and $\bar{C}_{12}$ are positive constants independent of $n$.
	Using these relations, we can write the following inequality
	\begin{align}
		&\left\Vert \frac{\partial \mathcal{L}_{n}}{\partial \textbf{a}_{x,1}} \right\Vert \leq \Bigg\| E_{\textbf{x}(0)}(\textbf{x}(0),\textbf{x}(1)) 
		+ \int_{0}^{1} F_{\textbf{x}}(\textbf{x}(t) , \textbf{u}(t))c_{1}(t)dt
		+ \bar{\bm{\nu}}^{T} \textbf{e}_{\textbf{x}(0)}(\textbf{x}(0),\textbf{x}(1)) 
		\notag \\
		& + E_{\textbf{x}(1)}(\textbf{x}(0),\textbf{x}(1)) 
		+ \bar{\bm{\nu}}^{T} \textbf{e}_{\textbf{x}(1)}(\textbf{x}(0),\textbf{x}(1)) 
		\notag \\
		& 
		- \int_{0}^{1} \bm{\lambda}^{T}(t) \dot{c}_{1}(t) dt
		+ \int_{0}^{1} \bm{\lambda}^{T}(t) \textbf{f}_{\textbf{x}}(\textbf{x}(t),\textbf{u}(t))c_{1}(t) dt \notag \\
		&+ \int_{0}^{1} \bm{\mu}^{T}(t) \textbf{h}_{\textbf{x}}(\textbf{x}(t),\textbf{u}(t)) c_{1}(t) dt
		\bigg\| 
		+ \bar{C}\,\rev{\eta_{n}\revthree{\pXULM}}
	\end{align}
	where \revthree{$\bar{C}$ is enlarged, if necessary, so that also} $\bar{C} > 8 \max_{i=5,\dots,12}\{ \bar{C}_{i}\}$. 
	We apply integration by parts to obtain:
	\begin{align}
		\int_{0}^{1} \bm{\lambda}^{T}(t) \dot{c}_{1}(t) dt & = - \int_{0}^{1} \dot{\bm{\lambda}}^{T}(t) c_{1}(t) dt +\Big[ \bm{\lambda}^{T}(t) c_{1}(t)\Big]_{0}^{1} \notag \\
		& = - \int_{0}^{1} \dot{\bm{\lambda}}^{T}(t) c_{1}(t) dt + \bm{\lambda}^{T}(1) - \bm{\lambda}^{T}(0)
	\end{align}
	where $c_{1}(0) = c_{1}(1) = 1$.
	We get that
	\begin{align}
		&\left\Vert \frac{\partial \mathcal{L}_{n}}{\partial \textbf{a}_{x,1}} \right\Vert \leq \left\| E_{\textbf{x}(0)}(\textbf{x}(0),\textbf{x}(1)) 
		+ \bm{\lambda}^{T}(0)
		+ \bar{\bm{\nu}}^{T} \textbf{e}_{\textbf{x}(0)}(\textbf{x}(0),\textbf{x}(1)) \right\| 
		\notag \\
		& + \left\| E_{\textbf{x}(1)}(\textbf{x}(0),\textbf{x}(1)) 
		- \bm{\lambda}^{T}(1)
		+ \bar{\bm{\nu}}^{T} \textbf{e}_{\textbf{x}(1)}(\textbf{x}(0),\textbf{x}(1)) \right\| 
		\notag \\
		& 
		+\left\| \int_{0}^{1} \Big( F_{\textbf{x}}(\textbf{x}(t) , \textbf{u}(t))
		+\dot{\bm{\lambda}}^{T}(t) 	+ \bm{\lambda}^{T}(t) \textbf{f}_{\textbf{x}}(\textbf{x}(t),\textbf{u}(t))
		+ \bm{\mu}^{T}(t) \textbf{h}_{\textbf{x}}(\textbf{x}(t),\textbf{u}(t))\Big)c_{1}(t) dt \right\| \notag \\
		& + \bar{C}\,\rev{\eta_{n}\revthree{\pXULM}}
	\end{align}
	The first two terms are nulled by applying Eq.~\eqref{eq:closure1} and~\eqref{eq:closure2}, and the last term is nulled using~\eqref{eq:lamb_dot}\revthree{; hence $\Vert\partial\mathcal{L}_{n}/\partial\textbf{a}_{x,1}\Vert \leq \delta_{D}^{n}\pXULM$ whenever $C_{D} \geq \bar{C}$}.
	\rev{The same estimate holds uniformly for every Fourier row $\revthree{\textbf{a}}_{x,j}$, $\revthree{\textbf{b}}_{x,j}$, $j \leq n_{F}$, and for the constant row $\revthree{\textbf{a}}_{x,0}$ (weight $1/2$): the weight bounds grow at most like $2\pi n_{F}$ and $4\pi^{2}n_{F}^{2}$, which the $n_{F}^{2}/n_{t}$ component of $\eta_{n}\revthree{\pXULM}$ absorbs; the $s_{j}$ rows carry no boundary group since $s_{j}(0)=s_{j}(1)=0$; and the $\textbf{u}$-block Fourier rows follow identically without $\dot{c}$, $\dot{s}$ terms.}
	
	The first closure condition~\eqref{eq:close1} satisfies
	\begin{align}
		&\left\Vert \frac{\bm{\lambda}_{n}^{T}(0)}{w} + \bar{\bm{\nu}}^{T}\textbf{e}_{\textbf{x}(0)}(\textbf{x}_{n}(0),\textbf{x}_{n}(1)) + E_{\textbf{x}(0)}(\textbf{x}_{n}(0),\textbf{x}_{n}(1)) \right\Vert \notag \\
		& = \left\Vert \bar{\bm{\lambda}}_{n}^{T}(0) + \bar{\bm{\nu}}^{T}\textbf{e}_{\textbf{x}(0)}(\textbf{x}_{n}(0),\textbf{x}_{n}(1)) + E_{\textbf{x}(0)}(\textbf{x}_{n}(0),\textbf{x}_{n}(1)) \right\Vert \notag \\
		&\leq \left\| \bm{\lambda}^{T}(0) + \bm{\nu}^{T}\textbf{e}_{\textbf{x}(0)}(\textbf{x}(0),\textbf{x}(1)) + 
		E_{\textbf{x}(0)}(\textbf{x}(0),\textbf{x}(1)) \right\| + \delta_{\bm{\lambda}}^{n}\revthree{\pL} + (\bar{C}_{5} + \bar{C}_{6})\delta_{\textbf{x}}^{n}\revthree{\pX}
	\end{align}
	where the first term after the inequality is nulled by applying Eq.~\eqref{eq:closure1}\revthree{, so that~\eqref{eq:close1} holds whenever} $C_{D} > 2 \max\{ 1 , (\bar{C}_{5} + \bar{C}_{6}) \}$.
	Using the same arguments, we can show that the other closure condition in~\eqref{eq:close2} also satisfies its inequality constraint \revthree{whenever $C_{D} > 2\max\{1, (\bar{C}_{11} + \bar{C}_{12})\}$}.
	\revthree{Collecting these requirements, with $\bar{C}$ enlarged, if necessary, to cover the analogous constants of the $\textbf{u}$-block rows, the constructed solution satisfies all conditions of Problem~$P_{n\lambda}$ for every $C_{D} > \max\{\tilde{C}, \bar{C}, 2\max\{1, \bar{C}_{5} + \bar{C}_{6}, \bar{C}_{11} + \bar{C}_{12}\}\}$, a threshold independent of $n_{B}$, $n_{F}$, and $n_{t}$.}
\end{proof}
\section{Proof of \texorpdfstring{\revthree{Thm.}}{Thm.}~\ref{th:consis_d}}
\label{sec:app4}
\begin{proof}
	\revthree{Let $(\textbf{x}^{\infty}(t),\textbf{u}^{\infty}(t),\bm{\lambda}^{\infty}(t),\bm{\mu}^{\infty}(t),\bar{\bm{\nu}}^{\infty})$ be the limit in Assumption~\ref{ass:conv_d}.} We show that \revthree{it satisfies} the conditions in Eqs.~\eqref{eq:dyn_P}--\eqref{eq:in_P} and the condition of the dual variables in~\eqref{eq:mu_h_0}--\eqref{eq:L_u}\revthree{, and that $(\textbf{x}^{\infty}(t),\textbf{u}^{\infty}(t))$ is an optimal solution of Problem~$P$. As in the proof of Thm.~\ref{th:feas_d}, $\bar{\bm{\lambda}}_{n} \dfn \bm{\lambda}_{n}/w$ and $\bar{\bm{\mu}}_{n} \dfn \bm{\mu}_{n}/w$. Under the coupling of the orders assumed in Thm.~\ref{th:consis_d}, $\delta_{P}^{n}\pXU \rightarrow 0$ and $\delta_{D}^{n}\pXULM \rightarrow 0$ by Thm.~\ref{th:feas_d}}.

	\paragraph*{\revthree{Primal conditions and optimality.}} \rev{Step~1 of the proof of \revthree{Thm.}~\ref{th:consis}, which uses only the convergence in Assumption~\ref{ass:conv_d} and $\delta_{P}^{n}\revthree{\pXU}\rightarrow 0$, shows the satisfaction of the conditions of the primal variables.} \revthree{Moreover, the primal pair of the solution of Assumption~\ref{ass:exis_d} is an optimal solution of Problem~$P$ with the regularity required in Assumption~\ref{ass:exis}, and the tolerance $\delta_{P}^{n}\pXU$ is computed for its decomposition. Since the pairs $(\textbf{x}_{n},\textbf{u}_{n})$ are optimal for Problem~$P_{n}$ and Assumption~\ref{ass:conv_d} contains the convergence of Assumption~\ref{ass:conv}, Thm.~\ref{th:consis}, applied with this optimal solution, shows that $(\textbf{x}^{\infty}(t),\textbf{u}^{\infty}(t))$ is an optimal solution of Problem~$P$.}

	\paragraph*{\revthree{Complementary slackness.}} \begin{revblock}Fix $t^{\prime}\in[0,1]$ and, for each $n$, let $t_{k(n)}$ be a constraint node nearest $t^{\prime}$, so that $\vert t_{k(n)} - t^{\prime}\vert \leq 1/(2n_{t}) \rightarrow 0$. By the triangle inequality,
	\begin{align}
		&\Vert(\revthree{\bm{\mu}^{\infty}})^{T}(t^{\prime})\textbf{h}(\textbf{x}^{\infty}(t^{\prime}),\textbf{u}^{\infty}(t^{\prime}))\Vert \notag \\
		&\leq \Vert(\revthree{\bm{\mu}^{\infty}})^{T}(t^{\prime})\textbf{h}(\textbf{x}^{\infty}(t^{\prime}),\textbf{u}^{\infty}(t^{\prime})) - (\revthree{\bm{\mu}^{\infty}})^{T}(t_{k(n)})\textbf{h}(\textbf{x}^{\infty}(t_{k(n)}),\textbf{u}^{\infty}(t_{k(n)}))\Vert \notag \\
		&\quad+ \Vert(\revthree{\bm{\mu}^{\infty}})^{T}(t_{k(n)})\textbf{h}(\textbf{x}^{\infty}(t_{k(n)}),\textbf{u}^{\infty}(t_{k(n)})) - \bar{\bm{\mu}}_{n}^{T}(t_{k(n)})\textbf{h}(\textbf{x}_{n}(t_{k(n)}),\textbf{u}_{n}(t_{k(n)}))\Vert \notag \\
		&\quad+ \tfrac{1}{w}\Vert\bm{\mu}_{n}^{T}(t_{k(n)})\textbf{h}(\textbf{x}_{n}(t_{k(n)}),\textbf{u}_{n}(t_{k(n)}))\Vert.
	\end{align}
	As $n\rightarrow\infty$, the first term vanishes by continuity of $\revthree{\bm{\mu}^{\infty}}$, $\textbf{x}^{\infty}$, $\textbf{u}^{\infty}$, and $\textbf{h}$; the second by Assumption~\ref{ass:conv_d} (uniform convergence through the Lipschitz continuity of $\textbf{h}$ and boundedness of the sequences); the third is bounded by $\frac{n_{t}+1}{n_{t}}\delta_{D}^{\revthree{n}}\revthree{\pXULM} \leq 2\delta_{D}^{\revthree{n}}\revthree{\pXULM} \rightarrow 0$ by Eq.~\eqref{eq:in_eq_Pnl} and Thm.~\ref{th:feas_d}. Since $t^{\prime}$ was arbitrary, the left-hand side vanishes identically. The sign condition follows likewise: 
    \begin{equation}
        \revthree{\bm{\mu}^{\infty}}(t^{\prime}) = \lim_{n\rightarrow\infty}\bar{\bm{\mu}}_{n}(t_{k(n)}) \geq -\lim_{n\rightarrow\infty}2\delta_{D}^{\revthree{n}}\revthree{\pXULM}\bm{1} = \bm{0}
    \end{equation}
    by the second inequality in Eq.~\eqref{eq:in_eq_Pnl}. Thus $\textbf{x}^{\infty}(t)$, $\textbf{u}^{\infty}(t)$, and $\revthree{\bm{\mu}^{\infty}}(t)$ \end{revblock}satisfy Eq.~\eqref{eq:mu_h_0}.

	\paragraph*{\revthree{Closure conditions.}} \begin{revblock}Next, we pass the closure conditions to the limit. Evaluating Eq.~\eqref{eq:close1} and Eq.~\eqref{eq:close2} along the sequence, using Assumption~\ref{ass:conv_d} (uniform convergence, hence convergence at $t = 0, 1$), the continuity of $E_{\textbf{x}(0)}$, $E_{\textbf{x}(1)}$, $\textbf{e}_{\textbf{x}(0)}$, $\textbf{e}_{\textbf{x}(1)}$ from Assumption~\ref{ass:lip_d}, and $\delta_{D}^{\revthree{n}}\revthree{\pXULM}\rightarrow 0$ (Thm.~\ref{th:feas_d}), we obtain
	\begin{align}
		& \revthree{\bm{\lambda}^{\infty}}(0) + \big((\bar{\bm{\nu}}^{\infty})^{T}\textbf{e}_{\textbf{x}(0)}(\textbf{x}^{\infty}(0),\textbf{x}^{\infty}(1)) + E_{\textbf{x}(0)}(\textbf{x}^{\infty}(0),\textbf{x}^{\infty}(1))\big)^{T} = \bm{0}, \notag \\
		& \revthree{\bm{\lambda}^{\infty}}(1) - \big((\bar{\bm{\nu}}^{\infty})^{T}\textbf{e}_{\textbf{x}(1)}(\textbf{x}^{\infty}(0),\textbf{x}^{\infty}(1)) + E_{\textbf{x}(1)}(\textbf{x}^{\infty}(0),\textbf{x}^{\infty}(1))\big)^{T} = \bm{0},
		\label{eq:close_lim}
	\end{align}
	i.e., Eqs.~\eqref{eq:closure1} and~\eqref{eq:closure2} hold for the limit.\end{revblock}

	\paragraph*{\revthree{Adjoint equation and control stationarity.}} \revthree{We pass the stationarity conditions~\eqref{eq:stat} to the limit in four steps.}

	\revthree{(i) \emph{A Fourier mode is fixed before $n \rightarrow \infty$.}} For the stationarity conditions, fix $j \in \mathbb{N} \cup \{0\}$ and consider the Fourier rows $\partial\mathcal{L}_{n}/\partial\revthree{\textbf{a}}_{x,j}$ and $\partial\mathcal{L}_{n}/\partial\revthree{\textbf{b}}_{x,j}$ \revthree{(the latter for $j \geq 1$)}, which are present for all $n$ with $n_{F} \geq j$\revthree{, hence for all sufficiently large $n$, since $n_{F} \geq n$}; their weights $c_{j}$, $s_{j}$, $\dot{c}_{j}$, $\dot{s}_{j}$ are fixed functions with $n$-independent bounds. \revthree{We detail the row of $\textbf{a}_{x,j}$, $j \geq 1$; the rows of $\textbf{b}_{x,j}$ and $\textbf{a}_{x,0}$ are treated in step~(iii).}

	\revthree{(ii) \emph{Limiting integral identity.} Substituting $\bm{\lambda}_{n}(t_{k}) = w\bar{\bm{\lambda}}_{n}(t_{k})$ and $\bm{\mu}_{n}(t_{k}) = w\bar{\bm{\mu}}_{n}(t_{k})$, as in the proof of Thm.~\ref{th:feas_d}, the row reads}
	\begin{revthreeblock}
	\begin{align}
		\frac{\partial \mathcal{L}_{n}}{\partial \textbf{a}_{x,j}} = {}& E_{\textbf{x}(0)} + \bar{\bm{\nu}}_{n}^{T}\textbf{e}_{\textbf{x}(0)} + E_{\textbf{x}(1)} + \bar{\bm{\nu}}_{n}^{T}\textbf{e}_{\textbf{x}(1)} + w\sum_{k=0}^{n_{t}} \textbf{G}_{n}(t_{k})\,c_{j}(t_{k}) - w\sum_{k=0}^{n_{t}} \bar{\bm{\lambda}}_{n}^{T}(t_{k})\,\dot{c}_{j}(t_{k}),
		\label{eq:row_aj}
	\end{align}
	where $E_{\textbf{x}(0)}$, $\textbf{e}_{\textbf{x}(0)}$, $E_{\textbf{x}(1)}$, and $\textbf{e}_{\textbf{x}(1)}$ are evaluated at $(\textbf{x}_{n}(0),\textbf{x}_{n}(1))$, and $\textbf{G}_{n} \dfn F_{\textbf{x}}(\textbf{x}_{n},\textbf{u}_{n}) + \bar{\bm{\lambda}}_{n}^{T}\textbf{f}_{\textbf{x}}(\textbf{x}_{n},\textbf{u}_{n}) + \bar{\bm{\mu}}_{n}^{T}\textbf{h}_{\textbf{x}}(\textbf{x}_{n},\textbf{u}_{n})$.\end{revthreeblock}
	Passing $n\rightarrow\infty$ in $\Vert\partial\mathcal{L}_{n}/\partial\revthree{\textbf{a}}_{x,j}\Vert \leq \delta_{D}^{\revthree{n}}\revthree{\pXULM}$, the node sums converge to the corresponding integrals of the limit functions: the argument swaps use Assumption~\ref{ass:conv_d} and Assumption~\ref{ass:lip_d}, and the transfer uses Eq.~\eqref{eq:transfer}, whose error terms now vanish as $n_{t}\rightarrow\infty$ with no condition among the orders because the weights are fixed; the endpoint terms converge as above, and $\delta_{D}^{\revthree{n}}\revthree{\pXULM}\rightarrow 0$. \revthree{In detail, let $\textbf{G}_{\infty} \dfn F_{\textbf{x}}(\textbf{x}^{\infty},\textbf{u}^{\infty}) + (\bm{\lambda}^{\infty})^{T}\textbf{f}_{\textbf{x}}(\textbf{x}^{\infty},\textbf{u}^{\infty}) + (\bm{\mu}^{\infty})^{T}\textbf{h}_{\textbf{x}}(\textbf{x}^{\infty},\textbf{u}^{\infty})$, which is continuous on $[0,1]$. By Assumption~\ref{ass:conv_d} and the Lipschitz continuity of the gradients (Assumption~\ref{ass:lip_d}), $\max_{k}\Vert\textbf{G}_{n}(t_{k}) - \textbf{G}_{\infty}(t_{k})\Vert \rightarrow 0$ and $\max_{k}\Vert\bar{\bm{\lambda}}_{n}(t_{k}) - \bm{\lambda}^{\infty}(t_{k})\Vert \rightarrow 0$. Since $w\sum_{k=0}^{n_{t}}\vert c_{j}(t_{k})\vert \leq 1$ and $w\sum_{k=0}^{n_{t}}\vert \dot{c}_{j}(t_{k})\vert \leq 2\pi j$, replacing $\textbf{G}_{n}$ by $\textbf{G}_{\infty}$ and $\bar{\bm{\lambda}}_{n}$ by $\bm{\lambda}^{\infty}$ in the two node sums of~\eqref{eq:row_aj} changes them by at most $\max_{k}\Vert\textbf{G}_{n}(t_{k}) - \textbf{G}_{\infty}(t_{k})\Vert + 2\pi j\max_{k}\Vert\bar{\bm{\lambda}}_{n}(t_{k}) - \bm{\lambda}^{\infty}(t_{k})\Vert \rightarrow 0$; the resulting node sums involve only the fixed continuous functions $\textbf{G}_{\infty}$ and $\bm{\lambda}^{\infty}$, to which Eq.~\eqref{eq:transfer} applies with $\psi = c_{j}$ and $\psi = \dot{c}_{j}$. Letting $n\rightarrow\infty$ in~\eqref{eq:row_aj} therefore yields the limiting identity}
	\begin{revthreeblock}
	\begin{align}
		& E_{\textbf{x}(0)} + (\bar{\bm{\nu}}^{\infty})^{T}\textbf{e}_{\textbf{x}(0)} + E_{\textbf{x}(1)} + (\bar{\bm{\nu}}^{\infty})^{T}\textbf{e}_{\textbf{x}(1)} \notag \\
		&\qquad + \int_{0}^{1}\textbf{G}_{\infty}(t)\,c_{j}(t)\,dt - \int_{0}^{1}(\bm{\lambda}^{\infty})^{T}(t)\,\dot{c}_{j}(t)\,dt = \bm{0},
		\label{eq:weak_lim}
	\end{align}
	with the endpoint functions now evaluated at $(\textbf{x}^{\infty}(0),\textbf{x}^{\infty}(1))$.\end{revthreeblock}

	\revthree{(iii) \emph{Adjoint equation.}} Integrating $\int_{0}^{1}(\revthree{\bm{\lambda}^{\infty}})^{T}\dot{c}_{j}\,dt$ by parts, which is admissible since $\bm{\lambda}^{\infty}\in\mathcal{C}^{1}$ by Assumption~\ref{ass:conv_d}, and canceling the boundary groups with the closure identities above, we obtain
	\begin{align}
		\int_{0}^{1}\textbf{g}_{\textbf{x}}^{T}(t)\,c_{j}(t)\,dt = \bm{0}, \qquad \int_{0}^{1}\textbf{g}_{\textbf{x}}^{T}(t)\,s_{j}(t)\,dt = \bm{0}, \qquad \int_{0}^{1}\textbf{g}_{\textbf{x}}^{T}(t)\,dt = \bm{0},
	\end{align}
	for every $j \in \mathbb{N}$, where
	\begin{align}
		\textbf{g}_{\textbf{x}}^{T} \dfn (\revthree{\dot{\bm{\lambda}}^{\infty}})^{T} + F_{\textbf{x}}(\textbf{x}^{\infty},\textbf{u}^{\infty}) + (\revthree{\bm{\lambda}^{\infty}})^{T}\textbf{f}_{\textbf{x}}(\textbf{x}^{\infty},\textbf{u}^{\infty}) + (\revthree{\bm{\mu}^{\infty}})^{T}\textbf{h}_{\textbf{x}}(\textbf{x}^{\infty},\textbf{u}^{\infty})
	\end{align}
	is continuous on $[0,1]$, and the third identity is the $\revthree{\textbf{a}}_{x,0}$ row. \revthree{In detail, since $c_{j}(0) = c_{j}(1) = 1$, integration by parts gives $\int_{0}^{1}(\bm{\lambda}^{\infty})^{T}\dot{c}_{j}\,dt = (\bm{\lambda}^{\infty})^{T}(1) - (\bm{\lambda}^{\infty})^{T}(0) - \int_{0}^{1}(\dot{\bm{\lambda}}^{\infty})^{T}c_{j}\,dt$. Substituted into~\eqref{eq:weak_lim}, the endpoint terms group into $\big[(\bm{\lambda}^{\infty})^{T}(0) + E_{\textbf{x}(0)} + (\bar{\bm{\nu}}^{\infty})^{T}\textbf{e}_{\textbf{x}(0)}\big] - \big[(\bm{\lambda}^{\infty})^{T}(1) - E_{\textbf{x}(1)} - (\bar{\bm{\nu}}^{\infty})^{T}\textbf{e}_{\textbf{x}(1)}\big]$, whose two brackets vanish by the closure identities~\eqref{eq:close_lim}; since $\textbf{g}_{\textbf{x}}^{T} = (\dot{\bm{\lambda}}^{\infty})^{T} + \textbf{G}_{\infty}$, what remains is the first identity. The rows of $\textbf{b}_{x,j}$ carry no endpoint terms, since $s_{j}(0) = s_{j}(1) = 0$, and give the second. The row of $\textbf{a}_{x,0}$, whose weight is the constant $1/2$, has no derivative term, and the closure identities, together with $(\bm{\lambda}^{\infty})^{T}(1) - (\bm{\lambda}^{\infty})^{T}(0) = \int_{0}^{1}(\dot{\bm{\lambda}}^{\infty})^{T}dt$, give the third.}

	\revthree{(iv) \emph{Pointwise conditions.}} The identical passage on the $\textbf{u}$-block Fourier rows\revthree{, which contain neither endpoint nor derivative terms,} yields 
    \begin{equation}
        \int_{0}^{1}\textbf{g}_{\textbf{u}}^{T}c_{j}\,dt = \int_{0}^{1}\textbf{g}_{\textbf{u}}^{T}s_{j}\,dt = \int_{0}^{1}\textbf{g}_{\textbf{u}}^{T}\,dt = \bm{0},
    \end{equation} 
    with 
    \begin{equation}
        \textbf{g}_{\textbf{u}}^{T} \dfn F_{\textbf{u}}(\textbf{x}^{\infty},\textbf{u}^{\infty}) + (\revthree{\bm{\lambda}^{\infty}})^{T}\textbf{f}_{\textbf{u}}(\textbf{x}^{\infty},\textbf{u}^{\infty}) + (\revthree{\bm{\mu}^{\infty}})^{T}\textbf{h}_{\textbf{u}}(\textbf{x}^{\infty},\textbf{u}^{\infty})
    \end{equation} 
    continuous. Every trigonometric Fourier coefficient of $\textbf{g}_{\textbf{x}}$ and $\textbf{g}_{\textbf{u}}$ therefore vanishes; since the system $\{1, \cos(2\pi j t), \sin(2\pi j t)\}_{j\in\mathbb{N}}$ is complete in $L^{2}[0,1]$\revthree{, which requires no periodicity of $\textbf{g}_{\textbf{x}}$ and $\textbf{g}_{\textbf{u}}$, both vanish almost everywhere; since, moreover,} $\textbf{g}_{\textbf{x}}$, $\textbf{g}_{\textbf{u}}$ are continuous, we conclude $\textbf{g}_{\textbf{x}} \equiv \bm{0}$ and $\textbf{g}_{\textbf{u}} \equiv \bm{0}$ on $[0,1]$, i.e., Eqs.~\eqref{eq:lamb_dot} and~\eqref{eq:L_u} hold.

	\begin{revthreeblock}\paragraph*{Role of the Bernstein rows.} Only the Fourier rows of~\eqref{eq:stat} were used above. They suffice because, for fixed $j$, their weights do not depend on $n$, have $n$-independent bounds, and, as $j$ ranges over $\mathbb{N}\cup\{0\}$, form a complete family. The Bernstein weights $b_{k,n_{B}}$, in contrast, change with $n_{B}$, so no Bernstein row can be fixed before passing to the limit. Moreover, for $0 < k < n_{B}$ the continuous counterpart of the row, $\int_{0}^{1}\textbf{g}_{\textbf{x}}^{T}b_{k,n_{B}}\,dt$, is of the order of $\int_{0}^{1}b_{k,n_{B}}\,dt = 1/(n_{B}+1)$, below the resolution of the tolerance, since $\delta_{D}^{n}\pXULM \geq C_{D}/\sqrt{n}$ by~\eqref{eq:del_D}. The Bernstein rows are satisfied along the sequence but are not needed to identify the limit. Their role is in Thm.~\ref{th:feas_d}, where every row of~\eqref{eq:stat} must hold for the constructed solution, and where their derivative weights produce the $n_{B}^{2}/n_{t}$ component of Eq.~\eqref{eq:del_D}.\end{revthreeblock}

	Therefore, we have shown that $\textbf{x}^{\infty}(t)$, $\textbf{u}^{\infty}(t)$, $\revthree{\bm{\lambda}^{\infty}}(t)$, $\revthree{\bm{\mu}^{\infty}}(t)$, and $\bar{\bm{\nu}}^{\infty}$ satisfy Eqs.~\revthree{\eqref{eq:dyn_P}--\eqref{eq:in_P} and~\eqref{eq:mu_h_0}--\eqref{eq:L_u}, i.e., they form a solution of Problem~$P_{\lambda}$ whose primal pair is optimal for Problem~$P$; hence the limit is an optimal solution of Problem~$P_{\lambda}$}.
\end{proof}

\bibliography{Bib1}
\end{document}